\documentclass{article} %
\usepackage{iclr2026_conference_arxiv,times}

\usepackage{amsmath,amsfonts,bm}

\def\eqref#1{equation~\ref{#1}}

\def\1{\bm{1}}

\def\va{{\bm{a}}}
\def\vb{{\bm{b}}}
\def\vc{{\bm{c}}}

\def\vf{{\bm{f}}}

\def\vh{{\bm{h}}}

\def\vm{{\bm{m}}}

\def\vp{{\bm{p}}}
\def\vq{{\bm{q}}}

\def\vu{{\bm{u}}}

\def\vx{{\bm{x}}}
\def\vy{{\bm{y}}}

\def\mA{{\bm{A}}}
\def\mB{{\bm{B}}}

\def\mD{{\bm{D}}}

\def\mI{{\bm{I}}}

\def\mL{{\bm{L}}}
\def\mM{{\bm{M}}}

\def\mS{{\bm{S}}}

\def\mU{{\bm{U}}}

\def\mX{{\bm{X}}}

\DeclareMathAlphabet{\mathsfit}{\encodingdefault}{\sfdefault}{m}{sl}
\SetMathAlphabet{\mathsfit}{bold}{\encodingdefault}{\sfdefault}{bx}{n}

\newcommand{\E}{\mathbb{E}}

\newcommand{\Var}{\mathrm{Var}}

\usepackage{hyperref}
\usepackage{cleveref}
\usepackage{url}
\usepackage{wrapfig}  %
\usepackage{xcolor}
\usepackage[dvipsnames]{xcolor}
\usepackage{marginnote}
\usepackage{todonotes}
\usepackage{bbm}
\usepackage[linesnumbered,ruled,vlined]{algorithm2e}
\usepackage{algpseudocode}
\usepackage{mathtools}
\usepackage{amsmath}
\usepackage{amsthm}
\usepackage{dsfont}
\usepackage{hyperref}
\usepackage{enumitem}
\usepackage{booktabs} %
\usepackage{subcaption}
\usepackage{microtype}
\usepackage{float}  %
\usepackage[table]{xcolor}
\usepackage{threeparttable}

\usepackage[T1]{fontenc}
\usepackage[utf8]{inputenc}

\setlist[itemize]{noitemsep, topsep=0pt, parsep=0pt, partopsep=0pt}
\setlist[enumerate]{noitemsep, topsep=0pt, parsep=0pt, partopsep=0pt}

\newcommand{\poly}{\textrm{poly}}
\newcommand{\vol}{\textrm{vol}}
\newtheorem{theorem}{Theorem}[section]
\newtheorem{lemma}{Lemma}[section]
\newtheorem{defn}{Definition}[section]
\newtheorem{fact}{Fact}[section]
\newtheorem{corollary}{Corollary}[section]
\newtheorem{remark}{Remark}[section]

\crefname{defn}{Definition}{Definitions}
\crefname{fact}{Fact}{Facts}

\allowdisplaybreaks

\title{Sublinear Spectral Clustering Oracle with Little Memory}

\author{Ranran Shen$^{1}$, Xiaoyi Zhu$^{2}$, Pan Peng$^{1*}$, Zengfeng Huang$^{23}$\thanks{Corresponding authors.}\\
$^{1}$School of Computer Science and Technology, University of Science and Technology of China, Hefei, China\\
$^{2}$School of Data Science, Fudan University, Shanghai, China\\
$^{3}$Shanghai Innovation Institute, Shanghai, China\\
\texttt{ranranshen@mail.ustc.edu.cn,zhuxy22@m.fudan.edu.cn,}\\
\texttt{ppeng@ustc.edu.cn, huangzf@fudan.edu.cn}}

\iclrfinalcopy %
\begin{document}

\maketitle

\begin{abstract}
     We study the problem of designing \emph{sublinear spectral clustering oracles} for well-clusterable graphs. Such an oracle is an algorithm that, given query access to the adjacency list of a graph $G$, first constructs a compact data structure $\mathcal{D}$ that captures the clustering structure of $G$. Once built, $\mathcal{D}$ enables sublinear time responses to \textsc{WhichCluster}$(G,x)$ queries for any vertex $x$. A major limitation of existing oracles is that constructing $\mathcal{D}$ requires $\Omega(\sqrt{n})$ memory, which becomes a bottleneck for massive graphs and memory-limited settings. In this paper, we break this barrier and establish a memory-time trade-off for sublinear spectral clustering oracles. Specifically, for well-clusterable graphs, we present oracles that construct $\mathcal{D}$ using much smaller than $O(\sqrt{n})$ memory (e.g., $O(n^{0.01})$) while still answering membership queries in sublinear time. We also characterize the trade-off frontier between memory usage $S$ and query time $T$, showing, for example, that $S\cdot T=\widetilde{O}(n)$ for clusterable graphs with a logarithmic conductance gap, and we show that this trade-off is nearly optimal (up to logarithmic factors) for a natural class of approaches. Finally, to complement our theory, we validate the performance of our oracles through experiments on synthetic networks.
\end{abstract}

\vspace{-0.7em}
\section{Introduction}
\vspace{-0.3em}

A central task in graph analysis is to uncover communities, which are groups of vertices that are more densely connected internally than externally. This problem, known as \textit{graph clustering}, has long been a cornerstone of graph theory and algorithms \citep{hagen1992new, chan1993spectral, ng2001spectral, czumaj2015testing, peng2020robust}. Beyond its theoretical significance, graph clustering underlies diverse applications, ranging from community detection in networks \citep{van2013community, bedi2016community, li2024comprehensive} to bioinformatics \citep{paccanaro2006spectral} and image segmentation \citep{shi2000normalized, felzenszwalb2004efficient}.

Despite their importance, most graph clustering algorithms are impractical for large graphs, as they require reading the entire input, spending $\Omega(n)$ time, and/or building data structures of size $\Omega(n)$, where $n$ is the number of vertices. Even when only a few cluster memberships are needed, these methods still carry out full global computations, making them unsuitable for massive graphs where both time and memory (or space) matter---but memory is the primary bottleneck.

From a systems perspective, this memory bottleneck is especially pressing. Many realistic environments severely restrict available working memory: streaming models limit algorithms to a single pass with sublinear space; cloud-based platforms often impose high storage and data-transfer costs, making it infeasible to materialize the entire graph; and GPUs and TPUs offer massive compute but only modest on-chip memory relative to dataset size. In all these settings, the primary challenge is to fit a compact representation of the clustering structure into limited fast memory. Thus, developing memory-efficient clustering algorithms is not only a theoretical pursuit but also a practical necessity for analyzing trillion-edge graphs in modern computing environments.

These considerations have motivated the study of \emph{local} clustering oracles that run in sublinear time and space. Our focus is on \emph{sublinear spectral clustering oracles} \citep{peng2020robust, gluch2021spectral, shen2024sublinear}, which construct a compact data structure $\mathcal{D}$ from query access to the adjacency list of the graph. Once built, $\mathcal{D}$ enables efficient evaluation of \textsc{WhichCluster}$(G,x)$ queries, that is, determining the cluster assignment of any vertex $x$ without incurring the global $\Omega(n)$ costs. Importantly, these oracles return consistent assignments (with a fixed random seed) and closely approximate the ground-truth clustering, thereby making local access to clustering information both theoretically sound and practically useful.

Several recent works \citep{peng2020robust, gluch2021spectral, shen2024sublinear} demonstrate that such oracles are possible under planted clustering assumptions, supporting cluster membership queries in both sublinear time and sublinear space. However, all existing sublinear spectral clustering oracles require at least $\Omega(\sqrt{n})$ space. In particular, Peng \citep{peng2020robust} constructs an oracle using $\tilde{\Theta}(\sqrt{n})$ space, while both Gluch et al. \citep{gluch2021spectral} and Shen et al. \citep{shen2024sublinear} require $\Omega(n^{1-\delta})$ space for any $\delta \le \tfrac{1}{2}$, which is again at least $\sqrt{n}$. We refer to \Cref{table:compare} and \Cref{sec:related} for more details. 
For truly massive graphs, this requirement is prohibitive, as limited working memory and frequent main-memory access quickly dominate the overall cost. 
This raises the central question:
\begin{center}
\emph{Is it possible to design a spectral clustering oracle that breaks the $\Omega(\sqrt{n})$ space barrier---can we use substantially less memory while still achieving sublinear query time? If so, what kinds of trade-offs between space and query efficiency can be realized?}
\end{center}

To the best of our knowledge, the question of establishing a space--time trade-off for sublinear spectral clustering oracle has not been explicitly studied in the prior literature. %
This challenge is reminiscent of recent work on space--time trade-offs in learning, beginning with \citet{raz2017time}’s result on parity learning and later extended to tasks such as linear regression \citep{sharan2019memory} and noisy parity \citep{garg2021memory}. In the area of distribution testing, a series of works \citep{diakonikolas2019communication,berg2022memory,roy2023testing,canonne2024simpler} have established sharp space--time trade-offs for fundamental problems such as uniformity testing and closeness testing. %
Much like in these learning problems and in recent advances on distribution testing, the central question for sublinear spectral clustering is how far memory usage can be reduced without making query times impractically large.

    In this paper, we give the first sublinear spectral clustering oracles with little memory (i.e., much less than $O(\sqrt{n})$) and a trade-off between memory usage $S$ and query time $T$ satisfying $S\cdot T\approx \widetilde{O}(n)$ (for a class of well clusterable graphs). We show that this trade-off is nearly optimal (up to logarithmic factors) for a natural class of approaches. In the following, we first present some basic definitions, then state our main results, provide a high-level overview of our techniques, and finally review related work.

\subsection{Basic definitions}

We measure cluster connectivity using conductance, a widely studied metric (e.g., \citep{chiplunkar2018testing,dey2019spectral, manghiuc2021hierarchical, shen2024sublinear}).  Let $G=(V,E)$ be an undirected graph. For any vertex $v\in V$, let $d_v$ denote the degree of $v$ in $G$. For any subset $C\subseteq V$, let $\vol(C)=\sum_{v\in C}{d_v}$ denote the volume of $C$. For any two subsets $S,C\subseteq V$, let $E(S,C)$ denote the set of edges between $S$ and $C$. 

        \begin{defn}[Outer and inner conductance]
            \upshape
            \label{defn:conductance}
            For any non-empty subset $C\subseteq V$, the \textit{outer conductance} and \textit{inner conductance} of $C$ is defined to be 
            \[
              \phi_{\textup{out}}(C,V)={|E(C,V\backslash C)|}/{\vol(C)}, \quad\phi_{\textup{in}}(C)=\min_{S\subseteq C, 0<\vol(S)\le {\vol(C)}/{2}}{\phi_{\textup{out}}(S, C)}.
            \]
            Specially, the \textit{conductance} of graph $G$ {\rm is defined to be} $\phi(G) = \min\limits_{C\subseteq V, 0< \vol(C)\le {\vol(V)}/{2}}{\phi_{\rm out}(C,V)}$.
        \end{defn}

        Intuitively, inner (resp. outer) conductance captures the internal (resp. external) connectivity of a cluster. A ``good'' cluster exhibits both large inner conductance and small outer conductance. Based on the definition of conductance, we give the formal definition of the input graph which is assumed to have a planted clustering structure (see~\Cref{defn:clusterable-graph}).

        \begin{defn}[$k$-partition]
            \upshape
            \label{defn:k-partition}
            Let $G=(V,E)$ be a graph. A \textit{$k$-partition of $V$} is a collection of $k$ disjoint subsets $C_1,\dots,C_k$ such that $\bigcup_{i=1}^k{C_i}=V$.
        \end{defn}
    
        \begin{defn}[$(k,\varphi,\varepsilon)$-clusterable graph]
            \upshape
            \label{defn:clusterable-graph}
            Let $k\ge 2$ be an integer and let $\varphi\in(0,1)$ and $\varepsilon\in[0,1)$. Let $G=(V,E)$ be a graph. If there exists a $k$-partition of $V$, denoted by $C_1,\dots,C_k$, such that for all $i\in[k]$, $\phi_{\textup{in}}(C_i)\ge \varphi$, $\phi_{\textup{out}}(C_i,V)\le \varepsilon$ and for all $i,j\in[k]$, one has $\frac{|C_i|}{|C_j|}\in O(1)$, then we call $G$ is a $(k,\varphi,\varepsilon)$\emph{-clusterable graph}.
        \end{defn}

\vspace{-0.2em}
We work in the \emph{adjacency list model}, where the algorithm can query any neighbor of a specified vertex in constant time.

  \vspace{-0.2em}
    \subsection{Main results}%

\paragraph{Sublinear spectral clustering oracle} 
A key contribution of this work is a spectral clustering oracle that operates with very little memory and provides an explicit trade-off between memory and query time. Given a $(k,\varphi,\varepsilon)$-clusterable graph, the goal of a clustering oracle is to build a data structure $\mathcal{D}$ in sublinear time such that, for any vertex $x$, the oracle can answer \textsc{WhichCluster}$(G,x)$ in sublinear time. Moreover, the clustering induced by answering \textsc{WhichCluster}$(G,x)$ for all $x$ should have a small misclassification error, that is, only a small fraction of vertices are assigned to the wrong clusters compared to the ground truth.  

In what follows, we state our main theorem in the simplified setting where $\varphi=\Omega(1)$ and $d,k=O(1)$. %
While we state our results for $d$-regular graphs, they naturally extend to \textbf{$d$-bounded} graphs, i.e., graphs in which every vertex has degree at most $d$ (see~\Cref{sec:prelimi}).

\begin{theorem}[Informal main result; full statement in \Cref{thrm:main}]
\label{thrm:informal-main}
Suppose $\varphi = \Omega(1)$, $d,k = O(1)$, and $\varepsilon \leq h(d,k,\varphi)$ for some function $h$. Let $G=(V,E)$ be a $d$-regular $(k,\varphi,\varepsilon)$-clusterable graph with clusters $C_1,\dots,C_k$. %
Let $1\le M \le O\left(n^{1/2 - O(\varepsilon)}\right)$ be a trade-off parameter. Then there exists a sublinear spectral clustering oracle that:

\begin{itemize}
    \item constructs a data structure $\mathcal{D}$ using $\widetilde{O}\left(n^{O(\varepsilon)}\cdot M\right)$ bits of space,
    \item answers any \textsc{WhichCluster} query in $\widetilde{O}\left(n^{1+O(\varepsilon)}/M\right)$ time,
    \item misclassifies at most $O(\varepsilon^{1/3})|C_i|$ vertices in each cluster $C_i$, $i \in [k]$.
\end{itemize}
\end{theorem}

\vspace{-0.2em}
Note that the space $S$ used to build $\mathcal{D}$ and the query time $T$ satisfy the trade-off
$S \cdot T = \widetilde{O}\!\left(n^{1+O(\varepsilon)}\right)$. The oracle is built upon a new subroutine \textsc{EstColliProb} (Alg.~\ref{algo:bounded-space-esti-colli-prob}) for estimating the collision probability of two random walk distributions with asymptotically space--time trade-off. 
In particular, when $\varepsilon \ll 1/\log n$, this simplifies to $S \cdot T = \widetilde{O}(n)$. The theorem establishes a trade-off: larger space $S$ yields faster queries, while smaller $S$ slows them down. Unlike prior oracles that require at least $\Omega(\sqrt{n})$ space, our method operates with substantially less space, often far below $\sqrt{n}$, thereby breaking the $\sqrt{n}$ space barrier. 

We provide a more detailed comparison between our main algorithmic result (\Cref{thrm:main}) and prior work in \Cref{table:compare}. Note that there are two types of results, one with $O(\log k \cdot \varepsilon)|C_i|$ misclassification error at the cost of larger space usage and query time (e.g., \cite{gluch2021spectral}) and the other with $O(\textup{poly}(k)\cdot \varepsilon^{1/3})|C_i|$ misclassification error and slightly smaller space usage and query time (e.g., \cite{shen2024sublinear}).

\vspace{-0.5em}
\begin{table}[H]
\centering
\caption{Comparison of our results (\Cref{thrm:main}) with previous work in terms of space usage, query time and misclassification error. We use $O_\varphi$ to suppress dependence on $\varphi$ and $\widetilde{O}$ to hide all $\poly(\log n)$ factors. Here $\delta\in (0,\frac{1}{2}]$ is a constant and $1\le M\le O(\frac{n^{1/2-O(\varepsilon/\varphi^2)}}{k})$ is a trade-off parameter.%
}
\label{table:compare}
\small
\begin{tabular}{lccc}

\toprule
work & space usage & query time & misclassification error \\
\midrule
\citet{peng2020robust}
& $\widetilde{O}_\varphi(\sqrt{n}\cdot \textup{poly}(\frac{k}{\varepsilon}))$
& $\widetilde{O}_\varphi(\sqrt{n}\cdot \textup{poly}(\frac{k}{\varepsilon}))$
& $O(kn\sqrt{\varepsilon})$
\\

\citet{gluch2021spectral}
& $\widetilde{O}_\varphi(n^{1-\delta+O(\varepsilon)}\cdot \textup{poly}(\frac{k}{\varepsilon}))$
& $\widetilde{O}_\varphi(n^{\delta+O(\varepsilon)}\cdot \textup{poly}(\frac{k}{\varepsilon}))$
& $O(\log k\cdot \varepsilon)|C_i|$ $^\dagger$

\\

\textbf{our (\Cref{itm:main-case1})} 
& $\widetilde{O}_\varphi(n^{O(\varepsilon)}\cdot M\cdot \textup{poly}(\frac{k}{\varepsilon}))$
& $\widetilde{O}_\varphi(n^{1+O(\varepsilon)}\cdot \frac{1}{M}\cdot \textup{poly}(\frac{k}{\varepsilon}))$
& $O(\log k\cdot \varepsilon)|C_i|$ $^\dagger$

\\

\citet{shen2024sublinear}
& $\widetilde{O}_\varphi(n^{1-\delta+O(\varepsilon)}\cdot \textup{poly}(k))$
& $\widetilde{O}_\varphi(n^{\delta+O(\varepsilon)}\cdot \textup{poly}(k))$
& $O(\textup{poly}(k)\cdot \varepsilon^{1/3})|C_i|$  $^\dagger$

\\

\textbf{our (\Cref{itm:main-case2})} 
& $\widetilde{O}_\varphi(n^{O(\varepsilon)}\cdot M\cdot \textup{poly}(k))$
& $\widetilde{O}_\varphi(n^{1+O(\varepsilon)}\cdot \frac{1}{M}\cdot \textup{poly}(k))$
& $O(\textup{poly}(k)\cdot \varepsilon^{1/3})|C_i|$  $^\dagger$

\\
\bottomrule
\end{tabular}

\begin{tablenotes}
\footnotesize
\item $\dagger$\quad for each cluster $C_i,i\in[k].$
\end{tablenotes}

\end{table}

Previous oracles require at least $\Omega(\sqrt{n})$ space usage while our oracle operates within much less space. Moreover, we stress that, our new clustering algorithms (\Cref{itm:main-case1} and \Cref{itm:main-case2}), although they introduce a space constraint, affect only the space usage and query time; all other guarantees (e.g., the conductance gap and the misclassification error) remain unchanged.

\paragraph{Distinguishing $\mathbf{1}$-cluster vs. $\mathbf{2}$-cluster} 
As a corollary of our main result, we obtain a sublinear algorithm for distinguishing between a single-cluster expander and a graph consisting of two disjoint clusters.  
Formally, let $\varphi = \Omega(1)$ and $d = O(1)$. Consider the following promise problem: the input is a $d$-regular graph $G = (V,E)$ that is guaranteed to be in one of two cases:  
(i) $G$ is a $\varphi$-expander on $n$ vertices (i.e., $(1,\varphi,0)$-clusterable); or (ii) $G$ is the disjoint union of two identical $\varphi$-expanders, each on $n/2$ vertices (i.e., $(2,\varphi,0)$-clusterable). The goal of the $1$-cluster vs.\ $2$-cluster problem is to determine which case holds.

We address this problem with an \textsc{EstColliProb}-based algorithm, yielding the following result.

\begin{theorem}[Upper bound]
\label{thrm:distinguish}
For any trade-off parameter $1 \le M \le O(\sqrt{n})$, there exists an algorithm (Alg.~\ref{algo:bounded-space-dinstinguish}) that, with probability at least $1 - 2n^{-100}$, solves the $1$-cluster vs. $2$-cluster problem. Moreover, the algorithm:
\begin{itemize}

\item uses $\widetilde{O}(M)$ bits of space,
\item runs in $\widetilde{O}\left(\tfrac{n}{M}\right)$ time.
\end{itemize}

\end{theorem}
We complement this with a lower bound for distinguishing between the two cases when the graph can only be accessed through random walk queries%
. %

\begin{defn}[Random walk queries]
        \upshape
        For any specified starting vertex $x$, a random walk query returns the endpoint of an $O(\log n)$-step random walk starting from $x$.
        \label{defn:random-walk-queries}
    \end{defn}

\begin{theorem}[Lower bound]
\label{thrm:distinguish_lower}

Any algorithm that correctly solves the $1$-cluster vs. $2$-cluster problem with error at most $1/3$ using only random walk oracles must satisfy $S\cdot T \geq \Omega(n)$, where $S$ and $T$ denote the space complexity and time complexity of the algorithm, respectively. 
\end{theorem}

Note that a random walk query can be simulated with $O(\log n)$ adjacency-list queries, so our upper bound matches the lower bound up to $\poly(\log n)$ factors. Since the \textsc{EstColliProb}-based approach solves the $1$-cluster vs. $2$-cluster problem, our lower bound indicates that its trade-off is nearly tight. This, in turn, suggests that the space--time trade-off of our clustering oracle is essentially tight, at least for approaches based on collision probability estimation.

\subsection{Technical overview}\label{sec:tech_overview}

\paragraph{Sublinear spectral clustering oracle}

To obtain sublinear spectral clustering oracles that rely on a $\log(k)$ or $\poly(k)$ conductance gap, a key primitive is the estimation of the dot product $\langle \vf_x,\vf_y\rangle$, where $\vf_x$ is the spectral embedding of $x\in V$ (see~\Cref{defn:spectral-embedding}). Suppose there exists an algorithm that estimates such dot products using $S$ space and $T$ time. We can then design a clustering oracle %
based on this primitive, which uses $\widetilde{O}(\poly(k)\cdot S)$ space to construct a data structure $\mathcal{D}$ and answers \textsc{WhichCluster} queries in $\widetilde{O}(\poly(k)\cdot T)$ time (see \Cref{sec:Spectral clustering oracles with little memory}). Thus, the central task is to understand the space--time trade-off for dot product estimation, as it directly determines the efficiency of the resulting clustering oracle.

Indeed, the previous $\Omega(\sqrt{n})$ space bottleneck in constructing $\mathcal{D}$ arises precisely from this dot product estimation step, rather than from the clustering procedure itself. This observation motivates our technical improvements. In particular, the dot product estimation algorithm of \citet{gluch2021spectral} does not directly compute \(\langle \vf_x, \vf_y \rangle\) for arbitrary vertex pairs. Instead, it applies a sequence of transformations and shows that estimating \(\langle \vf_x, \vf_y \rangle\) can be reduced to computing the collision probability 
$
(\mM^t\mathds{1}_x)^T(\mM^t\mathds{1}_y)=\langle \mM^t\mathds{1}_x, \mM^t\mathds{1}_y\rangle$, 
where $\mM$ is the random walk transition matrix of $G$ and $\mathds{1}_s$ is the indicator vector of vertex $s$.

        Previous dot product oracle estimates $\langle \mM^t \mathds{1}_x, \mM^t \mathds{1}_y \rangle$ by performing $R \approx \sqrt{n}$ independent random walks of length $t=O(\frac{\log n}{\varphi^2})$ from each vertex $x$ and $y$, respectively. The endpoints of these walks are stored to construct empirical distributions, whose dot product is then computed. This approach requires $O(R)$ words of space and $O(R t)$ time, tightly coupling space usage with computation time. In particular, to ensure sufficient accuracy, $R$ must be at least $\Omega(\sqrt{n})$, which implies that the space usage cannot be reduced below $O(\sqrt{n})$.
    
        To reduce the memory requirement below $O(\sqrt{n})$ and achieve a more flexible trade-off between space and time, we propose a batch-based estimation strategy. The idea behind this approach is inspired by \citet{canonne2024simpler}, where a similar batching technique is used to design memory-efficient algorithms for uniformity testing under memory constraints. While the underlying technique is inspired by prior work, we are the first to apply this idea in the graph setting to rigorously analyze random walks. Specifically, we partition the total of $R$ random walks into $B = R/M$ batches. In each batch, $M$ walks of length $t$ are performed from each vertex, and only the endpoints within the batch are stored to construct empirical distributions. The batch-level dot product is computed, and the final estimate is obtained by averaging over all batches. This approach reduces the space requirement to $O(M)$ words while keeping the total number of walks. %
By choosing $M$ smaller than $O(\sqrt{n})$, we can achieve a space--time trade-off satisfies $M\cdot R\approx n$. This allows for efficient estimation of the dot product even under memory constraints (see~\Cref{sec:dot-product-oracle}).
\paragraph{Distinguishing $\mathbf{1}$-cluster vs. $\mathbf{2}$-cluster} 

The core idea of our algorithm (Alg. \ref{algo:bounded-space-dinstinguish}) for distinguishing the $1$-cluster vs. $2$-cluster is to reduce the task to detecting a spectral gap in the random walk operator. Specifically, we set $t=O(\frac{\log n}{\varphi^2})$ so that in the $1$-cluster case, the second largest eigenvalue of $\mM^t$ becomes negligibly small, while in the $2$-cluster case it remains exactly $1$. To capture this behavior within bounded space, we avoid storing $\mM^t$ explicitly and instead construct a compact surrogate matrix $\mathcal{G}$ using the batch-based strategy described above. This surrogate preserves the essential spectral information of $\mM^t$, so that the separation between the two cases is faithfully reflected in the spectrum of $\mathcal{G}$ (~\Cref{lemma:expander-bounded-space-W} in~\Cref{subsec:upper-bound}). Consequently, analyzing $\mathcal{G}$ suffices to distinguish between the $1$-cluster and $2$-cluster cases using only $O(M)$ space.

To establish the space--time lower bound, we note that analyzing the distribution of random walks of the two cases reveals a fundamental discrepancy: in the $1$-cluster case, this distribution converges to uniformity over the entire set of points; whereas in the $2$-cluster case, it decomposes into two separate uniform distributions, each concentrated over half of the points. Under a sublinear space constraint, the algorithm cannot store enough indices to reliably identify which cluster a given sample belongs to, making the two cases intrinsically hard to distinguish. We formalize this intuition via a reduction to space-bounded distribution testing, leveraging the information-theoretic framework for distribution-testing lower bounds of \citet{diakonikolas2019communication}. A key technical challenge is that random walks do not produce perfectly uniform samples, and small deviations could accumulate over multiple steps and affect the memory state. To address this, we develop an inductive coupling argument that carefully controls the deviation between the random-walk and ideal uniform distributions, ensuring the accumulated discrepancy remains negligible.

A key novelty of our approach is a new reduction that connects random-walk-based graph
clustering with space-bounded distribution testing. We construct paired hard instances and show how any random-walk algorithm for distinguishing $1$-cluster vs. $2$-cluster instances can be simulated in the distribution-testing setting (ss~\Cref{subsec:lower-bound}). %

\subsection{Related work}\label{sec:related}
\citet{peng2020robust} (see also \citep{czumaj2015testing}) provided a robust sublinear spectral clustering oracle that constructs a data structure using $O(\sqrt{n}\cdot \poly(\frac{k\log n}{\varepsilon}))$ bits of space\footnote{Although the paper does not explicitly state the space complexity, it can be directly inferred from the algorithm description.} and answers any \textsc{WhichCluster}$(G,x)$ in $O(\sqrt{n}\cdot \poly(\frac{k\log n}{\varepsilon}))$ time. This oracle relies on a $\poly(k)\log n$ conductance gap between inner and outer conductance and misclassifies at most  $O(kn\sqrt{\varepsilon})$ vertices. 
\citet{gluch2021spectral} (resp. \citet{shen2024sublinear}\footnote{\citet{shen2024sublinear} stated their result for $\delta = 1/2$. Since their algorithm relies on the dot product oracle in \citet{gluch2021spectral}, the guarantee extends naturally to any $\delta \in (0,\tfrac{1}{2}]$.}) gave a sublinear spectral clustering oracle that constructs a data structure using $O(n^{1-\delta+O(\varepsilon)}\cdot\poly(\frac{k\log n}{\varepsilon}))$ (resp. $O(n^{1-\delta+O(\varepsilon)}\cdot \poly(k\log n))$) bits of space and answers any \textsc{WhichCluster}$(G,x)$ in $O(n^{\delta+O(\varepsilon)}\cdot \poly(\frac{k\log n}{\varepsilon}))$) (resp. $O(n^{\delta+O(\varepsilon)}\cdot \poly(k\log n))$) time, where $\delta\in(0,\frac{1}{2}]$. These two oracles have different conductance gap and misclassification error. 

Recently, \citet{neumann2022sublinear} studied designing sublinear spectral clustering oracles for signed graph. \citet{kapralov2023learning} studied designing sublinear hierarchical clustering oracle for graphs exhibiting hierarchical structure.

Besides the above most directly related work on sublinear spectral clustering oracles, several other research directions are also relevant to our study.

\paragraph{Property testing} %
One line of work is property testing (i.e., \emph{testing graph clusterability}), where the goal is to quickly distinguish whether a graph can be partitioned into $k$ clusters with high inner conductance, or whether it is far from having such clustering. For example, \citet{czumaj2015testing} studied testing whether a graph admits a good cluster structure in the adjacency list query model, providing algorithms with sublinear query time. This direction was later advanced by \citet{chiplunkar2018testing}. While property testing algorithms do not provide explicit cluster assignments, they capture the feasibility of clustering in sublinear resources and thus serve as an important precursor to oracle-based approaches like ours. For example, \citet{czumaj2015testing} implicitly yields a sublinear spectral clustering oracle under a $\log n$ conductance gap. This was later extended by \citet{peng2020robust}, who developed a robust oracle capable of handling noise.
    
    \paragraph{Local graph clustering} Another line of related work is \textit{local graph clustering} \citep{andersen2006local,spielman2013local,zhu2013local,gharan2014partitioning,andersen2016almost}. The goal of this category is to identify a cluster associated with a given vertex. In this setting, the algorithm outputs a set of vertices related to the input vertex, and its running time and memory usage are bounded by the size of the output cluster, up to a weak dependence on $n$. In particular, when the graph contains $k$ clusters and $n$ vertices, the complexity can be as large as $\Omega(n/k)$.

    \paragraph{Grapah problems under limited memory}
    Recently, there has been a surge of work on understanding learning under limited memory. Graph problems inherently require substantial space and time to compute, and have attracted increasing attention. One line of
    research focuses on the semi-streaming model where the algorithm is permitted $O(n \cdot \poly(\log n))$ space. Both upper bound algorithms and lower bound results are proposed for various graph problems, including Maximal Independent Set \citep{assadi2024log} and Matching \citep{kapralov2013better}.  There is also significant work on the Massively Parallel Computation model, where machines have sublinear memory to solve graph problems \citep{behnezhad2019massively,assadi2019massively,lkacki2020walking,ghaffari2020massively,nowicki2021dynamic}.

\section{Preliminaries}
\label{sec:prelimi}

    Let $G=(V,E)$ denote an unweighted, undirected $d$-regular graph with $n$ vertices, where $V=\{1,2,\dots,n\}$. Let $i\in[n]$ denote $1\le i\le n$. %
    For a graph $G=(V,E)$, let $  \mA\in \mathbb{R}^{n \times n}$ denote the adjacency matrix of $G$, where $  \mA(i,j)=1$ if $(i,j)\in E$, and $  \mA(i,j)=0$ otherwise, $i,j\in[n]$. Let $  \mD\in \mathbb{R}^{n \times n}$ denote a diagonal matrix, where $  \mD(i,i)=d_i$, $i\in[n]$. Let $  \mL=\mD^{-1}(\mD-\mA)\mD^{-1}=\mI-\frac{\mA}{d}$ denote the normalized Laplacian matrix of $G$, where $\mI\in\mathbb{R}^{n\times n}$ is the identity matrix. For $  \mL$, we use $0=\lambda_1\le \dots\le \lambda_n\le 2$ to denote its eigenvalues and $  \vu_1,\dots,\vu_n\in\mathbb{R}^n$ to denote the corresponding eigenvectors. Without loss of generality, we assume $\{\vu_1,\dots,\vu_n\}$ forms an orthonormal basis of $\mathbb{R}^n$. Let $  \mU=(\vu_1,\dots,\vu_n)\in\mathbb{R}^{n\times n}$. Based on $\mU$, we give the definition of spectral embedding (see~\Cref{defn:spectral-embedding}). Moreover, let $  \mM=\frac{1}{2}(\mI +\frac{\mA}{d})=\mI-\frac{\mL}{2}$ denote the transition matrix of lazy random walk on $G$. That is, if the walker is currently at a vertex $x\in V$, then in the next step it stays at $x$ with probability $\frac{1}{2}$, or moves to each neighbor of $x$ with probability $\frac{1}{2d}$.

    \begin{defn}[spectral embedding]
        \upshape
        \label{defn:spectral-embedding}
        Let $G=(V,E)$ be a graph. For any vertex $x\in V$, we use $  \vf_x\in\mathbb{R}^k$ to denote the \emph{spectral embedding} of $x$, where $\vf_x=\mU_{[k]}^T\mathds{1}_x=(\vu_1(x), \dots, \vu_k(x))^T.$

    \end{defn}

    \begin{defn}[$\varphi$-expander]
        \upshape
        \label{defn:expander}
        Let $G=(V,E)$ be a graph. Let $\varphi\in (0,1)$. Let $\phi(G)$ denote the conductance of $G$ (see~\Cref{defn:conductance}). If $\phi(G)\ge \varphi$, then we call $G$ a $\varphi$\emph{-expander}. 
    \end{defn}

    Let $  \va\in \mathbb{R}^n$ denote a column vector (unless otherwise stated). For any two vectors $  \va,\vb\in \mathbb{R}^n$, we use $ \langle \va,\vb\rangle=\va^T\vb$ to denote the dot product of $  \va$ and $  \vb$. For any $x\in V$, let $\mathds{1}_x\in \mathbb{R}^n$ denote the indicator vector of $x$, where $\mathds{1}_x(i)=1$ if $i=x$ and $0$ otherwise. For a vector $\va=(\va(1),\dots,\va(n))^T$, the $p$-norm ($p\ge 1$) of $\va$ is defined to be$\lVert \va \rVert_p=(\sum_{i=1}^n{|\va(i)|^p})^{\frac{1}{p}}$.

    For any symmetric matrix $\mB\in\mathbb{R}^{n\times n}$, we use $v_i(\mB)$ to denote the $i$-th largest eigenvalue of $\mB$%
    , $\lVert \mB\rVert_F=\sqrt{\sum_{i=1}^n{\sum_{j=1}^n{\mB^2(i,j)}}}$ to denote the Frobenius norm of $\mB$, $\lVert \mB\rVert_2=\max_{\vx\in\mathbb{R}^n,\lVert\vx\rVert_2=1}{\lVert \mB\vx\rVert_2}$ to denote the spectral norm of $\mB$, and $\mB_{[i]}$ to denote the first $i$ columns of $\mB, 1\le i\le n$.

    \begin{defn}[TV distance]
        \label{defn:TV-distance}
        \upshape
        For any two probability distributions $\vp,\vq$ over $[n]$, the \emph{total variance distance} (i.e., TV distance) of $\vp,\vq$ is defined to be
        \[
            d_\textup{TV}(\vp,\vq)=\frac{1}{2}\lVert \vp-\vq\rVert_1.
        \]
    \end{defn}

    \begin{fact}
        \label{fact:norm}
        For any vector $\vp\in\mathbb{R}^n$, we have $\lVert \vp\rVert_4^2\le \lVert \vp\rVert_2^2$.
        \begin{proof}
            Let $\lVert \vp\rVert_{\infty}=\max_{i=1}^n{|\vp(i)|}$%
            . Then, we have 

            \begin{align*}
                \lVert \vp\rVert_4^2&=\sqrt{\sum_{i=1}^n {\vp^4(i)}}\le \sqrt{\sum_{i=1}^n {\vp^2(i)\cdot \lVert \vp\rVert_{\infty}^2}}=\sqrt{\lVert \vp\rVert_{\infty}^2}\sqrt{\sum_{i=1}^n{\vp^2(i)}}\le \sqrt{\sum_{i=1}^n{\vp^2(i)}}\sqrt{\sum_{i=1}^n{\vp^2(i)}}=\lVert \vp\rVert_2^2.
            \end{align*}

        \end{proof}
    \end{fact}

    \paragraph{From $d$-bounded graphs to $d$-regular graphs}
    
    Although we state our results for $d$-regular graphs, they extend naturally to $d$-bounded graphs, i.e., graphs in which every vertex has degree at most $d$. The extension is straightforward: for a $d$-bounded graph $G^\prime=(V,E^\prime)$, for every $x\in V$, we can add $d-d_x$ self-loops with weight $\frac{1}{2}$ to $x$ to get a $d$-regular graph $G=(V,E)$. Note that the lazy random walk on $G$ is equivalent to the random walk on $G^\prime$, with the random walk satisfying that if the walker is currently at $x\in V$, then in the next step it stays at $x$ with probability $1-\frac{d_x}{2d}$, or moves to each neighbor of $x$ with probability $\frac{1}{2d}$.

\section{Dot product oracle with little memory}
\label{sec:dot-product-oracle}

As discussed in the technique overview, the main bottleneck in constructing sublinear spectral clustering oracles lies in dot product estimation of $\langle \vf_x,\vf_y\rangle$, whose space--time trade-off directly determines the overall efficiency. In this section, we present our batch-based dot product oracle for estimating $\langle \vf_x,\vf_y\rangle$ in small space and analyze its performance. The following theorem states the performance guarantees of our oracle.

        \begin{theorem}
            \label{thrm:bounded-space-dot-product-oracle}
            Let $k\ge 2$ be an integer. Let $\varepsilon,\varphi\in(0,1)$ with $\frac{\varepsilon}{\varphi^2}\le \frac{1}{10^5}$. Let $G=(V,E)$ be a $d$-regular and $(k,\varphi,\varepsilon)$-clusterable graph. 
            Let $\frac{1}{n^5}<\xi<1$. Let $1\le M_\textup{init},M_\textup{query}\le O(\frac{n^{1/2-20\varepsilon/\varphi^2}}{k})$. Then, with probability at least $1-2n^{-100}$, \textsc{InitOracle}$(G,k,\xi,M_\textup{init})$ (Alg. \ref{algo:bounded-space-initialize-oracle}) computes a sublinear space matrix $\Psi$ of size $n^{O(\varepsilon/\varphi^2)}\cdot \log^2n\cdot (\frac{k}{\xi})^{O(1)}$, such that the following property is satisfied:

            for every pair of vertices $x,y\in V$, \textsc{QueryDot}$(G,x,y,\xi,\Psi,M_\textup{query})$ (Alg. \ref{algo:bounded-space-query-dot}) computes an output value $\langle \vf_x,\vf_y\rangle_{\textup{apx}}$ such that with probability at least $1-6n^{-100}$:
            \[
                |\langle \vf_x,\vf_y\rangle_{\textup{apx}}-\langle \vf_x,\vf_y\rangle|\le \frac{\xi}{n}.
            \]
            Moreover, let $S_\textup{init},T_\textup{init}$ be the space and time costs of \textsc{InitOracle}$(G,k,\xi,M_\textup{init})$ (Alg.\ref{algo:bounded-space-initialize-oracle}), and let $S_\textup{query},T_\textup{query}$ be those of a single \textsc{QueryDot}$(G,x,y,\xi,\Psi,M_\textup{query})$ query (Alg.\ref{algo:bounded-space-query-dot}). %
            Then we have
            
    \begin{itemize}
\item $S_\textup{init}=(\frac{k}{\xi})^{O(1)}\cdot n^{O(\varepsilon/\varphi^2)}\cdot M_\textup{init}\cdot \log^4 n$, $\quad T_\textup{init}=(\frac{k}{\xi})^{O(1)}\cdot n^{1+O(\varepsilon/\varphi^2)}\cdot \frac{ \log^4 n}{M_\textup{init}}\cdot \frac{1}{\varphi^2}$,

\item $S_\textup{query}=(\frac{k}{\xi})^{O(1)}\cdot n^{O(\varepsilon/\varphi^2)}\cdot M_\textup{query}\cdot \log^3 n$, $\quad T_\textup{query}=(\frac{k}{\xi})^{O(1)}\cdot n^{1+O(\varepsilon/\varphi^2)}\cdot \frac{\log^3 n}{M_\textup{query}}\cdot \frac{1}{\varphi^2}$.

            \end{itemize}
        \end{theorem}

Note that to ensure that \textsc{InitOracle}$(G,k,\xi,M_\textup{init})$ (Alg. \ref{algo:bounded-space-initialize-oracle}) and \textsc{QueryDot}$(G,x,y,\xi,\Psi,M_\textup{query})$ (Alg. \ref{algo:bounded-space-query-dot}) run in sublinear time, it is required that $M_\textup{init},M_\textup{query}\ge n^{c\cdot \varepsilon/\varphi^2}$, where $c$ is a constant that is larger than the constant hidden in $O(\cdot)$-term of $n^{1+O(\varepsilon/\varphi^2)}$ in both $T_\textup{init}$ and $T_\textup{query}$. 

For initializing the dot product oracle, the previous dot product oracle in \cite{gluch2021spectral} requires at least $\widetilde{\Omega}(\sqrt{n})$ bits of space, whereas our proposed oracle can perform accurate estimation using at most  $\widetilde{O}(\sqrt{n})$ bits of space, thus breaking the $\sqrt{n}$ barrier.

\subsection{The dot product oracle}

Algorithm~\ref{algo:bounded-space-esti-rw-dot} estimates the collision probability (i.e., $\langle \mM^t\mathds{1}_x,\mM^t\mathds{1}_x\rangle$) of the random walk distributions from two given vertices within a bounded space $\widetilde{O}(M)$. This bounded-space guarantee is achieved through our batch technique, and we are the first to apply this idea in the graph setting for analyzing random walks. The formal guarantee of Alg.~\ref{algo:bounded-space-esti-rw-dot} is stated in~\Cref{lemma:bounded-space-z}.

\begin{algorithm}[H]%
        \DontPrintSemicolon %
        \caption{\textsc{EstRWDot}$(G,R,t,M,x,y)$}%
        \label{algo:bounded-space-esti-rw-dot}

        $Z\coloneqq0,B\coloneqq\frac{R}{M}$\textcolor{gray}{\Comment{$B$: number of batch}}\;
        
        \For{$b=1$ to $B$}{
            Run $M$ independent random walks of length $t$ starting from $x$ (resp. from $y$)\;

            Define $\widehat{\vp}_x(i)$ (resp. $\widehat{\vp}_y(i)$) as the fraction of random walks from $x$ (resp. from $y$) that end at $i$\;
            
            $Z_b\coloneqq \langle \widehat{\vp}_x,\widehat{\vp}_y\rangle$\;
            $Z\coloneqq Z+Z_b$\;
        }
        $Z\coloneqq \frac{Z}{B}$\;
        return $Z$\;
\end{algorithm}

Algorithm~\ref{algo:bounded-space-esti-colli-prob} computes an estimate of the Gram matrix $(\mM^t \mS)^T (\mM^t \mS)$ corresponding to the random walk distributions from a set $S$ of vertices, where $\mS\in\mathbb{R}^{n\times |S|}$ is a matrix whose $i$-th column is an indicator vector $\mathds{1}_{v}$ for $v\in S$, while operating within a bounded space $\widetilde{O}(M\cdot |S|^2)$. The formal guarantee of Alg.~\ref{algo:bounded-space-esti-colli-prob} are stated in~\Cref{lemma:bounded-space-G}.

\begin{algorithm}[H]%
        \DontPrintSemicolon %
        \caption{\textsc{EstColliProb}$(G,R,t,M, I_S)$}%
        \label{algo:bounded-space-esti-colli-prob}
        $s\coloneqq |I_S|=|\{s_1,\dots,s_s\}|$\;
        \For{$l=1$ to $O(\log n)$}{
            \For{$i=1$ to $s$}{
                \For{$j=i$ to $s$}{
                    $\mathcal{G}_l(j,i)\coloneqq\mathcal{G}_l(i,j)\coloneqq$ \textsc{EstRWDot}$(G,R,t,M,s_i,s_j)$\;
                }
            }
        }
        Let $\mathcal{G}$ be a matrix obtained by taking the entrywise median of $\mathcal{G}_l$'s\textcolor{gray}{\Comment{$\mathcal{G}\in\mathbb{R}^{s\times s}$} is symmetric}\;
        return $\mathcal{G}$\;
    \end{algorithm}

Algorithm~\ref{algo:bounded-space-initialize-oracle} initializes the dot product oracle by constructing a compact matrix $\Psi$ within approximately bounded space $\widetilde{O}(M)$. Then Algorithm~\ref{algo:bounded-space-query-dot} leverages $\Psi$ to estimate $\langle \vf_x,\vf_y \rangle$ while still operating under the same bounded space. The formal guarantees of these two procedures are stated in~\Cref{thrm:bounded-space-dot-product-oracle}.

\begin{algorithm}[H]%
        \DontPrintSemicolon %
        \caption{\textsc{InitOracle}($G,k,\xi,M_\textup{init}$)}%
        \label{algo:bounded-space-initialize-oracle}
        $t\coloneqq\frac{20\log n}{\varphi^2}$\;
        $R_{\textup{init}}\coloneqq \Theta(\frac{n^{1+920\varepsilon/\varphi^2}}{M_\textup{init}}\cdot \frac{k^{14}}{\xi^2})$ \;
        $s\coloneqq O(n^{480\cdot \varepsilon/\varphi^2}\cdot \log n\cdot k^8/\xi^2)$\;
        Let $I_S=\{s_1,\dots,s_s\}$ be the multiset of $s$ indices chosen i.u.r. from $V=\{1,\dots,n\}$\;
        $\mathcal{G}\coloneqq$ \textsc{EstColliProb}$(G,R_{\textup{init}},t,M_\textup{init}, I_S)$\;
        Let $\frac{n}{s}\cdot \mathcal{G}:=\widehat{W}\widehat{\Sigma}\widehat{W}^T$ be the eigendecomposition of $\frac{n}{s}\cdot \mathcal{G}$\;
        
        \If{$\widehat{\Sigma}^{-1}$ exists}{
            $\Psi:=\frac{n}{s}\cdot \widehat{W}_{[k]}\widehat{\Sigma}_{[k]}^{-2}\widehat{W}^T_{[k]}$\textcolor{gray}{\Comment{$\Psi\in\mathbb{R}^{s\times s}$}}\;
            return $\Psi$\;
        }
    \end{algorithm}

 \begin{algorithm}[H]%
        \DontPrintSemicolon %
        \caption{\textsc{QueryDot}($G,x,y,\xi$,$\Psi,M_\textup{query}$)}%
        \label{algo:bounded-space-query-dot}
        $t\coloneqq\frac{20\log n}{\varphi^2}$\;
        $R_{\textup{query}}\coloneqq \Theta(\frac{n^{1+440\varepsilon/\varphi^2}}{M_{\textup{query}}}\cdot\frac{k^6}{\xi^2})$ \;

        \For{$l=1$ to $O(\log n)$}{
            \For{$i=1$ to $s$}{
                $  \vx_l(i)\coloneqq$\textsc{EstRWDot}$(G,R_{\textup{query}},t,M_\textup{query},x,s_i)$\;
                $  \vy_l(i)\coloneqq$\textsc{EstRWDot}$(G,R_{\textup{query}},t,M_\textup{query},y,s_i)$\;
            }
        }

        Let $\bm{\alpha}_x$ (resp. $\bm{\alpha}_y$) be a vector obtained by taking entrywise median of $  \vx_l$'s (resp. $  \vy_l$'s)\textcolor{gray}{\Comment{$\bm{\alpha}_x,\bm{\alpha}_y\in\mathbb{R}^s$}}\;
        
        return $\langle \vf_x,\vf_y\rangle_{\textup{apx}}=\bm{\alpha}_x^T\Psi\bm{\alpha}_y$\;
    \end{algorithm}

\subsection{Analysis of the dot product oracle}

To prove~\Cref{thrm:bounded-space-dot-product-oracle}, we begin by analyzing $Z_b$ defined in Alg.~\ref{algo:bounded-space-esti-rw-dot}. The following lemma shows that $Z_b$ is an unbiased estimator of $\langle \mM^t\mathds{1}_x,\mM^t\mathds{1}_x \rangle$ and quantifies its variance.
    
    \begin{lemma}
        \label{lemma:bounded-space-z-b-E-V}
        Let $G=(V,E)$ be a graph. Let $R,t,M$ be integers, where $1\le M\le R$. Let $x,y\in V$ be two vertices. Let $\mM$ be the random walk transition matrix of $G$. Let $Z_b$ ($1\le b\le \frac{R}{M}$) be the random variable defined in \textsc{EstRWDot}$(G,R,t,M,x,y)$ (see line $5$ of Alg. \ref{algo:bounded-space-esti-rw-dot}). Then, we have
        \begin{align*}
            \E[Z_b]&=\langle \mM^t\mathds{1}_x,\mM^t\mathds{1}_y\rangle,\\
            \Var[Z_b]&\le\frac{1}{M^2}\lVert \mM^t\mathds{1}_x\rVert_2\cdot \lVert \mM^t\mathds{1}_y\rVert_2+\frac{1}{M}\left(\lVert \mM^t\mathds{1}_x\rVert_2 \cdot\lVert \mM^t\mathds{1}_y\rVert_2^2+\lVert \mM^t\mathds{1}_x\rVert_2^2 \cdot\lVert \mM^t\mathds{1}_y\rVert_2\right).
        \end{align*}

        \begin{proof}
            Run $M$ random walks of length $t$ from $x$ (resp. from $y$). Let $\vc_x(i)$ (resp. $\vc_y(i)$) denote the number of random walks from $x$ (resp. from $y$) that end at vertex $i$. It's clear that we have $\widehat{\vp}_x(i)=\frac{\vc_x(i)}{M}$ and $\widehat{\vp}_y(i)=\frac{\vc_y(i)}{M}$ (see line $4$ of Alg. \ref{algo:bounded-space-esti-rw-dot}). Let $\vp_x=\mM^t\mathds{1}_x$ (resp. $\vp_y=\mM^t\mathds{1}_y$) be the probability distribution of a length $t$ random walk starting from $x$ (resp. from $y$). Note that $\vc_x(i)\sim \mathrm{Binomial}(M, \vp_x(i))$ and $\vc_y(i)\sim \mathrm{Binomial}(M, \vp_y(i))$. According to line $5$ of Alg. \ref{algo:bounded-space-esti-rw-dot}, we have $Z_b=\langle\widehat{\vp}_x,\widehat{\vp}_y\rangle$. Therefore, about $\E[Z_b]$, we have
            
            \begin{align*}
                \E[Z_b]&=\langle \widehat{\vp}_x,\widehat{\vp}_y\rangle\\
                &=\E\left[\sum_{i=1}^n{\widehat{\vp}_x(i)\widehat{\vp}_y(i)}\right]\\
                &=\frac{1}{M^2}\cdot \sum_{i=1}^n{\E[\vc_x(i)\vc_y(i)]}\\
                &=\frac{1}{M^2}\cdot \sum_{i=1}^n{\E[\vc_x(i)]\E[\vc_y(i)]}\\
                &=\frac{1}{M^2}\cdot \sum_{i=1}^n{M\vp_x(i)M\vp_y(i)}\\
                &=\sum_{i=1}^n{\vp_x(i)\vp_y(i)}\\
                &=\langle \vp_x,\vp_y\rangle=\langle \mM^t\mathds{1}_x,\mM^t\mathds{1}_y\rangle.
            \end{align*}

            About $\Var[Z_b]$, since $\Var[Z_b]=\E[Z_b^2]-(\E[Z_b])^2$, it suffices to calculate $\E[Z_b^2]$ to get $\Var[Z_b]$.

            \begin{align*}
                \E[Z_b^2]&=\E\left[\langle \widehat{\vp}_x,\widehat{\vp}_y\rangle^2\right]\\
                &=\E\left[\left(\sum_{i=1}^n{\widehat{\vp}_x(i)\widehat{\vp}_y(i)}\right)^2\right]\\
                &=\E\left[\sum_{i=1}^n{\sum_{j=1}^n{\widehat{\vp}_x(i)\widehat{\vp}_y(i)\widehat{\vp}_x(j)\widehat{\vp}_y(j)}}\right]\\
                &=\frac{1}{M^4}\sum_{i=1}^n{\sum_{j=1}^n{\E\left[\vc_x(i)\vc_y(i)\vc_x(j)\vc_y(j)\right]}} \\
                &=\frac{1}{M^4}\sum_{i=1}^n{\sum_{j=1}^n{\E\left[\vc_x(i)\vc_x(j)\right]\cdot\E\left[\vc_y(i)\vc_y(j)\right]}} \\
                &=\frac{1}{M^4}\sum_{i=1}^n{\E\left[\vc_x^2(i)\right]\cdot \E\left[\vc_y^2(i)\right]}+\frac{1}{M^4}\sum_{i=1}^n{\sum_{j=1,j\ne i}^n{\E\left[\vc_x(i)\vc_x(j)\right]\cdot\E\left[\vc_y(i)\vc_y(j)\right]}}.
            \end{align*}
            
            For convenience, we use $A_1$ to denote $\frac{1}{M^4}\sum_{i=1}^n{\E\left[\vc_x^2(i)\right]\cdot \E\left[\vc_y^2(i)\right]}$ and $A_2$ to denote $\frac{1}{M^4}\sum_{i=1}^n{\sum_{j=1,j\ne i}^n{\E\left[\vc_x(i)\vc_x(j)\right]\cdot\E\left[\vc_y(i)\vc_y(j)\right]}}$.
            
            Since $\vc_x(i)\sim \mathrm{Binomial}(M, \vp_x(i))$, we have $\E[\vc_x(i)]=M\vp_x(i)$ and $\E[\vc_x^2(i)]=\Var[\vc_x(i)]+(\E[\vc_x(i)])^2=M\vp_x(i)(1-\vp_x(i))+M^2\vp_x^2(i)=M[\vp_x(i)+(M-1)\vp_x^2(i)]$. Therefore, we have
            \begin{align*}
                A_1&=\frac{1}{M^4}\sum_{i=1}^n{\E\left[\vc_x^2(i)\right]\cdot \E\left[\vc_y^2(i)\right]}\\
                &=\frac{1}{M^4}\sum_{i=1}^n{M\left[\vp_x(i)+(M-1)\vp_x^2(i)\right]\cdot M\left[\vp_y(i)+(M-1)\vp_y^2(i)\right]}\\
                &=\frac{1}{M^2}\sum_{i=1}^n{\vp_x(i)\vp_y(i)+(M-1)\left(\vp_x\vp_y^2(i)+\vp_x^2(i)\vp_y(i)\right)+(M-1)^2\vp_x^2(i)\vp_y^2(i)}\\
                &=\frac{1}{M^2}\langle \vp_x,\vp_y\rangle+\frac{M-1}{M^2}\left(\langle \vp_x,\vp_y^2\rangle+\langle \vp_x^2,\vp_y\rangle\right)+\frac{(M-1)^2}{M^2}\langle \vp_x^2,\vp_y^2\rangle,
            \end{align*}
    
            where with a slight abuse of notation, we use $\langle p_x,p_y^2\rangle$ to denote $\sum_{i=1}^n{p_x(i)p_y^2(i)}$, and we use $\langle p_x^2,p_y^2\rangle$ to denote $\sum_{i=1}^n{p_x^2(i)p_y^2(i)}$.
    
            To calculate $A_2$, we need to calculate $\E[\vc_x(i)\vc_x(j)]$ where $i\ne j$. We define $X_a^i$ %
            as follows:
            \[
            X_a^i = 
            \begin{cases}
            1, & \text{The $a$-th random walk from $x$ ends at $i$} \\
            0, & \text{otherwise}
            \end{cases}.
            \]

            So we have $\E[\vc_x(i)\vc_x(j)]=\E\left[\sum_{a=1}^M{X_a^i}\sum_{a=1}^M{X_a^j}\right]=\sum_{a=1}^M{\sum_{b=1}^M{\E[X_a^iX_b^j]}}$. For all $a=b$ and $i\ne j$, we have $\E[X_a^iX_b^j=0]$, since for a single random walk, it cannot ends at $i$ and $j$ the same time. For all $a\ne b$ and $i\ne j$, we have $\E[X_a^iX_b^j]=\vp_x(i)\vp_x(j)$. So we can get $\E[\vc_x(i)\vc_x(j)]=M(M-1)\vp_x(i)\vp_x(j)$. By the same augment, we get that for all $i\ne j$, $\E[\vc_y(i)\vc_y(j)]=M(M-1)\vp_y(i)\vp_y(j)$. Therefore,
            \begin{align*}
                A_2&=\frac{1}{M^4}\sum_{i=1}^n{\sum_{j=1,j\ne i}^n{\E\left[\vc_x(i)\vc_x(j)\right]\cdot\E\left[\vc_y(i)\vc_y(j)\right]}}\\
                &=\frac{1}{M^4}\sum_{i=1}^n{\sum_{j=1,j\ne i}^n{M(M-1)\vp_x(i)\vp_x(j)\cdot M(M-1)\vp_y(i)\vp_y(j)}}\\
                &=\frac{(M-1)^2}{M^2}\sum_{i=1}^m{\sum_{j=1,j\ne i}^n{\vp_x(i)\vp_y(i)\cdot \vp_x(j)\vp_y(j)}}\\
                &=\frac{(M-1)^2}{M^2}\left(\sum_{i=1}^n{\sum_{j=1}^n{\vp_x(i)\vp_y(i)\cdot \vp_x(j)\vp_y(j)}}-\sum_{i=1}^n{\vp_x^2(i)\vp_y^2(i)}\right)\\
                &=\frac{(M-1)^2}{M^2}\left(\sum_{i=1}^n{\vp_x(i)\vp_y(i)}\sum_{j=1}^n{\vp_x(j)\vp_y(j)}-\langle \vp_x^2,\vp_y^2\rangle\right)\\
                &=\frac{(M-1)^2}{M^2}\left(\langle \vp_x,\vp_y\rangle^2-\langle \vp_x^2,\vp_y^2\rangle\right).
            \end{align*}

            Put them together, we get
            \begin{align*}
                \E[Z_b^2]&=A_1+A_2\\
                &=\frac{1}{M^2}\langle \vp_x,\vp_y\rangle+\frac{M-1}{M^2}\left(\langle \vp_x,\vp_y^2\rangle+\langle \vp_x^2,\vp_y\rangle\right)+\frac{(M-1)^2}{M^2}\langle \vp_x^2,\vp_y^2\rangle\\
                &+\frac{(M-1)^2}{M^2}\left(\langle \vp_x,\vp_y\rangle^2-\langle \vp_x^2,\vp_y^2\rangle\right)\\
                &=\frac{1}{M^2}\langle \vp_x,\vp_y\rangle+\frac{M-1}{M^2}\left(\langle \vp_x,\vp_y^2\rangle+\langle \vp_x^2,\vp_y\rangle\right)+\frac{(M-1)^2}{M^2}\langle \vp_x,\vp_y\rangle^2.
            \end{align*}
            
            Therefore, we have
            \begin{align*}
                \Var[Z_b]&=\E[Z_b^2]-(\E[Z_b])^2\\
                &=\frac{1}{M^2}\langle \vp_x,\vp_y\rangle+\frac{M-1}{M^2}\left(\langle \vp_x,\vp_y^2\rangle+\langle \vp_x^2,\vp_y\rangle\right)+\frac{(M-1)^2}{M^2}\langle \vp_x,\vp_y\rangle^2-\langle \vp_x,\vp_y\rangle^2\\
                &=\frac{1}{M^2}\langle \vp_x,\vp_y\rangle+\frac{M-1}{M^2}\left(\langle \vp_x,\vp_y^2\rangle+\langle \vp_x^2,\vp_y\rangle\right)+\frac{1-2M}{M^2}\langle \vp_x,\vp_y\rangle^2\\
                &\le \frac{1}{M^2}\langle \vp_x,\vp_y\rangle+\frac{1}{M}\left(\langle \vp_x,\vp_y^2\rangle+\langle \vp_x^2,\vp_y\rangle\right)\\
                &=\frac{1}{M^2}\sum_{i=1}^n{\vp_x(i)\vp_y(i)}+\frac{1}{M}\left(\sum_{i=1}^n{\vp_x(i)\vp_y^2(i)}+\sum_{i=1}^n{\vp_x^2(i)\vp_y(i)}\right)\\
                &\le \frac{1}{M^2}\lVert \vp_x\rVert_2\cdot \lVert 
                \vp_y\rVert_2+\frac{1}{M}\left(\lVert \vp_x\rVert_2 \cdot\lVert\vp_y\rVert_4^2+\lVert \vp_x\rVert_4^2 \cdot\lVert \vp_y\rVert_2\right) %
                \\
                &\le \frac{1}{M^2}\lVert \vp_x\rVert_2\cdot \lVert \vp_y\rVert_2+\frac{1}{M}\left(\lVert \vp_x\rVert_2 \cdot\lVert \vp_y\rVert_2^2+\lVert \vp_x\rVert_2^2 \cdot\lVert \vp_y\rVert_2\right),%
            \end{align*}
            where the second-to-last inequality uses the Cauchy–Schwarz inequality and the last one follows from~\Cref{fact:norm}.
        \end{proof}
    \end{lemma}

    Building on Lemma~\ref{lemma:bounded-space-z-b-E-V}, we now consider the estimator $Z$ obtained by averaging $B=R/M$ independent copies of $Z_b$. The following lemma shows that $Z$ remains an unbiased estimator with variance reduced by a factor of $B=R/M$.
    
    \begin{lemma}
        \label{lemma:bounded-space-z-E-V}
        Let $G=(V,E)$ be a graph. Let $R,t,M$ be integers, where $1\le M\le R$. Let $x,y\in V$ be two vertices. Let $\mM$ be the random walk transition matrix of $G$. Let $Z$ be the output of \textsc{EstRWDot}$(G,R,t,M,x,y)$ (Alg. \ref{algo:bounded-space-esti-rw-dot}). Then, we have

        \begin{align*}
            \E[Z]&=\langle \mM^t\mathds{1}_x,\mM^t\mathds{1}_y\rangle,\\
            \Var[Z]&\le \frac{1}{R}\left[\frac{1}{M}\lVert \mM^t\mathds{1}_x\rVert_2\cdot \lVert \mM^t\mathds{1}_y\rVert_2+\left(\lVert \mM^t\mathds{1}_x\rVert_2 \cdot\lVert \mM^t\mathds{1}_y\rVert_2^2+\lVert \mM^t\mathds{1}_x\rVert_2^2 \cdot\lVert \mM^t\mathds{1}_y\rVert_2\right)\right].
        \end{align*}

        \begin{proof}
            According to Alg. \ref{algo:bounded-space-esti-rw-dot}, we know that $Z=\frac{1}{B}\sum_{b=1}^B{Z_b}$, where $B=\frac{R}{M}$. Therefore, using~\Cref{lemma:bounded-space-z-b-E-V}, we have $\E[Z]=\frac{1}{B}\sum_{b=1}^B\E[Z_b]=\langle\mM^t\mathds{1}_x,\mM^t\mathds{1}_y\rangle$ and 
            \begin{align*}
                \Var[Z]&=\frac{1}{B^2}\sum_{b=1}^B{\Var[Z_b]}\\
                &=\frac{1}{B}\Var[Z_b]\\
                &=\frac{M}{R}\Var[Z_b]\\
                &\le \frac{M}{R}\left[\frac{1}{M^2}\lVert \mM^t\mathds{1}_x\rVert_2\cdot \lVert \mM^t\mathds{1}_y\rVert_2+\frac{1}{M}\left(\lVert\mM^t\mathds{1}_x\rVert_2 \cdot\lVert \mM^t\mathds{1}_y\rVert_2^2+\lVert \mM^t\mathds{1}_x\rVert_2^2 \cdot\lVert \mM^t\mathds{1}_y\rVert_2\right)\right]\\
                &=\frac{1}{R}\left[\frac{1}{M}\lVert \mM^t\mathds{1}_x\rVert_2\cdot \lVert \mM^t\mathds{1}_y\rVert_2+\left(\lVert\mM^t\mathds{1}_x\rVert_2 \cdot\lVert \mM^t\mathds{1}_y\rVert_2^2+\lVert \mM^t\mathds{1}_x\rVert_2^2 \cdot\lVert \mM^t\mathds{1}_y\rVert_2\right)\right].
            \end{align*}
        \end{proof}
    \end{lemma}

\Cref{lemma:bounded-space-z} shows that, with suitable input parameters, \textsc{EstRWDot}$(G,R,t,M,x,y)$ (Alg. \ref{algo:bounded-space-esti-rw-dot}) approximates the dot product of the random walk distributions from any two vertices $x,y\in V$ within an error of $\sigma_\textup{err}$.

\begin{lemma}
        \label{lemma:bounded-space-z}

        Let $k\ge 2$ be an integer and $\varphi,\varepsilon\in(0,1)$. Let $G=(V,E)$ be a $d$-regular and $(k,\varphi,\varepsilon)$-clusterable graph. Let $\mM$ be the random walk transition matrix of $G$. %
        Let $Z$ be the output of \textsc{EstRWDot}$(G,R,t,M,x,y)$ (Alg. \ref{algo:bounded-space-esti-rw-dot}). Let $\sigma_\textup{err}>0$. Let $c>1$ be a large enough constant. For any $t\ge \frac{20\log n}{\varphi^2}$ and any $x,y\in V$, if $R\ge  \frac{c\cdot k^2n^{-1+40\varepsilon/\varphi^2}}{\sigma_\textup{err}^2M}$ and $1\le M\le O(\frac{n^{1/2-20\varepsilon/\varphi^2}}{k})$, then with probability at least $0.99$, we have
        \[
            |Z-\langle \mM^t\mathds{1}_x,\mM^t\mathds{1}_y\rangle|\le \sigma_\textup{err}.
        \]

        Moreover, \textsc{EstRWDot}$(G,R,t,M,x,y)$ runs in $O(Rt)$ time and uses $O(M\cdot \log n)$ bits of space.
    \end{lemma}

    \begin{remark}
        \label{remark:median-trick}
        The success probability of~\Cref{lemma:bounded-space-z} can be boosted up to $1-n^{-100}$ using median trick, i.e., by taking the median of $O(\log n)$ independent runs.
    \end{remark}

    To prove~\Cref{lemma:bounded-space-z}, we need the following lemma in \citet{gluch2021spectral}.
    
    \begin{lemma}[Lemma 22 in \citet{gluch2021spectral}]
        \label{lemma:gluch_0}
        Let $k\ge 2$ be an integer and $\varphi,\varepsilon\in(0,1)$. Let $G=(V,E)$ be a $d$-regular and $(k,\varphi,\varepsilon)$-clusterable graph. Let $\mM$ be the random walk transition matrix of $G$. For any $t\ge \frac{20\log n}{\varphi^2}$ and any $x\in V$ we have
        \[
            \lVert \mM^t\mathds{1}_x\rVert_2\le O(k\cdot n^{-1/2+(20\varepsilon/\varphi^2)}).
        \]

    \end{lemma}
    
    Now we are ready to prove~\Cref{lemma:bounded-space-z}.
    \begin{proof} [Proof of~\Cref{lemma:bounded-space-z}.]

        \textbf{Correctness.} By~\Cref{lemma:bounded-space-z-E-V} and~\Cref{lemma:gluch_0}, we can get that 
        
        \begin{align*}
            \Var[Z]&\le \frac{1}{R}\left[\frac{1}{M}\lVert \mM^t\mathds{1}_x\rVert_2\cdot \lVert \mM^t\mathds{1}_y\rVert_2+\left(\lVert \mM^t\mathds{1}_x\rVert_2 \cdot\lVert \mM^t\mathds{1}_y\rVert_2^2+\lVert \mM^t\mathds{1}_x\rVert_2^2 \cdot\lVert \mM^t\mathds{1}_y\rVert_2\right)\right]\\
            &=\frac{1}{R}\left(\frac{O(k^2\cdot n^{-1+40\varepsilon/\varphi^2})}{M}+O(k^3\cdot n^{-3/2+60\varepsilon/\varphi^2})\right).
        \end{align*}

        Using Chebyshev's inequality, we have

        \begin{align*}
            \Pr[|Z-\langle \mM^t\mathds{1}_x,\mM^t\mathds{1}_y\rangle|\ge \sigma_\textup{err}]&=\Pr[|Z-\E[Z]|\ge \sigma_\textup{err}]\\
            &\le \frac{\Var[Z]}{\sigma_\textup{err}^2}\\
            &\le \frac{1}{\sigma_\textup{err}^2}\cdot \frac{1}{R}\left(\frac{O(k^2\cdot n^{-1+40\varepsilon/\varphi^2})}{M}+O(k^3\cdot n^{-3/2+60\varepsilon/\varphi^2})\right)\\
            &\le \frac{1}{\sigma_\textup{err}^2}\cdot \frac{1}{R}\cdot O\left(\frac{k^2\cdot n^{-1+40\varepsilon/\varphi^2}}{M}\right) %
            \\
            &\le \frac{1}{100},
        \end{align*}
        where the second-to-last inequality holds by $M\le O\left(\frac{n^{1/2-20\varepsilon/\varphi^2}}{k}\right)$. And the last inequality holds by our choice of 
        \[
            R\ge \frac{c\cdot k^2n^{-1+40\varepsilon/\varphi^2}}{\sigma_\textup{err}^2M},
        \]
        where $c$ is a large enough constant that cancels the constant hidden in $O\left(\frac{k^2\cdot n^{-1+40\varepsilon/\varphi^2}}{M}\right)$.

        \textbf{Runtime and space.} 
        Algorithm \textsc{EstRWDot}$(G,R,t,M,x,y)$ (Alg. \ref{algo:bounded-space-esti-rw-dot}) performs $B=\frac{R}{M}$ bathches (i.e., $B=\frac{R}{M}$ iterations of the for-loop). In each batch, it runs $M$ random walks of length $t$, which requires $O(Mt)$ time and $O(M)$ words of space to store the $O(M)$ endpoints of the walks. Computing the dot product of two probability distributions takes $O(M)$ time, since each distribution has at most $M$ nonzero entries. Therefore, the runtime and space per batch are $O(Mt+M)=O(Mt)$ time and $O(M)$ words, respectively. Moreover, the space used within each batch can be reused across batches. Consequently, the overall runtime and space complexity of \textsc{EstRWDot}$(G,R,t,M,x,y)$ (Alg. \ref{algo:bounded-space-esti-rw-dot}) are $B\cdot O(Mt)=\frac{R}{M}\cdot O(Mt)=O(Rt)$ and $O(M)$ words (i.e., $O(M\cdot \log n)$ bits of space, since each endpoint can be stored in $\log n$ bits), respectively.
    \end{proof}

    \Cref{lemma:bounded-space-G} states that, under appropriate input parameters, the output $\mathcal{G}$ of our algorithm \textsc{EstColliProb} $(G,R,t,M,I_S)$ (Alg. \ref{algo:bounded-space-esti-colli-prob}) is close to $(\mM^t\mS)^T(\mM^t\mS)$ in spectral norm, where $(\mM^t\mS)^T(\mM^t\mS)$ is the Gram matrix of the random walk distributions from vertices in the sample set.
    
    \begin{lemma}
        \label{lemma:bounded-space-G}
        Let $k\ge 2$ be an integer and $\varphi,\varepsilon\in(0,1)$. Let $G=(V,E)$ be a $d$-regular and $(k,\varphi,\varepsilon)$-clusterable graph. Let $\mM$ be the random walk transition matrix of $G$. Let $I_S=\{s_1,\dots,s_s\}$ be a multiset of $s$ indices chosen from $\{1,\dots,n\}$. Let $\mS\in\mathbb{R}^{n\times s}$ be the matrix whose $i$-th column equals $\mathds{1}_{s_i}$. Let $\mathcal{G}\in \mathbb{R}^{s\times s}$ be the output of \textsc{EstColliProb} $(G,R,t,M,I_S)$ (Alg. \ref{algo:bounded-space-esti-colli-prob}). Let $\sigma_\textup{err}>0$. Let $c>1$ be a large enough constant. For any $t\ge \frac{20\log n}{\varphi^2}$, if $R\ge \frac{c\cdot k^2n^{-1+40\varepsilon/\varphi^2}}{\sigma_\textup{err}^2M}$ and $1\le M\le O\left(\frac{n^{1/2-20\varepsilon/\varphi^2}}{k}\right)$, then with probability at least $1-n^{-100}$, we have
        \[
            \lVert \mathcal{G}-(\mM^t\mS)^T(\mM^t\mS)\rVert_2\le s\cdot\sigma_\textup{err}.
        \]

        Moreover, \textsc{EstColliProb} $(G,R,t,M,I_S)$ runs in $O(Rt\cdot \log n\cdot s^2)$ time and uses $O(M\cdot \log^2 n\cdot s^2)$ bits of space.

        \begin{proof}
            \textbf{Correctness.} 
            Note that in line $5$ of Alg. \ref{algo:bounded-space-esti-colli-prob}, we get $\mathcal{G}_l(i,j)\coloneqq$\textsc{EstRWDot}$(G,R,t,M,s_i,s_j)$ (Alg. \ref{algo:bounded-space-esti-rw-dot}). Since $t\ge \frac{20\log n}{\varphi^2}$, $R\ge \frac{c\cdot k^2n^{-1+40\varepsilon/\varphi^2}}{\sigma_\textup{err}^2M}$ and $1\le M\le O\left(\frac{n^{1/2-20\varepsilon/\varphi^2}}{k}\right)$, then by~\Cref{lemma:bounded-space-z}, with probability at least $0.99$, for all $i,j\in[s]$, we have
            \[
                |\mathcal{G}_l(i,j)-\langle \mM^t\mathds{1}_{s_i},\mM^t\mathds{1}_{s_j}\rangle|=|\mathcal{G}_l(i,j)-(\mM^t\mathds{1}_{s_i})^T(\mM^t\mathds{1}_{s_j})|\le \sigma_\textup{err}.
            \]
            Note that in line $6$ of Alg. \ref{algo:bounded-space-esti-colli-prob}, we define $\mathcal{G}$ as a matrix obtained by taking the entrywise median of $\mathcal{G}_l$'s over $O(\log n)$ runs. Thus with probability at least $1-n^{-100}$ (see~\Cref{remark:median-trick}), for all $i,j\in[s]$, we have 
            \[
                |\mathcal{G}(i,j)-(\mM^t\mathds{1}_{s_i})^T(\mM^t\mathds{1}_{s_j})|\le \sigma_\textup{err},
            \]
            which implies 
            \[
                \lVert \mathcal{G}-(\mM^t\mS)^T(\mM^t\mS)\lVert_F\le s\cdot \sigma_\textup{err}.
            \]
            Moreover, we have
            \[
                \lVert \mathcal{G}-(\mM^t\mS)^T(\mM^t\mS)\lVert_2\le\lVert \mathcal{G}-(\mM^t\mS)^T(\mM^t\mS)\lVert_F \le s\cdot \sigma_\textup{err}.
            \]

            \textbf{Runtime and space.} In Alg.~\ref{algo:bounded-space-esti-colli-prob}, Alg.~\ref{algo:bounded-space-esti-rw-dot} is called $\log n \cdot s^2$ times. Since the runtime and space of Alg.~\ref{algo:bounded-space-esti-rw-dot} are $O(Rt)$ and $O(M \log n)$ bits, respectively, the runtime and space of Alg.~\ref{algo:bounded-space-esti-colli-prob} are $O(Rt \cdot \log n \cdot s^2)$ and $O(M \cdot \log^2 n \cdot s^2)$ bits, respectively.
        \end{proof}
    \end{lemma}

    The following lemma shows that the output value $\langle \vf_x,\vf_y\rangle_\textup{apx}$ of Alg. \ref{algo:bounded-space-query-dot} is close to $(\mM^t\mathds{1}_x)^T(\mM^t\mS)\left(\frac{n}{s}\cdot \widetilde{W}_{[k]}\widetilde{\Sigma}_{[k]}^{-4}\widetilde{W}_{[k]}^T\right)(\mM^t\mS)^T(\mM^t\mathds{1}_y)$. The proof of~\Cref{lemma:bounded-space-apx} is largely analogous to that of Lemma 29 in \citet{gluch2021spectral}. We therefore defer the proof to~\Cref{appendix:section-thrm-dot-product} for completeness.

    \begin{lemma}
        \label{lemma:bounded-space-apx}
        Let $k\ge 2$ be an integer and $\varphi,\varepsilon\in(0,1)$. Let $G=(V,E)$ be a $d$-regular and $(k,\varphi,\varepsilon)$-clusterable graph. Let $\mM$ be the random walk transition matrix of $G$. Let $I_S=\{s_1,\dots,s_s\}$ be a multiset of $s$ indices chosen independently and uniformly at random form $V=\{1,\dots,n\}$. Let $\mS\in\mathbb{R}^{n\times s}$ be the matrix whose $i$-th column equals $\mathds{1}_{s_i}$. Let $\sqrt{\frac{n}{s}}\cdot \mM^t\mS=\widetilde{U}\widetilde{\Sigma}\widetilde{W}^T$ be an SVD of $\sqrt{\frac{n}{s}}\cdot \mM^t\mS$ where $\widetilde{U}\in\mathbb{R}^{n\times n},\widetilde{\Sigma}\in\mathbb{R}^{n\times n},\widetilde{W}\in\mathbb{R}^{s\times n}$. Let $\frac{1}{n^6}<\xi<1$ and $1\le M_\textup{init}\le O\left(\frac{n^{1/2-20\varepsilon/\varphi^2}}{k}\right)$. Let $t\ge \frac{20\log n}{\varphi^2}$. Let $c>1$ be a large enough constant. Let $s\ge c\cdot n^{240\varepsilon/\varphi^2}\cdot \log n\cdot k^4$. Let $\Psi$ denote the matrix constructed by \textsc{InitOracle} $(G,k,\xi,M_{\textup{init}})$ (Alg. \ref{algo:bounded-space-initialize-oracle}). 
        
        Let $x,y\in V$. Let $\langle \vf_x,\vf_y\rangle_{\textup{apx}}\in \mathbb{R}$ denote the value returned by \textsc{QueryDot} $(G,x,y,\xi,\Psi,M_{\textup{query}})$ (Alg. \ref{algo:bounded-space-query-dot}). If $\frac{\varepsilon}{\varphi^2}\le \frac{1}{10^5}$, Alg. \ref{algo:bounded-space-initialize-oracle} succeeds and $1\le M_\textup{query}\le O\left(\frac{n^{1/2-20\varepsilon/\varphi^2}}{k}\right)$, then with probability at least $1-5n^{-100}$ matrix $\widetilde{\Sigma}_{[k]}^{-4}$ exists and we have 
        \[
            \left|\langle \vf_x,\vf_y\rangle_{\textup{apx}}-(\mM^t\mathds{1}_x)^T(\mM^t\mS)\left(\frac{n}{s}\cdot \widetilde{W}_{[k]}\widetilde{\Sigma}_{[k]}^{-4}\widetilde{W}_{[k]}^T\right)(\mM^t\mS)^T(\mM^t\mathds{1}_y)\right|<\frac{\xi}{n}.
        \]
    \end{lemma}

    Having~\Cref{lemma:bounded-space-z} and~\Cref{lemma:bounded-space-apx}, to prove~\Cref{thrm:bounded-space-dot-product-oracle}, we also need the following lemma.

    \begin{lemma}[Lemma 19 in \citet{gluch2021spectral}]
        \label{lemma:gluch_3}
        Let $k\ge 2$ be an integer and $\varphi,\varepsilon\in(0,1)$. Let $G=(V,E)$ be a $d$-regular and $(k,\varphi,\varepsilon)$-clusterable graph. Let $\mM$ be the random walk transition matrix of $G$. Let $I_S=\{s_1,\dots,s_s\}$ be a multiset of $s$ indices chosen independently and uniformly at random form $V=\{1,\dots,n\}$. Let $\mS\in\mathbb{R}^{n\times s}$ be the matrix whose $i$-th column equals $\mathds{1}_{s_i}$. Let $\sqrt{\frac{n}{s}}\cdot \mM^t\mS=\widetilde{U}\widetilde{\Sigma}\widetilde{W}^T$ be an SVD of $\sqrt{\frac{n}{s}}\cdot \mM^t\mS$ where $\widetilde{U}\in\mathbb{R}^{n\times n},\widetilde{\Sigma}\in\mathbb{R}^{n\times n},\widetilde{W}\in\mathbb{R}^{s\times n}$. Let $\frac{1}{n^6}<\xi<1$ and $t\ge \frac{20\log n}{\varphi^2}$. Let $c>1$ be a large enough constant. Let $s\ge c\cdot n^{480\varepsilon/\varphi^2}\cdot \log n\cdot k^8/\xi^2$. If $\frac{\varepsilon}{\varphi^2}\le \frac{1}{10^5}$, then with probability at least $1-n^{-100}$, matrix $\widetilde{\Sigma}_{[k]}^{-4}$ exists and we have
        \[
            \left|\mathds{1}_x^T\mU_{[k]}\mU_{[k]}^T\mathds{1}_y-(\mM\mathds{1}_x)^T(\mM^t\mS)\left(\frac{n}{s}\cdot \widetilde{W}_{[k]}\widetilde{\Sigma}_{[k]}^{-4}\widetilde{W}_{[k]}^T\right)(\mM^t\mS)^T(\mM\mathds{1}_y)\right|\le \frac{\xi}{n}.
        \]
        
    \end{lemma}
    
    Now we are ready to prove~\Cref{thrm:bounded-space-dot-product-oracle}.
    \begin{proof}[Proof of~\Cref{thrm:bounded-space-dot-product-oracle}]
        \textbf{Correctness.}
        Equipped with~\Cref{lemma:bounded-space-apx}, based on the correctness proof of Theorem 2 in \citet{gluch2021spectral}, we can directly obtain the correctness. Nevertheless, for completeness, we provide a concise proof here.

        Note that in line $3$ of Alg. \ref{algo:bounded-space-initialize-oracle}, we set $s=O(n^{480\varepsilon/\varphi^2}\cdot \log n\cdot k^8/\xi^2)$, and in line $4$ of Alg. \ref{algo:bounded-space-initialize-oracle}, we sample $s$ indices independently and uniformly at random form $V=\{1,\dots,n\}$ to get $I_S=\{s_1,\dots,s_s\}$. Recall that $\mM$ is the random walk transition matrix of $G$. Let $\mS\in\mathbb{R}^{n\times s}$ be the matrix whose $i$-th column is $\mathds{1}_{s_i}$. Let $\sqrt{\frac{n}{s}}\cdot \mM^t\mS=\widetilde{U}\widetilde{\Sigma}\widetilde{W}^T$ be an SVD of $\sqrt{\frac{n}{s}}\cdot \mM^t\mS$ where $\widetilde{U}\in\mathbb{R}^{n\times n},\widetilde{\Sigma}\in\mathbb{R}^{n\times n},\widetilde{W}\in\mathbb{R}^{s\times n}$.

        Recall that for any vertex $x\in V$, we define $\vf_x=\mU_{[k]}^T\mathds{1}_x$ (see~\Cref{defn:spectral-embedding}), thus we have $\langle \vf_x,\vf_y\rangle=\vf_x^T\vf_y=(\mU_{[k]}^T\mathds{1}_x)^T\mU_{[k]}^T\mathds{1}_y=\mathds{1}_x^T\mU_{[k]}\mU_{[k]}^T\mathds{1}_y$. For convenience, let us denote $\mB=(\mM^t\mathds{1}_x)^T(\mM^t\mS)\left(\frac{n}{s}\cdot \widetilde{W}_{[k]}\widetilde{\Sigma}_{[k]}^{-4}\widetilde{W}_{[k]}^T\right)(\mM^t\mS)^T(\mM^t\mathds{1}_y)$. By trangle inequality, we have
        \begin{align*}
            |\langle \vf_x,\vf_y \rangle_\textup{apx}-\langle \vf_x,\vf_y \rangle|&=|\langle \vf_x,\vf_y \rangle_\textup{apx}-\mB+\mB-\langle \vf_x,\vf_y \rangle|\\
            &\le |\langle \vf_x,\vf_y \rangle_\textup{apx}-\mB|+|\mB-\langle \vf_x,\vf_y \rangle|\\
            &=|\langle \vf_x,\vf_y \rangle_\textup{apx}-\mB|+|\mB-\langle  \mathds{1}_x^T\mU_{[k]}\mU_{[k]}^T\mathds{1}_y\rangle|.
        \end{align*}

        Let $\xi^\prime=\frac{\xi}{2}$. Let $c^\prime$ be a constant in front of $s$ form~\Cref{lemma:bounded-space-apx}. Since $s=O(n^{480\varepsilon/\varphi^2}\cdot \log n\cdot k^8/\xi^2)\ge c^\prime \cdot n^{240\varepsilon/\varphi^2}\cdot \log n\cdot k^4$, then by~\Cref{lemma:bounded-space-apx}, with probability at least $1-5n^{-100}$, we have $|\langle \vf_x,\vf_y \rangle_\textup{apx}-\mB|\le \frac{\xi^\prime}{n}=\frac{\xi}{2n}$.
        
        Let $c$ be a constant in front of $s$ form~\Cref{lemma:gluch_3}. Since $s=O(n^{480\varepsilon/\varphi^2}\cdot \log n\cdot k^8/\xi^2)\ge c\cdot n^{480\varepsilon/\varphi^2}\cdot \log n\cdot k^8/{\xi^\prime}^2$ and $\frac{\varepsilon}{\varphi^2}\le \frac{1}{10^5}$, then by~\Cref{lemma:gluch_3}, with probability at least $1-n^{-100}$, we have $|\mB-\langle  \mathds{1}_x^T\mU_{[k]}\mU_{[k]}^T\mathds{1}_y\rangle|\le \frac{\xi^\prime}{n}=\frac{\xi}{2n}$.

        Therefore, by union bound, with probability at least $1-5n^{-100}-n^{-100}=1-6n^{-100}$ , we have $|\langle \vf_x,\vf_y \rangle_\textup{apx}-\langle \vf_x,\vf_y \rangle|\le \frac{\xi}{2n}+\frac{\xi}{2n}=\frac{\xi}{n}$.

        \textbf{Runtime and space of \textsc{InitOracle}.} Algorithm \textsc{InitOracle}($G,k,\xi,M_\textup{init}$) (Alg. \ref{algo:bounded-space-initialize-oracle}) calls \textsc{EstColliProb}$(G,R_{\textup{init}},t,M_\textup{init},I_S)$ (Alg. \ref{algo:bounded-space-esti-colli-prob}) to get $\mathcal{G}$ (see line $5$ of Alg. \ref{algo:bounded-space-initialize-oracle}). According to~\Cref{lemma:bounded-space-G}, \textsc{EstColliProb}$(G,R_{\textup{init}},t,M_\textup{init},I_S)$ runs in $O(R_\textup{init}\cdot t\cdot \log n \cdot s^2)$ time and uses $O(M_\textup{init}\cdot \log^2 n\cdot s^2)$ bits of space. 
        Then in line $7$ of \textsc{InitOracle}, it computes the SVD of matrix $\mathcal{G}$ in $s^3$ time and it uses $s^2\cdot \log n$ bits of space to store $\Psi\in \mathbb{R}^{n\times n}$. Thus overall \textsc{InitOracle} runs in $O(R_\textup{init}\cdot t\cdot \log n\cdot s^2+s^3)$ time and uses $O( M_\textup{init}\cdot \log^2 n\cdot s^2 + s^2\cdot \log n)$ bits of space. By the choice of $t\coloneqq\frac{20\log n}{\varphi^2}$,
        $R_{\textup{init}}\coloneqq \Theta(\frac{n^{1+920\varepsilon/\varphi^2}}{M_{\textup{init}}}\cdot \frac{k^{14}}{\xi^2})$ and 
        $s\coloneqq O(n^{480\cdot \varepsilon/\varphi^2}\cdot \log n\cdot k^8/\xi^2)$ as in \textsc{InitOracle}, we get that \textsc{InitOracle} runs in $T_\textup{init}=(\frac{k}{\xi})^{O(1)}\cdot n^{1+O(\varepsilon/\varphi^2)}\cdot \frac{1}{M_\textup{init}}\cdot \log^4 n\cdot \frac{1}{\varphi^2}$ time and uses $S_\textup{init}=(\frac{k}{\xi})^{O(1)}\cdot n^{O(\varepsilon/\varphi^2)}\cdot M_\textup{init}\cdot \log^4 n$ bits of space.

        \textbf{Runtime and space of \textsc{QueryDot}.} In \textsc{QueryDot} (Alg. \ref{algo:bounded-space-query-dot}), in lines $3\sim 6$, it calls \textsc{EstRWDot}$(G,R_\textup{query},t,M_\textup{query},x,s_i)$ (Alg. \ref{algo:bounded-space-esti-rw-dot}) for $O(\log n\cdot s)$ times. According to~\Cref{lemma:bounded-space-z}, \textsc{EstRWDot}$(G,R_\textup{query},t,M_\textup{query},x,s_i)$ runs in $O(R_\textup{query}\cdot t)$ time and uses $O(M_\textup{query}\cdot \log n)$ bits of space. Moreover, in line $9$ of \textsc{QueryDot}, it returns $\langle \vf_x,\vf_y\rangle_\textup{apx}=\bm\alpha_x^T\Psi\bm\alpha_y$, which can be computed in $O(s^2)$ time, since we can compute $\bm{a}=\bm\alpha_x^T\Psi$ in $s^2$ time and then we compute $\bm{a}\bm\alpha_y$ in $s^2$ time. Thus overall \textsc{QueryDot} runs in $O(\log n\cdot s\cdot R_\textup{query}\cdot t+s^2)$ time and $O(\log^2 n\cdot s\cdot M_\textup{query})$ bits of space. By the choice of $t\coloneqq\frac{20\log n}{\varphi^2}$,
        $R_{\textup{query}}\coloneqq \Theta(\frac{n^{1+440\varepsilon/\varphi^2}}{M_{\textup{query}}}\cdot \frac{k^{6}}{\xi^2})$ and 
        $s\coloneqq O(n^{480\cdot \varepsilon/\varphi^2}\cdot \log n\cdot k^8/\xi^2)$ as in \textsc{QueryDot}, we get that \textsc{QueryDot} runs in $T_\textup{query}=(\frac{k}{\xi})^{O(1)}\cdot n^{1+O(\varepsilon/\varphi^2)}\cdot \frac{1}{M_\textup{query}}\cdot \log^3 n\cdot \frac{1}{\varphi^2}$ time and uses $S_\textup{query}=(\frac{k}{\xi})^{O(1)}\cdot n^{O(\varepsilon/\varphi^2)}\cdot M_\textup{query}\cdot \log^3 n$ bits of space.

    \end{proof}

\section{Spectral clustering oracles with little memory}
\label{sec:Spectral clustering oracles with little memory}

In this section, we present and prove our main algorithmic result, stated in the theorem below. We emphasize that the resulting algorithms exhibit different trade-offs between the conductance gap ($\varphi$ vs.\ $\varepsilon$), the misclassification ratio, and the corresponding space--time bounds, depending on the clustering algorithms employed, either that of \citet{gluch2021spectral} or \citet{shen2024sublinear}.

       \begin{theorem}
            \label{thrm:main}
            
            Let $k\ge 2$ be an integer, $\varphi,\varepsilon\in (0,1)$ and $h_1(k,\varphi),h_2(k,\varepsilon)$ and $h_3(k,\varphi,\varepsilon)$ be three functions. Let $\varepsilon \ll h_1(k,\varphi)$. Let $G=(V,E)$ be a $d$-regular and $(k,\varphi,\varepsilon)$-clusterable graph with $C_1,\dots,C_k$. Let $1\le M\le O(\frac{n^{1/2-O(\varepsilon/\varphi^2)}}{k})$ be a trade-off parameter. There exists a sublinear spectral clustering oracle that, with probability at least $0.9$:%

            \begin{itemize}
                \item constructs a data structure $\mathcal{D}$ using $\widetilde{O}_\varphi(h_2(k)\cdot n^{O(\varepsilon/\varphi^2)}\cdot M)$ bits of space,
                \item answers any \textsc{WhichCluster} query using $\mathcal{D}$ in $\widetilde{O}_\varphi(h_2(k)\cdot n^{1+O(\varepsilon/\varphi^2)}\cdot \frac{1}{M})$ time\footnote{In order for the query time to be sublinear, $M$ must satisfy $M\ge n^{c\cdot \varepsilon/\varphi^2}$, where $c$ is a constant that is larger than the constant hidden in $O(\cdot)$-term of $n^{1+O(\varepsilon/\varphi^2)}$.},
                \item has $O\left(h_3(k,\varphi,\varepsilon)\right)|C_i|$ misclassification error for each $i\in[k]$,
            \end{itemize}
            where we use $O_\varphi$ to suppress dependence on $\varphi$ and $\widetilde{O}$ to hide all $\textup{poly}(\log n)$ factors and:
            \begin{enumerate}[label=\arabic*]
                \item \label{itm:main-case1} if $h_1(k,\varphi)=\frac{\varphi^3}{\log k}$, then $h_2(k,\varepsilon)=(\frac{k}{\varepsilon})^{O(1)}$ and $h_3(k,\varphi,\varepsilon)=\frac{\varepsilon}{\varphi^3}\cdot \log k$;
                
                \item \label{itm:main-case2} if $h_1(k,\varphi)=\frac{\varphi^2\cdot \gamma^3}{k^{\frac{9}{2}}\cdot\log^3 k}$, then $h_2(k)=(\frac{k}{\gamma})^{O(1)}$ and $h_3(k,\varphi,\varepsilon)=(\frac{\varepsilon}{\varphi^2})^{\frac{1}{3}}\cdot k^{\frac{3}{2}}$, where $\gamma\in (0.001,1]$ is a constant such that for all $i\in[k]$, $\gamma\frac{n}{k}\le |C_i|\le \frac{n}{\gamma k}$.
            \end{enumerate}
        \end{theorem}

    We present the proof of Item~\ref{itm:main-case2} of Theorem~\ref{thrm:main} here, while the proof of the remaining case, Item~\ref{itm:main-case1}, is deferred to~\Cref{appendix:section-log(k)-oracle}.

\Cref{itm:main-case2}, which addresses a sublinear spectral clustering oracle under a $\poly(k)$ conductance gap. Our sublinear spectral clustering oracle closely follows the construction in \citet{shen2024sublinear}, except that we substitute our new dot product oracle from \Cref{sec:dot-product-oracle} in place of theirs. 

\paragraph{High-level idea of the algorithm} 
Now we briefly outline the main idea of the oracle. \citet{shen2024sublinear} showed that for most vertices in a $(k,\varphi,\varepsilon)$-clusterable graph, if $x,y\in V$ belong to the same cluster, then $\langle \vf_x,\vf_y\rangle\approx \frac{k}{n}$, otherwise, $\langle \vf_x,\vf_y\rangle\approx 0$. Leveraging this property, we can design a clustering oracle as follows: it first samples $s=\frac{k\log k}{\gamma}$ vertices to form a set $S$, and for each pair $u,v\in S$, it computes the dot product $\langle \vf_u,\vf_v\rangle_\textup{apx}$ using our new dot product oracle. If the value is large, an edge $(u,v)$ is added to the initially empty similarity graph $H=(S,\emptyset)$. At query time, the oracle uses $H$ and its connected components to determine the cluster assignment of vertices. We provide a full description of the clustering oracle in \Cref{appendix:section-poly(k)-oracle}. 
Now we present the proof of~\Cref{itm:main-case2} in~\Cref{thrm:main} as follows.

\vspace{-0.7em}
        \begin{proof}[Proof of~\Cref{itm:main-case2} in~\Cref{thrm:main}]
            \textbf{Space and runtime.} In the preprocessing phase, \textsc{ConstructOracle}$(G,k,\varphi,\varepsilon,\gamma,M)$ (Alg. \ref{alg:shen-construct}) invokes our \textsc{InitOracle}$(G,k,\xi,M)$ (Alg. \ref{algo:bounded-space-initialize-oracle}) one time to get a matrix $\Psi$ (see line $5$ of Alg. \ref{alg:shen-construct}), then \textsc{ConstructOracle}$(G,k,\varphi,\varepsilon,\gamma,M)$ invokes our \textsc{QueryDot}$(G,u,v,\xi,\Psi,M)$ %
            $O((k^2\log^2k)/\gamma^2)$ times (see lines $6\sim 9$ of Alg. \ref{alg:shen-construct}) to get a similarity graph $H$. Therefore, \textsc{ConstructOracle}$(G,k,\varphi,\varepsilon,\gamma,M)$ %
            uses %
            $S_\textup{init}+O((k^2\log^2k)/\gamma^2)\cdot S_\textup{query}$ bits of space. Using~\Cref{thrm:bounded-space-dot-product-oracle}, we get that \textsc{ConstructOracle}$(G,k,\varphi,\varepsilon,\gamma,M)$ %
            uses $O(n^{O(\varepsilon/\varphi^2)}\cdot M\cdot \poly(\frac{k\log n}{\gamma}))$ bits of space to get matrix $\Psi$ and a similarity graph $H$.

            In the query phase, \textsc{WhichCluster}$(G,x,M)$ (Alg. \ref{alg:shen-which}) invokes \textsc{Search}$(H,\ell,x,M)$ (Alg. \ref{alg:shen-search}) one time. \textsc{Search}$(H,\ell,x,M)$ invokes our \textsc{QueryDot}$(G,u,x,\xi,\Psi,M)$ %
            $O((k\log k)/\gamma)$ 
            times (see lines $1\sim 2$ of Alg. \ref{alg:shen-search}) and relies on the similarity graph $H$ (see lines $3\sim 6$ of Alg. \ref{alg:shen-search}). Therefore, \textsc{WhichCluster}$(G,x,M)$ uses %
            $O((k\log k)/\gamma)\cdot S_\textup{query}$ 
            bits of space and runs in %
            $O((k\log k)/\gamma)\cdot T_\textup{query}$ time. Using~\Cref{thrm:bounded-space-dot-product-oracle}, we get that \textsc{WhichCluster}$(G,x,M)$ uses $O(n^{O(\varepsilon/\varphi^2)}\cdot M\cdot \poly(\frac{k\log n}{\gamma}))$ bits of space and runs in $O(n^{1+O(\varepsilon/\varphi^2)}\cdot \frac{1}{M}\cdot \poly(\frac{k\log n}{\gamma\varphi}))$ time. 

            Thus, the oracle constructs a data structure $\mathcal{D}$ (including  $\Psi$, similarity graph $H$ etc) using $O(n^{O(\varepsilon/\varphi^2)}\cdot M\cdot \poly(\frac{k\log n}{\gamma}))$ bits of space. Using $\mathcal{D}$, any \textsc{WhichCluster}$(G,x)$ query can be answered by Alg. \ref{alg:shen-which} in $O(n^{1+O(\varepsilon/\varphi^2)}\cdot \frac{1}{M}\cdot \poly(\frac{k\log n}{\gamma\varphi}))$ time.

            \textbf{Correctness.} Since the correctness guarantees (i.e., conductance gap and misclassification error) of the clustering oracle rely on the properties of the dot product oracle, and our dot product oracle satisfies the same correctness guarantees with the previous one, the correctness of the overall clustering oracle follows directly from the correctness of the clustering oracle in \citet{shen2024sublinear}.
        \end{proof}

\section{Distinguishing $1$-cluster vs.\ $2$-cluster}

In this section, we present both the upper and lower bounds for distinguishing $1$-cluster from $2$-cluster.

\subsection{Upper bound}
\label{subsec:upper-bound}

We now describe our algorithm for distinguishing $1$-cluster from $2$-cluster and then provide the analysis that establishes the upper bound.

\subsubsection{The algorithm}
Algorithm \ref{algo:bounded-space-dinstinguish} is based on estimating the second largest eigenvalue of $\mathbf{M}^t$ using a subroutine \textsc{EstColliProb} (Alg. \ref{algo:bounded-space-esti-colli-prob}) from \Cref{sec:dot-product-oracle}.%
 
    \vspace{-0.5em}
    \begin{algorithm}[H]%
 \DontPrintSemicolon %
\caption{\textsc{Distinguish}($G,M$)}%
    \label{algo:bounded-space-dinstinguish}
        $t\coloneqq\frac{20\log n}{\varphi^2}$, 
        $R\coloneqq \Theta(\frac{n}{M})$, 
        $s\coloneqq O(\log n)$\;
        Let $I_S=\{s_1,\dots,s_s\}$ be the multiset of $s$ indices chosen independently and uniformly at random from $V=\{1,\dots,n\}$\;
        $\mathcal{G}\coloneqq$ \textsc{EstColliProb}$(G,R,t,M, I_S)$\;
        Let $v_2(\frac{n}{s}\mathcal{G})$ be the second largest eigenvalue of matrix $\frac{n}{s}\mathcal{G}$\;
        \If{$\left(v_2(\frac{n}{s}\mathcal{G})\right)^2<0.6$}{
            return ``$1$-cluster''\;
        }
        return ``$2$-cluster''\;
    \end{algorithm}

The formal guarantee of this algorithm is given in Theorem~\ref{thrm:distinguish}. We now describe the main idea of the algorithm and present the proof of Theorem~\ref{thrm:distinguish}.

\begin{theorem}[Restatement of~\Cref{thrm:distinguish}]
For any trade-off parameter $1 \le M \le O(\sqrt{n})$, there exists an algorithm (Alg.~\ref{algo:bounded-space-dinstinguish}) that, with probability at least $1 - 2n^{-100}$, solves the $1$-cluster vs. $2$-cluster problem. Moreover, the algorithm:
\begin{itemize}

\item uses $\widetilde{O}(M)$ bits of space,
\item runs in $\widetilde{O}\left(\tfrac{n}{M}\right)$ time.
\end{itemize}

\end{theorem}

\subsubsection{Analysis of the upper bound}
Consider the case when the input graph $G$ is a $\varphi$-expander. By Cheeger’s inequality (\Cref{lemma:cheeger}), we get that the second smallest eigenvalue of $\mL$ satisfies $\lambda_{2}\ge \varphi^{2}/2$. Equivalently, the lazy random walk matrix $\mM=\mI-\mL/2$ has its second largest eigenvalue $v_2(\mM)\le 1-\varphi^{2}/4$. In contrast, if $G$ consists of two disjoint $\varphi$-expanders of equal size, then $\lambda_{2}=0$ and hence $v_2(\mM)=1$. Setting $t=O(\frac{\log n}{\varphi^{2}})$, we obtain that in the $1$-cluster case, the contribution of $v_2(\mM)\le n^{-10}$, while in the $2$-cluster case, $v_2(\mM)$ remains exactly $1$. Thus, $\mM^t$ exhibits a clear spectral gap between the two cases. Alg. \ref{algo:bounded-space-dinstinguish} constructs an approximation $\mathcal{G} \approx (\mM^{t}\mS)^T(\mM^{t}\mS)\in\mathbb{R}^{O(\log n)\times O(\log n)}$ within bounded space (see~\Cref{lemma:expander-bounded-space-G}), where each column of $\mM^t\mS$ corresponds to the $t$-step lazy random walk distribution starting from a vertex in the sampled set~$I_S$.

\begin{lemma}
        \label{lemma:expander-bounded-space-G}
        Let $\varphi\in(0,1)$. Let $G=(V,E)$ be either a $d$-regular $\varphi$-expander with size $n$ or the disjoint union of two identical $d$-regular $\varphi$-expander of size $n/2$. Let $\mM$ be the random walk transition matrix of $G$. Let $I_S=\{s_1,\dots,s_s\}$ be a multiset of $s$ indices chosen from $\{1,\dots,n\}$. Let $\mS\in\mathbb{R}^{n\times s}$ be the matrix whose $i$-th column equals $\mathds{1}_{s_i}$. Let $\mathcal{G}\in \mathbb{R}^{s\times s}$ be the output of \textsc{EstColliProb} $(G,R,t,M,I_S)$ (Alg. \ref{algo:bounded-space-esti-colli-prob}). Let $\sigma_\textup{err}>0$. Let $c>1$ be a large enough constant. For any $t\ge \frac{20\log n}{\varphi^2}$, if $R\ge \frac{c\cdot n^{-1}}{\sigma_\textup{err}^2M}$ and $1\le M\le O\left(n^{1/2}\right)$, then weith probability $1-n^{-100}$, we have
        \[
            \lVert \mathcal{G}-(\mM^t\mS)^T(\mM^t\mS)\rVert_2\le s\cdot\sigma_\textup{err}.
        \]

        Moreover, \textsc{EstColliProb} $(G,R,t,M,I_S)$ runs in $O(Rt\cdot \log n\cdot s^2)$ time and uses $O(M\cdot \log^2 n\cdot s^2)$ bits of space.

    \end{lemma}

The second largest eigenvalue of $\mathcal{G}$ closely reflects that of $\mM^t$, thereby preserving the above separation (see~\Cref{lemma:expander-bounded-space-W}). Moreover, since $\mathcal{G}$ is a small matrix, we can afford to perform an eigen-decomposition on it directly. Consequently, examining the spectrum of $\mathcal{G}$ suffices to distinguish between the $1$-cluster and $2$-cluster cases using $\widetilde{O}(M)$ bits of space and $\widetilde{O}(n/M)$ time. The proofs of~\Cref{lemma:expander-bounded-space-G} and~\Cref{lemma:expander-bounded-space-W} are deferred to~\Cref{appendix:sec-distinguish}.

\begin{lemma}
        \label{lemma:expander-bounded-space-W}
        Let $\varphi\in(0,1)$. Let $G=(V,E)$ be a $d$-regular graph. Let $I_S=\{s_1,\dots,s_s\}$ be a multiset of $s$ indices chosen independently and uniformly at random form $V=\{1,\dots,n\}$. Let $\mathcal{G}\in \mathbb{R}^{s\times s}$ be the output of \textsc{EstColliProb} $(G,R,t,M,I_S)$ (Alg. \ref{algo:bounded-space-esti-colli-prob}). Let $c_1>1$ be a large enough constant. For any $t\ge \frac{20\log n}{\varphi^2}$, if $R\ge \frac{c_1\cdot n}{M}$ and $1\le M\le O\left(n^{1/2}\right)$, then with probability at least $1-2\cdot n^{-100}$, 
        
        \begin{enumerate}[label=\arabic*]
            \item \label{itm:v2G-case1} if $G$ is a $\varphi$-expander of size $n$ and $s\ge 1$, then $v_2\left(\left(\frac{n}{s}\mathcal{G}\right)^2\right)=\left(v_2(\frac{n}{s}\mathcal{G})\right)^2<0.001$,

            \item \label{itm:v2G-case2} if $G$ is the disjoint union of two identical $\varphi$-expanders of size $n/2$ and $s\ge c_2\cdot \log n$, where $c_2>1$ is a large enough constant, then $v_2\left(\left(\frac{n}{s}\mathcal{G}\right)^2\right)=\left(v_2(\frac{n}{s}\mathcal{G})\right)^2>0.95$.
        \end{enumerate}
    \end{lemma}

Now we are ready to prove~\Cref{thrm:distinguish}.
    
    \begin{proof}[Proof of~\Cref{thrm:distinguish}]
        \textbf{Correctness.} By the promise in the theorem statement, the input $d$-regular graph $G=(V,E)$ is guaranteed to be either a $\varphi$-expander or the disjoint union of two identical $\varphi$-expanders, each of size $n/2$. We run algorithm \textsc{Distinguish}$(G,M)$ (Alg. \ref{algo:bounded-space-dinstinguish}) to distinguish the above two cases. Note that the choices of $t$, $s$, and $R$ are made so that all the assumptions required by~\Cref{lemma:expander-bounded-space-W} are satisfied. Therefore, by~\Cref{lemma:expander-bounded-space-W}, we get that in case (i) (when $G$ is a $\varphi$-expander), with probability at least $1-2n^{-100}$, $(v_2(\frac{n}{s}\mathcal{G}))^2<0.001<0.6$; in case (ii), with probability at least $1-2n^{-100}$, $(v_2(\frac{n}{s}\mathcal{G}))^2>0.95>0.6$. Therefore, we get that, with probability at least $1-2n^{-100}$, algorithm \textsc{Distinguish} correctly distinguishes which case holds.

        \textbf{Space and runtime.} 
        According to~\Cref{lemma:expander-bounded-space-G}, getting matrix $\mathcal{G}$ requires $O(R\cdot t\cdot \log n\cdot s^2)$ time and $O(M\cdot \log^2 n\cdot s^2)$ bits of space. Computing $(\tfrac{n}{s}\mathcal{G})^2$ requires $O(s^3)$ time and $O(s^2\cdot \log n)$ bits of space.
        Therefore, the overall runtime and space complexity are $O(R\cdot t\cdot \log n\cdot s^2 + s^3)$ and $O(M\cdot \log^2 n\cdot s^2+s^2\log n)$ bits, respectively. By setting $t=\frac{20\log n}{\varphi^2},R=\Theta(\frac{n}{M})$ and $s=O(\log n)$, we get that \textsc{Distinguish}$(G,M)$ runs in $n\cdot \frac{1}{M}\cdot \poly(\log n)\cdot \frac{1}{\varphi^2}$ time and uses $M\cdot \poly(\log n)$ bits of space.

    \end{proof}

\subsection{Lower bound}
\label{subsec:lower-bound}

\vspace{-0.5em}

In this section, we prove the lower bound for distinguishing $1$-cluster from $2$-cluster, stated in Theorem~\ref{thrm:distinguish_lower}.

\begin{theorem}[Restatement of~\Cref{thrm:distinguish_lower}]

Any algorithm that correctly solves the $1$-cluster vs. $2$-cluster problem with error at most $1/3$ using only random walk oracles must satisfy $S\cdot T \geq \Omega(n)$, where $S$ and $T$ denote the space complexity and time complexity of the algorithm, respectively. 
\end{theorem}

At a high level, the proof proceeds in two steps.
First, we establish a fundamental lower bound for distinguishing between two reference distributions: (i) a uniform distribution over all vertices and (ii) two separate uniform distributions, each supported on half of the vertex set. In the second step, we establish a reduction from distinguishing the reference distributions to distinguishing $1$-cluster from $2$-cluster.

\subsubsection{Hard Instance I}

We first consider the following Hard Instance, inspired by \citet{diakonikolas2019communication} and commonly used in uniformity testing. Note that in our construction, at each time $t$, the player is allowed to pick a $W_t\in[2n]$. The proof of~\Cref{thrm:hard-instance-1} follows closely that of Theorem 23 in \citet{diakonikolas2019communication} and is therefore deferred to~\Cref{subsec:hard-instance-1}.%

    \begin{defn}
        [Hard Instance I] 
        \upshape
        Let $X$ be a uniformly random bit. Based on $X$, the adversary chooses the distribution $p$ on $[2n]$ bins as
            follows:
    
            \begin{itemize}
            \item $X=0$ : Pick $p=U_{2 n}$, where $U_{2n}$ is the uniform distribution on $[2n]$.
            \item $X=1$ : 
            We construct two sets as follows: Pair the bins as $\{1,2\},\{3,4\}, \cdots,\{2 n-1,2 n\}$. Now on each pair $\{2 i-1,2 i\}$ pick a random $Y_i \in\{ \pm 1\}$. If $Y_i = 1$, we put bin $2 i-1$ to set $1$ and bin $2i$ to set $2$; otherwise, we put bin $2 i$ to set $1$ and bin $2i-1$ to set $2$. Each time, the player picks $W_t \in [2n]$. If $W_t$ belongs to set $1$, we have $Z_t = 1$; otherwise, $Z_t = -1$. The distribution is then

            $$
            \left(p_{2 i-1}, p_{2 i}\right)= \left(\frac{1+Y_i Z_t}{2 n}, \frac{1-Y_i Z_t}{2 n}\right).
            $$
            \end{itemize}
        
    \end{defn}

    We have the space--time tradeoff of this instance to be:

    \begin{theorem} 
        \label{thrm:hard-instance-1}
        Let $\mathcal{A}$ be an algorithm that detects the Hard Instance I with error at most $1 / 3$. The algorithm can access the samples in a single-pass streaming fashion using $M$ bits of space and $T$ samples. Furthermore, at each step, the algorithm may choose which set to sample by specifying $W_t$. We then have $T \cdot M=\Omega\left(n\right)$.
    \end{theorem}

\subsubsection{Hard Instance II}

For the $1$-cluster vs. $2$-cluster problem, we would consider the following Hard Instance. 
    
    \begin{defn}[Hard Instance II] 
        Let $X$ be a uniformly random bit. Let $\varphi \in (0,1)$ with $\varphi = \Omega(1)$, and let $d = O(1)$. Based on $X$, the adversary chooses a $d$-regular graph $G$ on $2n$ vertices as follows:
    
            \begin{itemize}
            \item $X=0$ : Pick the graph to be a $\varphi$-expander on $2n$ vertices.
            \item $X=1$ : We construct two sets as follows: Pair bins the   as $\{1,2\},\{3,4\}, \cdots,\{2 n-1,2 n\}$. Now on each pair $\{2 i-1,2 i\}$ pick a random $Y_i \in\{ \pm 1\}$. If $Y_i = 1$, we put vertex $2 i-1$ to set $1$ and vertex $2i$ to set $2$; otherwise, we put vertex $2 i$ to set $1$ and vertex $2i-1$ to set $2$. The graph is then composed of two identical $\varphi$-expanders over set $1$ and set $2$. 
            \end{itemize}
    \end{defn}
    
    We would assume that the algorithm has access to the graph only via the random walk queries (see~\Cref{defn:random-walk-queries}). We have the space--time tradeoff of this instance to be:

    \begin{theorem} [Variant of~\Cref{thrm:distinguish_lower}]
        \label{thrm:hard-instance-2}
        Let $\mathcal{A}$ be an algorithm which detects the Hard Instance II with error probability at most $1 / 3$. The algorithm can perform $T$ random walk queries using $M$ bits of space. We have $M \cdot T=\Omega\left(n\right)$.
    \end{theorem}

    To prove~\Cref{thrm:hard-instance-2}, we will use the following lemma, whose proof has been deferred to~\Cref{subsec:hard-instance-2}.

    \begin{lemma} 
        \label{lemma:mixing time}
        Assume $G=(V,E)$ is a $d$-regular $\varphi$-expander on $n$ vertices. Let $\mM$ be the lazy random walk transition matrix of $G$. Let $\mM^t\mathds{1}_x$ be the probability distribution of a random walk with length $O(\frac{\log n}{\varphi^2})$ starting from vertex $x\in V$. Let $\pi=(\frac{1}{n},\dots,\frac{1}{n})^T\in\mathbb{R}^n$ be the uniform distribution over $n$ vertices. We have that $d_{\textup{TV}}(\mM^t\mathds{1}_x,\pi)\le \frac{0.01}{n^2}$.
    \end{lemma}

    With the above results, we would show the space--time trade-off of identifying Hard Instance II.

    \begin{proof}[Proof of~\Cref{thrm:hard-instance-2}]
Assume we have an algorithm $\mathcal{A}$ that solves the Hard Instance II. We would show how it can be used to solve Hard Instance I. At each time, the algorithm would choose to make a random walk query starting from vertex $i$. We would then set $W_t$ to the Hard Instance I and get the feedback sample $s_t$. We would feed $s_t$ to the algorithm $\mathcal{A}$ and then to the next round. Finally, after $T$ rounds, we would output the results of $\mathcal{A}$. 
    
    To prove the correctness, we need to show that the total variation distance is $O(1)$ between the history generated by Hard Instance I: $\left(s_1, m_1, \ldots, s_T, m_T\right)$ and the history generated by Hard Instance II: $\left(s_1^{\prime}, m_1^{\prime}, \ldots, s_T^{\prime}, m_T^{\prime}\right)$. We would prove by math induction.

    Now for $d_{\textup{TV}}((m_t,s_t),(m_t{'},s_t^{\prime}))$, we consider any fixed $x \in [2n], m \in [M]$ that 
    
    \begin{align*}
    &\left|p(m_t = m, s_t = x) - p(m_t^{\prime} = m, s_t^{\prime} = x)\right| \\&= \Big|\sum_{(\widetilde{m},\widetilde{x})} p(m_t = m, s_t = x | m_{t-1} = \widetilde{m}, s_{t-1} = \widetilde{x} ) \cdotp(m_{t-1} = \widetilde{m}, s_{t-1} = \widetilde{x} ) 
    \\&  - \sum_{(\widetilde{m},\widetilde{x})} p(m_t^{\prime} = m, s_t^{\prime} = x|m_{t-1}^{\prime} = \widetilde{m}, s_{t-1}^{\prime} = \widetilde{x})\cdot p(m_{t-1}^{\prime} = \widetilde{m}, s_{t-1}^{\prime} = \widetilde{x})\Big|
    \\& \leq \Big|\sum_{(\widetilde{m},\widetilde{x})} p(m_t = m, s_t = x|m_{t-1} = \widetilde{m}, s_{t-1} = \widetilde{x})  
     \cdot \left(p(m_{t-1} = \widetilde{m}, s_{t-1} = \widetilde{x}) -p(m_{t-1}^{\prime} = \widetilde{m}, s_{t-1}^{\prime} = \widetilde{x})\right)\Big|
    \\& +  \Big|\sum_{(\widetilde{m},\widetilde{x})} p(m_{t-1}^{\prime} = \widetilde{m}, s_{t-1}^{\prime} = \widetilde{x})\\
    &\cdot \left(p(m_t = m, s_t = x|m_{t-1} = \widetilde{m}, s_{t-1} = \widetilde{x}) - p(m_t^{\prime} = m, s_t^{\prime} = x|m_{t-1}^{\prime} = \widetilde{m}, s_{t-1}^{\prime} = \widetilde{x})\right)\Big|.
    \end{align*}
    
    Now for the first part, we have 
    \begin{align*}
    & \sum_{(m,x)}  \Big|\sum_{(\widetilde{m},\widetilde{x})} p(m_t = m, s_t = x|m_{t-1} = \widetilde{m}, s_{t-1} = \widetilde{x})\cdot \left(p(m_{t-1} = \widetilde{m}, s_{t-1} = \widetilde{x}) -p(m_{t-1}^{\prime} = \widetilde{m}, s_{t-1}^{\prime} = \widetilde{x})\right)\Big|
    \\& \leq  \sum_{(m,x)}\sum_{(\widetilde{m},\widetilde{x})} \Big(p(m_t = m, s_t = x|m_{t-1} = \widetilde{m}, s_{t-1} = \widetilde{x})  \cdot\left| p(m_{t-1} = \widetilde{m}, s_{t-1} = \widetilde{x}) -p(m_{t-1}^{\prime} = \widetilde{m}, s_{t-1}^{\prime} = \widetilde{x})\right|\Big)
    \\&  =  \sum_{(\widetilde{m},\widetilde{x})}  \Big(\left| p(m_{t-1} = \widetilde{m}, s_{t-1} = \widetilde{x}) -p(m_{t-1}^{\prime} = \widetilde{m}, s_{t-1}^{\prime} = \widetilde{x})\right|  \cdot \sum_{(m,x)} p(m_t = m, s_t = x|m_{t-1} = \widetilde{m}, s_{t-1} = \widetilde{x})\Big)
    \\& =   \sum_{(\widetilde{m},\widetilde{x})}  \left| p(m_{t-1} = \widetilde{m}, s_{t-1} = \widetilde{x}) -p(m_{t-1}^{\prime} = \widetilde{m}, s_{t-1}^{\prime} = \widetilde{x})\right|  
    \\&= 2d_{\textup{TV}}((m_{t-1},s_{t-1}),(m_{t-1}^{\prime},s_{t-1}^{\prime})).
    \end{align*}
    
    For the second part, we notice that 
    
    \begin{align*}
    & p(m_t = m, s_t = x|m_{t-1} = \widetilde{m}, s_{t-1} = \widetilde{x}) - p(m_t^{\prime} = m, s_t^{\prime} = x|m_{t-1}^{\prime} = \widetilde{m}, s_{t-1}^{\prime} = \widetilde{x}) 
    \\&= p(m_t = m|s_t = x, m_{t-1} = \widetilde{m}, s_{t-1} = \widetilde{x})\cdot p( s_t = x|m_{t-1} = \widetilde{m}, s_{t-1} = \widetilde{x}) 
    \\& - p(m_t^{\prime} = m|s_t^{\prime} = x, m_{t-1}^{\prime} = \widetilde{m}, s_{t-1}^{\prime} = \widetilde{x}) \cdot p(s_t^{\prime} = x|m_{t-1}^{\prime} = \widetilde{m}, s_{t-1}^{\prime} = \widetilde{x}).
    \end{align*}
    
    Note that since we are using the same algorithm, when fixing $m_{t-1}$ and $s_t$, the update of $m_t$ and $m_t^{\prime}$ is the same, and thus 
    
    \begin{align*}
    & p(m_t = m, s_t = x|m_{t-1} = \widetilde{m}, s_{t-1} = \widetilde{x}) - p(m_t^{\prime} = m, s_t^{\prime} = x|m_{t-1}^{\prime} = \widetilde{m}, s_{t-1}^{\prime} = \widetilde{x}) 
    \\&= p(m_t = m|s_t = x, m_{t-1} = \widetilde{m})\cdot \left(p( s_t = x|m_{t-1} = \widetilde{m})- p(s_t^{\prime} = x|m_{t-1}^{\prime} = \widetilde{m})\right).
    \end{align*}
    Moreover, by the property of lazy random walk (\Cref{lemma:mixing time}), we should have that for any $\widetilde{m}$, 
    $$\frac{1}{2}\sum_{x} \left|p( s_t = x|m_{t-1} = \widetilde{m})- p(s_t^{\prime} = x|m_{t-1}^{\prime} = \widetilde{m})\right| \leq \frac{0.01}{n^2}.$$

    Summing over all $(m,x)$, we have the second part is bounded by 
    
    \begin{align*}
    &\sum_{(m,x)}\Big|\sum_{(\widetilde{m},\widetilde{x})} p(m_{t-1}^{\prime} = \widetilde{m}, s_{t-1}^{\prime} = \widetilde{x})\cdot p(m_t = m|s_t = x, m_{t-1} = \widetilde{m})\cdot\left(p( s_t = x|m_{t-1} = \widetilde{m})- p(s_t^{\prime} = x|m_{t-1}^{\prime} = \widetilde{m})\right)\Big|
    \\& \leq\sum_{(m,x,\widetilde{m},\widetilde{x})}p(m_{t-1}^{\prime} = \widetilde{m}, s_{t-1}^{\prime} = \widetilde{x})\cdot p(m_t = m|s_t = x, m_{t-1} = \widetilde{m})\cdot \left|p( s_t = x|m_{t-1} = \widetilde{m})- p(s_t^{\prime} = x|m_{t-1}^{\prime} = \widetilde{m})\right|
    \\& = \sum_{(x,\widetilde{m},\widetilde{x})}p(m_{t-1}^{\prime} = \widetilde{m}, s_{t-1}^{\prime} = \widetilde{x})\left|p( s_t = x|m_{t-1} = \widetilde{m})- p(s_t^{\prime} = x|m_{t-1}^{\prime} = \widetilde{m})\right|\cdot \sum_{m}p(m_t = m|s_t = x, m_{t-1} = \widetilde{m})
    \\& = \sum_{(x,\widetilde{m},\widetilde{x})}p(m_{t-1}^{\prime} = \widetilde{m}, s_{t-1}^{\prime} = \widetilde{x})\left|p( s_t = x|m_{t-1} = \widetilde{m})- p(s_t^{\prime} = x|m_{t-1}^{\prime} = \widetilde{m})\right|
    \\& = \sum_{(\widetilde{m},\widetilde{x})}p(m_{t-1}^{\prime} = \widetilde{m}, s_{t-1}^{\prime} = \widetilde{x})\cdot \sum_{x}\left|p( s_t = x|m_{t-1} = \widetilde{m})- p(s_t^{\prime} = x|m_{t-1}^{\prime} = \widetilde{m})\right|
    \\& \leq 2\times\frac{0.01}{n^2} \sum_{(\widetilde{m},\widetilde{x})}p(m_{t-1}^{\prime} = \widetilde{m}, s_{t-1}^{\prime} = \widetilde{x}) \\
    &= 2\times\frac{0.01}{n^2}.
    \end{align*}
    
    Combining the results, we have 
    \begin{align*}
    d_{\textup{TV}}((m_t,s_t),(m_t{'},s_t^{\prime})) &= \frac{1}{2}\sum_{(m,x)} \left|p(m_t = m, s_t = x) - p(m_t^{\prime} = m, s_t^{\prime} = x)\right|\\& \leq  d_{\textup{TV}}((m_{t-1},s_{t-1}),(m_{t-1}^{\prime},s_{t-1}^{\prime})) + \frac{0.01}{n^2}.
    \end{align*}
    
    Moreover, for the initial points, we have that 
    $$d_{\textup{TV}}(s_1,s_1^{\prime}) \leq \frac{0.01}{n^2}.$$
    Since $m_1, m_1^{\prime}$ are merely a function of $s_1,s_1^{\prime}$, we have that  
    $$d_{\textup{TV}}(m_1,m_1^{\prime}) \leq \frac{0.01}{n^2}.$$
    Therefore 
    $$d_{\textup{TV}}((m_1,s_1),(m_1{'},s_1^{\prime})) \leq d_{\textup{TV}}(s_1,s_1^{\prime})+d_{\textup{TV}}(m_1,m_1^{\prime}) \leq \frac{0.02}{n^2},$$
    $$d_{\textup{TV}}((m_t,s_t),(m_t{'},s_t^{\prime})) \leq  \frac{0.01(1+t)}{n^2}.$$
    
    This means that 
    
    $$d_{\textup{TV}}(m_T,m_T^{\prime}) \leq d_{\textup{TV}}((m_T,s_T),(m_T{'},s_T^{\prime})) \leq \frac{0.01(1+T)}{n^2}  \leq 0.01,$$
    where we use the fact that $T \leq O(n^2)$ since otherwise we can get the output using constant space. 
    
    Now note that the output result is only the function of $m_T$. Since the total variation distance of $m_T$ is bounded, the correctness can still be guaranteed using the uniform distribution rather than the random walk distribution. 
    \end{proof}

\section{Experiments}

    To evaluate the space--time trade-off of our sublinear spectral clustering oracles, we conducted experiments in Python on graphs generated from the stochastic block model (SBM) with parameters $n$ (num of vertices), $k$ (num of clusters), and edge probabilities $p$ (within-cluster) and $q$ (between-cluster). Experiments were run on a server with an Intel(R) Xeon(R) Platinum 8562Y processor (2.80 GHz) and 768 GB RAM. Each reported data is the average over five independent runs.

    We implemented two variants of the $\poly(k)$-conductance-gap clustering oracle\footnote{We did not experiment with the $\log(k)$-conductance-gap oracle due to its impractical runtime of $2^{\poly(k)}\cdot n^{1+O(\varepsilon)}\cdot \frac{1}{M}$ for constructing $\mathcal{D}$.}: the original oracle from \citet{shen2024sublinear}, and our memory-efficient variant that operates within a smaller space. For each, we recorded the number of words stored in each component of the data structure $\mathcal{D}$ as a proxy for space $S$, evaluated accuracy (the fraction of vertices correctly classified), the success rate (i.e., the fraction of successful runs among $5$ runs\footnote{If the available space is too limited, the construction of the similarity graph $H$ may yield either too many or too few connected components, in which case the construction of $\mathcal{D}$ fails.}). Both variants used the same number of sampled vertices, random walk length, and median-trick repetitions; differences arose only in space--time-related parameters. We instantiated this setup on an SBM graph with $n=3000$, $k=3$, $p=0.07$, and $q=0.002$, yielding clusters of $1000$ vertices each. Additional implementation details are provided in~\Cref{appendix:section-experiment}.

    \paragraph{Space efficiency} Prior sublinear spectral clustering oracles require at least $\Omega(\sqrt{n})$ space to construct data structure $\mathcal{D}$. In contrast, our clustering oracle allows constructing $\mathcal{D}$ using substantially less space, well below $\sqrt{n}$. In this section, we provide experimental evidence to validate this improvement. 
\vspace{-1em}
\begin{table}[h]
    \centering
    \caption{Comparison of space usage for clustering oracles, with 10400 words used as the baseline.}
    \label{table:space}
    \begin{tabular}{c|ccc|cccc} %
        \toprule
        \multicolumn{1}{c|}{clustering oracle} & 
        \multicolumn{3}{c|}{ours} & 
        \multicolumn{4}{c}{previous} \\
        \midrule
        
        \textcolor{gray}{space (\# of words)}  & \textcolor{gray}{$9900$} & \textcolor{gray}{$10100$} & \textcolor{gray}{$\mathbf{10400}$} %
        & \textcolor{gray}{$34840$} & \textcolor{gray}{$43888$}  & \textcolor{gray}{$44383$}  & \textcolor{gray}{$61223$}\\

        space ($\times$ baseline)  & $0.95\times$ & $0.97\times$ & \textcolor{ForestGreen}{$\mathbf{1\times}$} %
        & \textcolor{red}{$\mathbf{3.35\times}$} & $4.22\times$ & \textcolor{ForestGreen}{$\mathbf{4.27\times}$} & $5.89\times$\\

        success rate for constructing $\mathcal{D}$  & $1$ & $1$ & \textcolor{ForestGreen}{$\mathbf{1}$} %
        & \textcolor{red}{$\mathbf{0}$} & $0.6$ & \textcolor{ForestGreen}{$\mathbf{1}$}  & $1$\\
        accuracy  & $0.9833$ & $0.9900$ & \textcolor{ForestGreen}{$\mathbf{0.9907}$} %
        & \textcolor{red}{$\mathbf{0}$} & $0.9860$ & \textcolor{ForestGreen}{$\mathbf{0.9997}$} & $1.0000$\\
        \bottomrule
    \end{tabular}
\end{table}
    \Cref{table:space} demonstrate that our clustering oracle achieves high accuracy using substantially less space (10400 words as $1\times$). In contrast, the previous clustering oracle requires $4.27$ times of the baseline space to achieve comparable accuracy, and even when given $3.35$ times the baseline space, it fails to construct $\mathcal{D}$ successfully (i.e., success rate is $0$). These results confirm that our approach significantly improves space efficiency without compromising accuracy.

    \paragraph{space--time trade-off} As established in~\Cref{thrm:main}, there is a trade-off between the space $S$ required to construct $\mathcal{D}$ and the query time $T$, satisfying $S \cdot T \approx \widetilde{O}(n^{1+O(\varepsilon)})$, where $\varepsilon$ is the small constant corresponding to the outer conductance.

    To validate this experimentally, we also measured $S$ as the total number of words stored to construct $\mathcal{D}$. %
    We use the total number of random walks per \textsc{WhichCluster} query as a proxy for time $T$, since this dominates the query cost. Across all tested parameter settings, the oracle maintains high accuracy ($0.9833\sim1$), confirming the practical validity of the configurations used.

    \begin{figure}[h]
        \centering
        \includegraphics[width=0.5\linewidth]{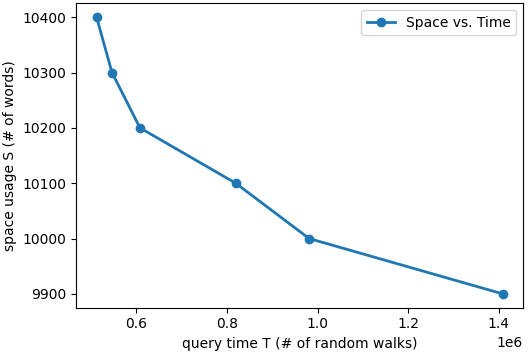}
        \caption{space--time trade-off of the sublinear spectral clustering oracle, showing $S,T$ are inversely proportional.}
        \label{fig:trade-off}
    \end{figure}

    \Cref{fig:trade-off} plots $S$ (y-axis) versus $T$ (x-axis), illustrating the space--time trade-off: memory usage decreases as query time increases, and vice versa, consistent with the theoretical bound.

\clearpage

\newpage

\section*{Acknowledgments}
Ranran Shen and Pan Peng are supported in part by NSFC Grant 62272431 and Quantum Science and Technology - National Science and Technology Major Project (Grant No. 2021ZD0302901). Xiaoyi Zhu and Zengfeng Huang are supported in part by National Natural Science Foundation of China No. 62276066.

\bibliography{iclr2026_conference}
\bibliographystyle{iclr2026_conference}

\newpage
\appendix
\begin{center}\huge\bf Appendix \end{center}

The appendix is organized as follows. 
\begin{itemize}

    \item \Cref{appendix:section-thrm-dot-product} presents the deferred proofs about the dot product oracle that operates under limited memory.

    \item \Cref{appendix:section-log(k)-oracle} provides the proof of~\Cref{itm:main-case1} of our main result (~\Cref{thrm:main}).

    \item \Cref{appendix:section-poly(k)-oracle} describes the sublinear spectral clustering oracle related to~\Cref{itm:main-case2} of our main result (~\Cref{thrm:main}).

    \item \Cref{appendix:sec-distinguish} presents the deferred proofs about the upper bound for distinguishing $1$-cluster vs. $2$-cluster problem (\Cref{thrm:distinguish}).

    \item \Cref{appendix:section-lower-bound} presents the deferred proofs about the lower bound for distinguishing $1$-cluster vs. $2$-cluster problem (\Cref{thrm:distinguish_lower}).

    \item \Cref{appendix:section-experiment} provides details on the experimental setup and parameter choices.
\end{itemize}

\section{Deferred proofs about the dot product oracle with little memory}
    \label{appendix:section-thrm-dot-product}

    Recall that we use $(\mM^t\mathds{1}_x)^T(\mM^t\mS)(\frac{n}{s}\cdot \widetilde{W}_{[k]}\widetilde{\Sigma}^{-4}_{[k]}\widetilde{W}^T_{[k]}) (\mM^t\mS)^T(\mM^t\mathds{1}_y)$ to estimate $\langle \vf_x,\vf_y\rangle$.~\Cref{lemma:bounded-space-W} states that under appropriate parameters, Alg. \ref{algo:bounded-space-initialize-oracle} outputs a matrix $\Psi = \frac{n}{s}\cdot \widehat{W}_{[k]}\widehat{\Sigma}_{[k]}^{-2}\widehat{W}_{[k]}^T$ which, with high probability, is spectrally close to $\frac{n}{s}\cdot \widetilde{W}_{[k]}\widetilde{\Sigma}_{[k]}^{-4}\widetilde{W}_{[k]}^T$. %
    The proof of~\Cref{lemma:bounded-space-W} is analogous to that of Lemma 24 in \citet{gluch2021spectral}. Nevertheless, for completeness, we provide a concise proof here.

    \begin{lemma}
        \label{lemma:bounded-space-W}
        Let $k\ge 2$ be an integer and $\varphi,\varepsilon\in(0,1)$. Let $G=(V,E)$ be a $d$-regular and $(k,\varphi,\varepsilon)$-clusterable graph. Let $\mM$ be the random walk transition matrix of $G$. Let $I_S=\{s_1,\dots,s_s\}$ be a multiset of $s$ indices chosen independently and uniformly at random form $V=\{1,\dots,n\}$. Let $\mS\in\mathbb{R}^{n\times s}$ be the matrix whose $i$-th column equals $\mathds{1}_{s_i}$. Let $\mathcal{G}\in \mathbb{R}^{s\times s}$ be the output of \textsc{EstColliProb} $(G,R,t,M,I_S)$ (Alg. \ref{algo:bounded-space-esti-colli-prob}). Let $\sqrt{\frac{n}{s}}\cdot \mM^t\mS=\widetilde{U}\widetilde{\Sigma}\widetilde{W}^T$ be an SVD of $\sqrt{\frac{n}{s}}\cdot \mM^t\mS$ where $\widetilde{U}\in\mathbb{R}^{n\times n},\widetilde{\Sigma}\in\mathbb{R}^{n\times n},\widetilde{W}\in\mathbb{R}^{s\times n}$. Let $\frac{n}{s}\cdot\mathcal{G}=\widehat{W}\widehat{\Sigma}\widehat{W}^T$ be an eigendecomposition of $\frac{n}{s}\cdot \mathcal{G}$. Let $\frac{1}{n^8}< \xi<1$. Let $c_1>1$ and $c_2>1$ be two large enough constants. For any $t\ge \frac{20\log n}{\varphi^2}$, if $\frac{\varepsilon}{\varphi^2}\le \frac{1}{10^5}$, $s\ge c_1\cdot n^{240\varepsilon/\varphi^2}\cdot \log n\cdot k^4$, $R\ge \frac{c_2\cdot k^6\cdot n^{1+760\varepsilon/\varphi^2}}{M\cdot \xi^2}$ and $1\le M\le O\left(\frac{n^{1/2-20\varepsilon/\varphi^2}}{k}\right)$, then with probability at least $1-2\cdot n^{-100}$, matrices $\widehat{\Sigma}_{[k]}^{-2}$ and $\widetilde{\Sigma}_{[k]}^{-4}$ exist and we have
        \[
            \lVert \widetilde{W}_{[k]}\widetilde{\Sigma}_{[k]}^{-4}\widetilde{W}_{[k]}^T-\widehat{W}_{[k]}\widehat{\Sigma}_{[k]}^{-2}\widehat{W}_{[k]}^T\rVert_2< \xi.
        \]
    \end{lemma}

    Equipped with~\Cref{lemma:bounded-space-G}, 
    to prove~\Cref{lemma:bounded-space-W}, we also need the following lemmas.

    \begin{lemma}[Lemma 18 in \citet{gluch2021spectral}]
        \label{lemma:gluch_1}
        Let $\widetilde{A},\widehat{A}\in \mathbb{R}^{n\times n}$ be symmetric matrices with eigendecomposition $\widetilde{A}=\widetilde{Y}\widetilde{\Gamma}\widetilde{Y}^T$ and $\widehat{A}=\widehat{Y}\widehat{\Gamma}\widehat{Y}^T$. Let the eigenvalues of $\widetilde{A}$ be $1\ge \gamma_1\ge \dots\ge \gamma_n\ge 0$. Suppose that $\lVert \widetilde{A}-\widehat{A}\rVert_2\le \frac{\gamma_k}{100}$ and $\gamma_{k+1}<\frac{\gamma_k}{4}$. Then we have
        \[
            \lVert \widetilde{Y}_{[k]}\widetilde{\Gamma}_{[k]}^{-1}\widetilde{Y}^T_{[k]}-\widehat{Y}_{[k]}\widehat{\Gamma}_{[k]}^{-1}\widehat{Y}^T_{[k]}\rVert_2\le \frac{16\lVert \widetilde{A}-\widehat{A}\rVert_2+4\gamma_{k+1}}{\gamma_k^2}.
        \]
    \end{lemma}

    \begin{lemma}[Lemma 28 in \citet{gluch2021spectral}]
        \label{lemma:gluch_2}
        Let $k\ge 2$ be an integer and $\varphi,\varepsilon\in(0,1)$. Let $G=(V,E)$ be a $d$-regular and $(k,\varphi,\varepsilon)$-clusterable graph. Let $\mM$ be the random walk transition matrix of $G$. Let $I_S=\{s_1,\dots,s_s\}$ be a multiset of $s$ indices chosen independently and uniformly at random form $V=\{1,\dots,n\}$. Let $\mS\in\mathbb{R}^{n\times s}$ be the matrix whose $i$-th column equals $\mathds{1}_{s_i}$. Let $c>1$ be a large enough constant. For any $t\ge \frac{20\log n}{\varphi^2}$, if $\frac{\varepsilon}{\varphi^2}\le \frac{1}{10^5}$ and $s\ge c\cdot n^{240\varepsilon/\varphi^2}\cdot \log n\cdot k^4$, then with probability at least $1-n^{-100}$, we have

        \begin{itemize}
            \item $v_k\left(\frac{n}{s}\cdot (\mM^t\mS)(\mM^t\mS)^T\right)=v_k\left(\frac{n}{s}\cdot (\mM^t\mS)^T(\mM^t\mS)\right)\ge \frac{n^{-80\varepsilon/\varphi^2}}{2}$,

            \item $v_{k+1}\left(\frac{n}{s}\cdot (\mM^t\mS)(\mM^t\mS)^T\right)\le n^{-9}$.
        \end{itemize}
    \end{lemma}

    \begin{lemma}[Weyl's Inequality]
        \label{lemma:weyl}
        Let $A,B\in \mathbb{R}^{n\times n}$ be symmetric matrices. Let $\alpha_1,\dots,\alpha_n$ and $\beta_1,\dots,\beta_n$ be the eigenvalues of $A$ and $B$ respectively. Then for any $i\in[n]$,  we have
        \[
            \left| \alpha_i-\beta_i\right|\le \Vert A-B\Vert_2.
        \]
    \end{lemma}

    Now we are ready to prove~\Cref{lemma:bounded-space-W}.

    \begin{proof}[Proof of~\Cref{lemma:bounded-space-W}]
        Let $c_3>1$ be a large enough constant and let $\sigma_\textup{err}=\frac{\xi\cdot n^{-1-360\varepsilon/\varphi^2}}{c_3\cdot k^2}$. Let $c$ be a constant from~\Cref{lemma:bounded-space-G}. By the assumption of the lemma for a large enough constant $c_2>1$, we have
        \[
            R\ge \frac{c_2\cdot k^6\cdot n^{1+760\varepsilon/\varphi^2}}{M\cdot \xi^2}\ge \frac{c\cdot k^2n^{-1+40\varepsilon/\varphi^2}}{\sigma_\textup{err}^2M}.
        \]
        Thus we can apply~\Cref{lemma:bounded-space-G}. Hence, with probability at least $1-n^{-100}$, we have
        \[
            \lVert \mathcal{G}-(\mM^t\mS)^T(\mM^t\mS)\rVert_2\le s\cdot\sigma_\textup{err}.
        \]

        Let $\widetilde{A}=\frac{n}{s}\cdot (\mM^t\mS)^T(\mM^t\mS)=\widetilde{W}\widetilde{\Sigma}^2\widetilde{W}^T$ and $\widehat{A}=\frac{n}{s}\cdot \mathcal{G}$. Thus, we have $\widetilde{A}^2=\left(\frac{n}{s}\cdot (\mM^t\mS)^T(\mM^t\mS)\right)^2=\widetilde{W}\widetilde{\Sigma}^4\widetilde{W}^T$ and $\widehat{A}^2=\left(\frac{n}{s}\cdot \mathcal{G}\right)^2=\widehat{W}\widehat{\Sigma}^2\widehat{W}^T$. To use~\Cref{lemma:gluch_1}, we have to bound $\lVert \widetilde{A}^2-\widehat{A}^2 \rVert_2=\left(\frac{n}{s}\right)^2\lVert \left((\mM^t\mS)^T(\mM^t\mS)\right)^2-\mathcal{G}^2 \rVert_2$. Using the triangle inequality and sub-multiplicativity of spectral norm and the above $\lVert \mathcal{G}-(\mM^tS)^T(\mM^tS)\rVert_2\le s\cdot\sigma_\textup{err}$ bound, we can get that 
        \[
            \lVert \left((\mM^t\mS)^T(\mM^t\mS)\right)^2-\mathcal{G}^2 \rVert_2\le (s\cdot \sigma_\textup{err})^2+2\cdot s\cdot \sigma_\textup{err}\lVert(\mM^t\mS)^T(\mM^t\mS) \rVert_2.
        \]

        Note that $\lVert(\mM^t\mS)^T(\mM^t\mS) \rVert_2\le \lVert(\mM^t\mS)^T(\mM^t\mS) \rVert_F=\sqrt{\sum_{i=1}^s{\sum_{j=1}^s{((\mM^t\mathds{1}_{s_i})^T(\mM^t\mathds{1}_{s_j}))^2}}}$, by Cauchy Schwarz inequality and~\Cref{lemma:gluch_0}, we can get that $\lVert(\mM^t\mS)^T(\mM^t\mS) \rVert_2\le O(s\cdot k^2\cdot n^{-1+40\varepsilon/\varphi^2})$.
        Put them together and by the choice of $\sigma_\textup{err}=\frac{\xi\cdot n^{-1-360\varepsilon/\varphi^2}}{c_3\cdot k^2}$, we have that 
        \[
            \lVert \widetilde{A}^2-\widehat{A}^2 \rVert_2\le O\left(\frac{\xi \cdot n^{-320\varepsilon/\varphi^2}}{c_3}\right).
        \]

        Moreover, let $c_1$ be the constant from~\Cref{lemma:gluch_2}, since $s\ge c_1\cdot n^{240\varepsilon/\varphi^2}\cdot \log n\cdot k^4$, by~\Cref{lemma:gluch_2}, with probability at least $1-n^{-100}$, we have
        \[
            v_k\left(\widetilde{A}^2\right)=v_k\left(\left(\frac{n}{s}\cdot (\mM^t\mS)^T(\mM^t\mS)\right)^2\right)\ge \left(\frac{n^{-80\varepsilon/\varphi^2}}{2}\right)^2=\frac{n^{-160\varepsilon/\varphi^2}}{4},
        \]

        and 
        \[
            v_{k+1}\left(\widetilde{A}^2\right)=v_{k+1}\left(\left(\frac{n}{s}\cdot (\mM^t\mS)^T(\mM^t\mS)\right)^2\right)\le (n^{-9})^2=n^{-18}.
        \]

        By Weyl's inequality, we have that $v_k(\widehat{A}^2)\ge v_k(\widetilde{A}^2)-\lVert \widetilde{A}^2-\widehat{A}^2 \rVert_2\ge \frac{n^{-160\varepsilon/\varphi^2}}{4}-O(\frac{\xi\cdot n^{-320\varepsilon/\varphi^2}}{c_3})>0$, so $\widehat{\Sigma}_{[k]}^{-2}$ exists. Moreover, since $\widetilde{A}^2,\widehat{A}^2$ are symmetric matrices, $\lVert \widetilde{A}^2-\widehat{A}^2\rVert_2\le \frac{v_k(\widetilde{A}^2)}{100}$ and $v_{k+1}(\widetilde{A}^2)<\frac{v_k(\widetilde{A}^2)}{4}$, by~\Cref{lemma:gluch_1}, we have that 
        \begin{align*}
            \lVert \widetilde{W}_{[k]}\widetilde{\Sigma}_{[k]}^{-4}\widetilde{W}_{[k]}^T-\widehat{W}_{[k]}\widehat{\Sigma}_{[k]}^{-2}\widehat{W}_{[k]}^T\rVert_2&\le \frac{16\lVert \widetilde{A}^2-\widehat{A}^2\rVert_2+4v_{k+1}(\widetilde{A}^2)}{v_k(\widetilde{A}^2)^2}\\
            &\le\frac{O\left(\frac{\xi\cdot n^{-320\varepsilon/\varphi^2}}{c_3}\right)+4n^{-18}}{\frac{n^{-320\varepsilon/\varphi^2}}{16}}\\
            &\le O\left(\frac{\xi}{c_3}\right)+64n^{-17}\\
            &\le \xi.&&\text{$\frac{1}{n^8}\le\xi$}
        \end{align*}

        Moreover, both~\Cref{lemma:bounded-space-G} and~\Cref{lemma:gluch_2} fail with probability at most $n^{-100}$, by union bound, we can get that the above inequality holds with probability at least $1-2n^{-100}$.
    \end{proof}

    \begin{lemma}[Restatement of~\Cref{lemma:bounded-space-apx}]
        Let $k\ge 2$ be an integer and $\varphi,\varepsilon\in(0,1)$. Let $G=(V,E)$ be a $d$-regular and $(k,\varphi,\varepsilon)$-clusterable graph. Let $\mM$ be the random walk transition matrix of $G$. Let $I_S=\{s_1,\dots,s_s\}$ be a multiset of $s$ indices chosen independently and uniformly at random form $V=\{1,\dots,n\}$. Let $\mS\in\mathbb{R}^{n\times s}$ be the matrix whose $i$-th column equals $\mathds{1}_{s_i}$. Let $\sqrt{\frac{n}{s}}\cdot \mM^t\mS=\widetilde{U}\widetilde{\Sigma}\widetilde{W}^T$ be an SVD of $\sqrt{\frac{n}{s}}\cdot \mM^t\mS$ where $\widetilde{U}\in\mathbb{R}^{n\times n},\widetilde{\Sigma}\in\mathbb{R}^{n\times n},\widetilde{W}\in\mathbb{R}^{s\times n}$. Let $\frac{1}{n^6}<\xi<1$ and $1\le M_\textup{init}\le O\left(\frac{n^{1/2-20\varepsilon/\varphi^2}}{k}\right)$. Let $t\ge \frac{20\log n}{\varphi^2}$. Let $c>1$ be a large enough constant. Let $s\ge c\cdot n^{240\varepsilon/\varphi^2}\cdot \log n\cdot k^4$. Let $\Psi$ denote the matrix constructed by \textsc{InitOracle} $(G,k,\xi,M_{\textup{init}})$ (Alg. \ref{algo:bounded-space-initialize-oracle}). 
        
        Let $x,y\in V$. Let $\langle \vf_x,\vf_y\rangle_{\textup{apx}}\in \mathbb{R}$ denote the value returned by \textsc{QueryDot} $(G,x,y,\xi,\Psi,M_{\textup{query}})$ (Alg. \ref{algo:bounded-space-query-dot}). If $\frac{\varepsilon}{\varphi^2}\le \frac{1}{10^5}$, Alg. \ref{algo:bounded-space-initialize-oracle} succeeds and $1\le M_\textup{query}\le O\left(\frac{n^{1/2-20\varepsilon/\varphi^2}}{k}\right)$, then with probability at least $1-5n^{-100}$ matrix $\widetilde{\Sigma}_{[k]}^{-4}$ exists and we have 
        \[
            \left|\langle \vf_x,\vf_y\rangle_{\textup{apx}}-(\mM^t\mathds{1}_x)^T(\mM^t\mS)\left(\frac{n}{s}\cdot \widetilde{W}_{[k]}\widetilde{\Sigma}_{[k]}^{-4}\widetilde{W}_{[k]}^T\right)(\mM^t\mS)^T(\mM^t\mathds{1}_y)\right|<\frac{\xi}{n}.
        \]
        
        \begin{proof}
            Note that in line $8$ of Alg. \ref{algo:bounded-space-query-dot}, $\langle \vf_x,\vf_y\rangle_{\textup{apx}}$ is defined as $\bm{\alpha}_x^T\Psi\bm{\alpha}_y$, where in line $8$ of Alg. \ref{algo:bounded-space-initialize-oracle}, $\Psi\in\mathbb{R}^{s\times s}$ is defined to be $\Psi = \frac{n}{s}\cdot \widehat{W}_{[k]}\widehat{\Sigma}_{[k]}^{-2}\widehat{W}^T_{[k]}$ and $\bm\alpha_x,\bm\alpha_y\in\mathbb{R}^s$ are vectors obtained by taking entriwise median over all $O(\log n)$ runs (see lines $3\sim 7$ of Alg. \ref{algo:bounded-space-query-dot}). 

            For any vertex $x\in V$, we use $\vp_x$ to denote $\vp_x=\mM^t\mathds{1}_x$. We then define 
            \begin{align*}
                \mathbf{a}_x = \vp_x^T(\mM^t\mS), A=\frac{n}{s}\cdot \widetilde{W}_{[k]}\widetilde{\Sigma}_{[k]}^{-4}\widetilde{W}_{[k]}^T, \mathbf{a}_y = (\mM^t\mS)^T\vp_x,\\
                \mathbf{e}_x=\bm\alpha_x^T-\mathbf{a}_x,\quad E=\Psi-A, \quad \mathbf{e}_y=\bm\alpha_y-\mathbf{a}_y.
            \end{align*}
            Then by triangle inequality, we have
            \begin{align*}
                &\left|\bm\alpha^T\Psi\bm\alpha_y-\vp_x^T(\mM^t\mS)\left(\frac{n}{s}\cdot \widetilde{W}_{[k]}\widetilde{\Sigma}_{[k]}^{-4}\widetilde{W}_{[k]}^T\right)(\mM^t\mS)^T\vp_y\right|\\
                &=\left|(\mathbf{a}_x+\mathbf{e}_x)(A+E)(\mathbf{a}_y+\mathbf{e}_y)-\mathbf{a}_xA\mathbf{a}_y\right|\\
                &\le \lVert \mathbf{e}_x\rVert_2\lVert E\rVert_2\lVert \mathbf{e}_y\rVert_2 + \lVert \mathbf{e}_x\rVert_2\lVert A\rVert_2\lVert \mathbf{e}_y\rVert_2 + \lVert \mathbf{a}_x\rVert_2\lVert E\rVert_2\lVert \mathbf{e}_y\rVert_2\\
                &+ \lVert \mathbf{a}_x\rVert_2\lVert A\rVert_2\lVert \mathbf{e}_y\rVert_2 + \lVert \mathbf{a}_x\rVert_2\lVert E\rVert_2\lVert \mathbf{a}_y\rVert_2 + \lVert \mathbf{e}_x\rVert_2\lVert A\rVert_2\lVert \mathbf{a}_y\rVert_2 + \lVert \mathbf{a}_x\rVert_2\lVert E\rVert_2\lVert \mathbf{a}_y\rVert_2.
            \end{align*}
            
            In the following, we bound $\lVert \mathbf{a}_x\rVert_2,\lVert \mathbf{a}_y\rVert_2, \lVert E\rVert_2, \lVert A\rVert_2, \lVert \mathbf{e}_x\rVert_2$ and $\lVert \mathbf{e}_x\rVert_2$. 
            
            Let $c^\prime>1$ be a constant and let $\xi^\prime=\frac{\xi}{c^\prime \cdot k^4\cdot n^{80\varepsilon/\varphi^2}}$. Thus for large enough constant $c$, we have $s\ge c_1\cdot n^{240\varepsilon/\varphi^2}\cdot \log n\cdot k^4$ and $R_{\textup{init}}=\Theta(\frac{n^{1+920\varepsilon/\varphi^2}}{M_\textup{init}}\cdot \frac{k^{14}}{\xi^2})\ge \frac{c_2k^6\cdot n^{1+760\varepsilon/\varphi^2}}{M_\textup{init}\cdot \xi^{\prime2}}$ as in line $2$ of Alg. \ref{algo:bounded-space-initialize-oracle}, hence, by~\Cref{lemma:bounded-space-W} applied with $\xi^\prime$ we have that with probability at least $1-2n^{-100}$, $\widehat{\Sigma}_{[k]}^{-2}$ and $\widetilde{\Sigma}_{[k]}^{-4}$ exist and we have
            \begin{align}
                \label{eq:1}
                \lVert E\rVert_2=\frac{n}{s}\cdot \lVert\widehat{W}_{[k]}\widehat{\Sigma}_{[k]}^{-2}\widehat{W}_{[k]}^T -  \widetilde{W}_{[k]}\widetilde{\Sigma}_{[k]}^{-4}\widetilde{W}_{[k]}^T\rVert_2< \frac{n}{s}\cdot \xi^\prime=\frac{\xi\cdot n}{c^\prime \cdot k^4\cdot n^{80\varepsilon/\varphi^2}\cdot s}.
            \end{align}

            Moreover, according to the proof of Lemma 29 in \cite{gluch2021spectral}, we have that, with probability at least $1-n^{-100}$, 
            \begin{align}
            \label{eq:2}
                \lVert A\rVert_2\le \frac{4\cdot n^{1+160\varepsilon/\varphi^2}}{s}.
            \end{align}

            And with probability $1$, we have
            \begin{align}
            \label{eq:3}
                \lVert \mathbf{a}_x\rVert_2\le O(\sqrt{s}\cdot k^2\cdot n^{-1+40\varepsilon/\varphi^2})
            \end{align}
            and
            \begin{align}
            \label{eq:4}
                \lVert \mathbf{a}_y\rVert_2\le O(\sqrt{s}\cdot k^2\cdot n^{-1+40\varepsilon/\varphi^2}).
            \end{align}

            Now we need to bound $\mathbf{e}_x$ and $\mathbf{e}_y$. Recall that $\mathbf{e}_x=\bm\alpha_x^T-\vp_x^T(\mM^t\mS)$, where $\bm\alpha_x\in \mathbb{R}^s$ is obtained by taking entrywise median over all $\vx_l$'s. Note that in line $5$ of Alg. \ref{algo:bounded-space-query-dot}, $\vx_l(i)$ is the output of \textsc{EstRWDot} $(G,R_{\textup{query}},t,M_\textup{query},x,s_i)$ (Alg. \ref{algo:bounded-space-esti-rw-dot}). Let $c_3$ be a constant infront of $R$ in~\Cref{lemma:bounded-space-z}. Let $\sigma_\textup{err}=\frac{\xi}{c^\prime\cdot k^2\cdot n^{1+200\varepsilon/\varphi^2}}$. Thus by our choice of $R_{\textup{query}}=\Theta(\frac{n^{1+440\varepsilon/\varphi^2}}{M_{\textup{query}}}\cdot\frac{k^6}{\xi^2})$ in line $2$ of Alg. \ref{algo:bounded-space-query-dot}, the prerequisites of~\Cref{lemma:bounded-space-z} are satisfied:
            \[
                R_{\textup{query}}=\Theta\left(\frac{n^{1+440\varepsilon/\varphi^2}}{M_\textup{query}}\cdot \frac{k^6}{\xi^2}\right)\ge \frac{c_3\cdot k^2n^{-1+40\varepsilon/\varphi^2}}{\sigma_\textup{err}^2\cdot M_\textup{query}}.
            \]
            Thus we can apply~\Cref{lemma:bounded-space-z}. Hence, for any $1\le i\le s$ with probability at least $0.99$, we have
            \[
                |\vx_l(i)-\vp_x^T\vp_{s_i}|\le \sigma_\textup{err}.
            \]

            Since we are running $O(\log n)$ rounds to compute $\vx_l$'s and $\bm\alpha_x$ is obtained by taking entrywise median, we can get that with probability at least $1-n^{-100}$ for all $z\in I_S$ (see~\Cref{remark:median-trick}), we have
            \[
                |\bm\alpha_x(z)-\vp_x^T\vp_{z}|\le \sigma_\textup{err}.
            \]

            Therefore, with probability at least $1-n^{-100}$, we can get
            \begin{align}
            \label{eq:5}
                \lVert \mathbf{e}_x \rVert_2=\lVert \bm\alpha_x^T-\vp_x^T(\mM^t\mS)\rVert_2\le \sqrt{s}\cdot \sigma_\textup{err}=\frac{\sqrt{s}\cdot \xi}{c^\prime\cdot k^2\cdot n^{1+200\varepsilon/\varphi^2}}.
            \end{align}
            Using the same analysis, with probability at least $1-n^{-100}$, we can get that
            \begin{align}
            \label{eq:6}
                \lVert \mathbf{e}_y \rVert_2=\lVert \bm\alpha_y-(\mM^t\mS)^T\vp_y\rVert_2\le \sqrt{s}\cdot \sigma_\textup{err}=\frac{\sqrt{s}\cdot \xi}{c^\prime\cdot k^2\cdot n^{1+200\varepsilon/\varphi^2}}.
            \end{align}
            Putting (\ref{eq:1}),(\ref{eq:2}),(\ref{eq:3}),(\ref{eq:4}),(\ref{eq:5}),(\ref{eq:6}) together and for large enough $n$, we can get
            \begin{align*}
                &\left|\bm\alpha^T\Psi\bm\alpha_y-\vp_x^T(\mM^t\mS)\left(\frac{n}{s}\cdot \widetilde{W}_{[k]}\widetilde{\Sigma}_{[k]}^{-4}\widetilde{W}_{[k]}^T\right)(\mM^t\mS)^T\vp_y\right|\\
                &\le \lVert \mathbf{e}_x\rVert_2\lVert E\rVert_2\lVert \mathbf{e}_y\rVert_2 + \lVert \mathbf{e}_x\rVert_2\lVert A\rVert_2\lVert \mathbf{e}_y\rVert_2 + \lVert \mathbf{a}_x\rVert_2\lVert E\rVert_2\lVert \mathbf{e}_y\rVert_2\\
                &+ \lVert \mathbf{a}_x\rVert_2\lVert A\rVert_2\lVert \mathbf{e}_y\rVert_2 + \lVert \mathbf{a}_x\rVert_2\lVert E\rVert_2\lVert \mathbf{a}_y\rVert_2 + \lVert \mathbf{e}_x\rVert_2\lVert A\rVert_2\lVert \mathbf{a}_y\rVert_2 + \lVert \mathbf{a}_x\rVert_2\lVert E\rVert_2\lVert \mathbf{a}_y\rVert_2\\
                &\le O(\frac{\xi}{c^\prime\cdot n})\\
                &\le \frac{\xi}{n}.
            \end{align*}
            The last inequality holds by setting $c^\prime$ be a large enough constant to cancel the hidden constant of $O(\frac{\xi}{c^\prime\cdot n})$.

            Using union bound, if Alg. \ref{algo:bounded-space-initialize-oracle} succeeds, then the above inequality holds with probability at least $1-2n^{-100}-n^{-100}-2n^{-100}=1-5n^{-100}$.
        \end{proof}
    \end{lemma}

\section{Proof of~\Cref{itm:main-case1} in~\Cref{thrm:main}}
\label{appendix:section-log(k)-oracle}

In this section, we first present an algorithm for computing the spectral dot product in a subspace, which will serve as a building block for the sublinear spectral clustering oracle that relies on a $\log(k)$ conductance gap. Next, we introduce the sublinear spectral clustering oracle, originally proposed in \citet{gluch2021spectral}, corresponding to~\Cref{itm:main-case1} in~\Cref{thrm:main}. Finally, we provide the proof of~\Cref{itm:main-case1} in~\Cref{thrm:main}.

\subsection{Dot product oracle on subspace}

Note that the clustering oracle in \citet{gluch2021spectral} relies on cluster centers: 

\begin{defn}[Cluster center]
    \label{defn:cluster-center}
    \upshape
    For a vertex set $C\subset V$, the \emph{cluster center} of $C$ is defined to be
    \[
        \mu_C=\frac{1}{|C|}\sum_{x\in C}{\vf_x}.
    \]
\end{defn}

They proved that if $x\in C_i$, then $\vf_x$ is close to $\mu_{C_i}$, which means $\langle \vf_x,\mu_C\rangle\ge c\cdot \lVert \mu_C\rVert_2^2$, where $c$ is a constant. Therefore, the key idea behind the clustering oracle in \citet{gluch2021spectral} is to sample a subset of vertices and enumerate possible $k$-partition in order to obtain a good approximation $\widehat{\mu}_1,\dots,\widehat{\mu}_k$ to the true cluster centers $\mu_1,\dots,\mu_k$ (see lines $6\sim 11$ of Alg. \ref{alg:gluch-find-centers}). When answering an arbitrary \textsc{WhichCluster} $(G,x)$ query, the oracle assigns the $x$ to the cluster whose center is close to $\vf_x$ while other cluster centers are not close to $\vf_x$ (see line $5$ of Alg. \ref{alg:gluch-hyperplane}).

In fact, their clustering algorithm uses hyperplane partitioning, which requires computing dot products in the subspace (i.e., $\langle \vf_x, \Pi\mu\rangle$). Therefore, we first present the algorithm that computes the dot products in the subspace based on our improved version. We highlight that this (i.e., Alg. \ref{alg:gluch-on-subspace}) is not our contribution.

\begin{algorithm}[H]%
    \DontPrintSemicolon
    \caption{\textsc{DotProductOracleOnSubspace}$(G,x,y, \xi,\Psi,M,B_1,\dots,B_r)$}
    \label{alg:gluch-on-subspace}

    Let $\mX\in\mathbb{R}^{r\times r},\vh_x\in \mathbb{R}^r,\vh_y\in \mathbb{R}^r$\;

    Let $\xi^\prime = \Theta(\xi \cdot n^{-80\varepsilon/\varphi^2}\cdot k^{-6})$\;
    \For{$i,j\in[r]$}{
        $\mX(i,j)\coloneqq\frac{1}{|B_i||B_j|}\cdot \sum_{z_i\in B_i}\sum_{z_j\in B_j}\textsc{QueryDot}(G,z_i,z_j,\xi^\prime,\Psi,M)$\;
    }

    \For{$i\in [r]$}{
        $\vh_x(i)\coloneqq\frac{1}{|B_i|}\cdot \sum_{z_i\in B_i}\textsc{QueryDot}(G,z_i,x,\xi^\prime,\Psi,M)$\;
        $\vh_y(i)\coloneqq\frac{1}{|B_i|}\cdot \sum_{z_i\in B_i}\textsc{QueryDot}(G,z_i,y,\xi^\prime,\Psi,M)$\;
    }
    return $\langle \vf_x,\widehat{\Pi}f_y\rangle_\textup{apx}\coloneqq\textsc{QueryDot}(G,x,y,\xi^\prime,\Psi,M)-\vh_x^T\mX^{-1}\vh_y$\;
\end{algorithm}

In the following, we will give some informal theorem and corollaries about Alg. \ref{alg:gluch-on-subspace}. Note that the only modification we make to Alg. \ref{alg:gluch-on-subspace} is to replace \textsc{SpectralDotProduct} with our improved version. Since our dot product oracle provides the same correctness guarantees as the original one, the correctness of the theorem and corollaries concerning Alg.~\ref{alg:gluch-on-subspace} follows immediately from the proof of Theorem $6$ in \citet{gluch2021spectral}. Therefore, we focus on analyzing the time and space complexities.

\begin{theorem}[Informal]
    \label{thrm:subspace}
    Let $k\ge$ be an integer, $\varphi$, $\frac{1}{n^5}<\xi<1$ and $\frac{\varepsilon}{\varphi^2}$ be smaller than a positive absolute constant. Let $G=(V,E)$ be a $d$-regular and $(k,\varphi,\varepsilon)$-clusterable graph with $C_1,\dots,C_k$. %

    Let $r\in[k]$. Let $B_1,\dots,B_r$ denote multisets of vertices. Let $b=\max_{i\in[r]}{|B_i|}$. %
    Let $\widehat{\mu}_i=\frac{1}{|B_i|}\sum_{x\in B_i}{\vf_x}$. Let $\widehat{\Pi}$ is defined as a orthogonal projection onto the span $(\{\widehat{\mu}_1,\dots,\widehat{\mu}_r\})^\bot$. Then for all $x,y\in V$, we have
    
    \begin{enumerate}[label=\arabic*]%
        \item $\left|\langle \vf_x,\widehat{\Pi}\vf_y\rangle_\textup{apx}-\langle\vf_x,\widehat{\Pi}\vf_y \rangle\right|\le \frac{\xi}{n}$, where $\langle \vf_x,\widehat{\Pi}\vf_y\rangle_\textup{apx}$ is the output of Alg. \ref{alg:gluch-on-subspace},
        \item Alg. \ref{alg:gluch-on-subspace} runs in $b^2\cdot (\frac{k}{\xi})^{O(1)}\cdot n^{1+O(\varepsilon/\varphi^2)}\cdot \frac{1}{M}\cdot \log^3 n\cdot \frac{1}{\varphi^2}$ time,
        \item Alg. \ref{alg:gluch-on-subspace} uses $b^2\cdot (\frac{k}{\xi})^{O(1)}\cdot n^{O(\varepsilon/\varphi^2)}\cdot M\cdot \log^3 n$ bits of space.
    \end{enumerate}

    \begin{proof}
        In lines $3\sim 4$ of Alg. \ref{alg:gluch-on-subspace}, to compute $\mX$, Alg. \ref{alg:gluch-on-subspace} calls \textsc{QueryDot} for $r^2\cdot b^2\le k^2\cdot b^2$ times. In lines $5\sim 7$ of Alg. \ref{alg:gluch-on-subspace}, to compute $\vh_x,\vh_y$, Alg. \ref{alg:gluch-on-subspace} calls \textsc{QueryDot} for $r\cdot b\le k\cdot b$ times. To compute $\mX^{-1}$, it takes $r^3\le k^3$ time. Therefore, Alg. \ref{alg:gluch-on-subspace} runs in $k^2\cdot b^2\cdot T_\textup{query}+k\cdot b\cdot T_\textup{query}+k^3$ time and it uses $k^2\cdot b^2\cdot S_\textup{query}+k\cdot b\cdot S_\textup{query}+k^2$ bits of space. Note that $T_\textup{query}=(\frac{k}{\xi^\prime})^{O(1)}\cdot n^{1+O(\varepsilon/\varphi^2)}\cdot \frac{1}{M}\cdot \log^3 n\cdot \frac{1}{\varphi^2}$ and $S_\textup{query}=(\frac{k}{\xi^\prime})^{O(1)}\cdot n^{O(\varepsilon/\varphi^2)}\cdot M\cdot \log^3 n$, where $\xi ^\prime= \Theta(\xi \cdot n^{-80\varepsilon/\varphi^2}\cdot k^{-6})$. Therefore, we get that Alg. \ref{alg:gluch-on-subspace} runs in $b^2\cdot (\frac{k}{\xi})^{O(1)}\cdot n^{1+O(\varepsilon/\varphi^2)}\cdot \frac{1}{M}\cdot \log^3 n\cdot \frac{1}{\varphi^2}$ time and uses $b^2\cdot (\frac{k}{\xi})^{O(1)}\cdot n^{O(\varepsilon/\varphi^2)}\cdot M\cdot \log^3 n$ bits of space.
    \end{proof}
\end{theorem}

\begin{corollary}
    \label{coro:subspace-1}
    There exists an algorithm that
    \begin{enumerate}[label=\arabic*]
        \item returns a value $\langle \vf_x,\widehat{\Pi}\widehat{\mu}\rangle_\textup{apx}$ such that $\left|\langle \vf_x,\widehat{\Pi}\widehat{\mu}\rangle_\textup{apx}-\langle \vf_x,\widehat{\Pi}\widehat{\mu}\rangle\right|\le \frac{\xi}{n}$,
        \item runs in $b^3\cdot (\frac{k}{\xi})^{O(1)}\cdot n^{1+O(\varepsilon/\varphi^2)}\cdot \frac{1}{M}\cdot \log^3 n\cdot \frac{1}{\varphi^2}$ time,
        \item uses $b^3\cdot (\frac{k}{\xi})^{O(1)}\cdot n^{O(\varepsilon/\varphi^2)}\cdot M\cdot \log^3 n$ bits of space.

    \end{enumerate}

    \begin{proof}
        One can compute $\langle \vf_x,\widehat{\Pi}\widehat{\mu}\rangle_\textup{apx}\coloneqq\frac{1}{|B|}\cdot \sum_{y\in B}{\textsc{DotProductOracleOnSubspace}}(G,x,y,\\ \xi,\Psi,M,B_1,\dots,B_r)$ (Alg. \ref{alg:gluch-on-subspace}). Therefore, the algorithm that computes $\langle \vf_x,\widehat{\Pi}\widehat{\mu}\rangle_\textup{apx}$ calls Alg. \ref{alg:gluch-on-subspace} $b$ times, which ends the proof.
    \end{proof}
    
\end{corollary}

\begin{corollary}
    \label{coro:subspace-2}
    There exists an algorithm that
    \begin{enumerate}[label=\arabic*]
        \item returns a value $\lVert \widehat{\Pi}\widehat{\mu}\rVert^2_\textup{apx}$ such that $\left|\lVert \widehat{\Pi}\widehat{\mu}\rVert^2_\textup{apx}-\lVert \widehat{\Pi}\widehat{\mu}\rVert^2\right|\le \frac{\xi}{n}$,
        \item runs in $b^4\cdot (\frac{k}{\xi})^{O(1)}\cdot n^{1+O(\varepsilon/\varphi^2)}\cdot \frac{1}{M}\cdot \log^3 n\cdot \frac{1}{\varphi^2}$ time,
        \item uses $b^4\cdot (\frac{k}{\xi})^{O(1)}\cdot n^{O(\varepsilon/\varphi^2)}\cdot M\cdot \log^3 n$ bits of space.

    \end{enumerate}

    \begin{proof}
        One can compute $\lVert \widehat{\Pi}\widehat{\mu}\rVert^2_\textup{apx}=(\widehat{\Pi}\widehat{\mu})^T(\widehat{\Pi}\widehat{\mu})=\widehat{\mu}^T\widehat{\Pi}^T\widehat{\Pi}\widehat{\mu}=\widehat{\mu}^T\widehat{\Pi}\widehat{\mu}=\langle\widehat{\mu}, \widehat{\Pi}\widehat{\mu}\rangle=\frac{1}{|B|}\cdot \sum_{x\in B}\langle \vf_x,\widehat{\Pi}\widehat{\mu}\rangle_\textup{apx}$. Therefore, the algorithm that computes $\lVert \widehat{\Pi}\widehat{\mu}\rVert^2_\textup{apx}$ calls the algorithm in ~\Cref{coro:subspace-1} $b$ times, which ends the proof.
    \end{proof}
    
\end{corollary}

    \subsection{Sublinear spectral clustering oracle}
        Now we present the sublinear spectral clustering oracle with a $\log(k)$ gap between inner and outer conductance, originally proposed in \citet{gluch2021spectral}, and adapt it by incorporating our dot product oracle, which operates with very little memory.
        
        Algorithm \ref{alg:gluch-find-centers} finds some cluster centers that reflects the clustering structure of the input graph.
    
    \begin{algorithm}[H]%
        \DontPrintSemicolon
        \caption{\textsc{FindCenters}($G, M$)}
        \label{alg:gluch-find-centers}

        \textsc{InitOracle}$(G,k,10^{-6}\cdot \frac{\sqrt{\varepsilon}}{\varphi}, M)$\;
        $s_1\coloneqq\Theta\left(\frac{\varphi^2}{\varepsilon}k^5\log^2k\log (1/\eta)\right)$, $s_2\coloneqq\Theta\left(\frac{\varphi^4}{\varepsilon^2}k^5\log^2k\log (1/\eta)\right)$\;

        \For{$t\in [1 \dots \log(2/\eta)]$}{
            $S\coloneqq$Random samples of vertices of $V$ of size $s=\Theta(\frac{\varphi^2}{\varepsilon}k^4\log k)$\;
            \For{$(P_1,P_2,\dots,P_k)\in $\textsc{Partition} $(S)$}{
                \For{$i=1$ to $k$}{
                    $\widehat{\mu}_i\coloneqq \frac{1}{|P_i|}\sum_{x\in P_i}{f_x}$\;
                }
                $(r,C)\coloneqq$ \textsc{ComputerOrderedPartition}$(G,(\widehat{\mu}_1,\dots,\widehat{\mu}_k)),s_1,s_2,M)$\;

                \If{$r=$\textsc{True}}{
                    return $C$\;
                }
            }
        }
    \end{algorithm}

    \begin{algorithm}[H]%
        \DontPrintSemicolon
        \caption{\textsc{ComputeOrderedPartition}($G, (\widehat{\mu}_1,\dots,\widehat{\mu}_k),s_1,s_2,M$)}
        \label{alg:gluch-compute-ordered-partition}

        $S\coloneqq\{\widehat{\mu}_1,\dots,\widehat{\mu}_k\}$\;
        
        \For{$i=1$ to $\lceil \log k\rceil$}{
            $T_i\coloneqq\emptyset$\;
            \For{$\widehat{\mu}\in S$}{
                $\psi\coloneqq$\textsc{OuterConductance}$(G,\widehat{\mu},(T_1,\dots,T_{i-1}),S,s_1,s_2,M)$\;
                \If{$\psi\le O(\frac{\varepsilon}{\varphi^2}\cdot \log k)$}{
                    $T_i\coloneqq T_i\cup \{\widehat{\mu}\}$\;
                }
            }
            $S\coloneqq S\backslash T_i$\;
            \If{$S=\emptyset$}{
                return $(\textsc{True},(T_1,\dots,T_i))$\;
            }
        }
        return $(\textsc{False},\bot)$\;
    \end{algorithm}

    \begin{algorithm}[H]%
        \DontPrintSemicolon
        \caption{\textsc{OuterConductance}($G, \widehat{\mu},(T_1,\dots,T_b),S,s_1,s_2,M$)}
        \label{alg:gluch-outer-conductance}

        $\textup{cnt}\coloneqq 0$\;
        \For{$t=1$ to $s_1$}{
            $x\sim $\textsc{Uniform}$\{1\dots n\}$\;
            \If{\textsc{IsInside}$(x,\widehat{\mu},(T_1,\dots,T_b),S,M)$}{
                $\textup{cnt}\coloneqq \textup{cnt}+1$\;
            }
        }
        \If{$\frac{n}{s_1}\cdot \textup{cnt}<\min_{p\in[k]}{|C_p|/2}$}{
            return $\infty$\;
        }
        $e\coloneqq 0,a\coloneqq 0$\;
        \For{$t=1$ to $s_2$}{
            $x\sim \textsc{Uniform}\{1\dots n\}$\;
            $y\sim \textsc{Uniform}\{w\in \mathcal{N}(u)\}$\;
            \If{\textsc{IsInside}$(x,\widehat{\mu},(T_1,\dots,T_b),S,M)$}{
                $a\coloneqq a+1$\;
                \If{$\neg$\textsc{IsInside}$(y,\widehat{\mu},(T_1,\dots,T_b),S,M)$}{
                    $e\coloneqq e+1$\;
                }
            }
        }
        return $\frac{e}{a}$\;
    \end{algorithm}

    \begin{algorithm}[H]%
        \DontPrintSemicolon
        \caption{\textsc{IsInside}($x, \widehat{\mu},(T_1,\dots,T_b),S,M$)}
        \label{alg:gluch-is-inside}

        \For{$i=1$ to $b$}{
            Let $\Pi$ be the projection onto the span $(\cup_{j<i}T_j)^{\bot}$\;
            Let $S_i=(\cup_{j\ge i}{T_j})\cup S$\;
            \For{$\widehat{\mu}_i\in T_i$}{
                \If{$x\in C_{\Pi\widehat{\mu}_i,0.93}^{\textup{apx}}\backslash \cup_{\widehat{\mu}^\prime\in S_i\backslash \{\widehat{\mu}_i\}}{C_{\Pi\widehat{\mu}^\prime,0.93}^{\textup{apx}}}$}{
                    return \textsc{False}\;
                }
            }
        }

        Let $\Pi$ be the projection onto the span $(\cup_{j\le b}T_j)^{\bot}$\;
        \If{$x\in C_{\Pi\widehat{\mu},0.93}^{\textup{apx}}\backslash \cup_{\widehat{\mu}^\prime\in S\backslash \{\widehat{\mu}\}}{C_{\Pi\widehat{\mu}^\prime,0.93}^{\textup{apx}}}$}{
            return \textsc{True}\;
        }
        return \textsc{False}\;
        
    \end{algorithm}

    Algorithm \ref{alg:gluch-hyperplane} corresponds to the query phase of the clustering oracle where it is used to assign vertices to clusters based on cluster centers.

    \begin{algorithm}[H]%
        \DontPrintSemicolon
        \caption{\textsc{HyperplanePartitioning}($x,(T_1,\dots,T_b),M$)}
        \label{alg:gluch-hyperplane}

        \For{$i=1$ to $b$}{
            Let $\Pi$ be the projection onto the span $(\cup_{j<i}T_j)^{\bot}$\;
            Let $S_i=(\cup_{j\ge i}{T_j})$\;
            \For{$\widehat{\mu}\in T_i$}{
                \If{$x\in C_{\Pi\widehat{\mu},0.93}^{\textup{apx}}\backslash \cup_{\widehat{\mu}^\prime\in S_i\backslash \{\widehat{\mu}\}}{C_{\Pi\widehat{\mu}^\prime,0.93}^{\textup{apx}}}$}{
                    return $\widehat{\mu}$\;
                }
            }
        }
        
    \end{algorithm}

    \subsection{Deferred proof}
        \label{section:proof-of-log(k)}

        \begin{theorem}[Restate of~\Cref{itm:main-case1} in~\Cref{thrm:main}]
            \label{thrm:main-restate-item-1}
            
            Let $k\ge 2$ be an integer, $\varphi,\varepsilon\in (0,1)$ and $h_1(k,\varphi),h_2(k,\varepsilon)$ and $h_3(k,\varphi,\varepsilon)$ be three functions. Let $\varepsilon \ll h_1(k,\varphi)$. Let $G=(V,E)$ be a $d$-regular and $(k,\varphi,\varepsilon)$-clusterable graph with $C_1,\dots,C_k$. Let $1\le M\le O\left(\frac{n^{1/2-O(\varepsilon/\varphi^2)}}{k}\right)$ be a trade-off parameter. There exists a sublinear spectral clustering oracle that:%

            \begin{itemize}
                \item constructs a data structure $\mathcal{D}$ using $\widetilde{O}_\varphi\left(h_2(k)\cdot n^{O(\varepsilon/\varphi^2)}\cdot M\right)$ bits of space,
                \item answers any \textsc{WhichCluster} query using $\mathcal{D}$ in $\widetilde{O}_\varphi\left(h_2(k)\cdot n^{1+O(\varepsilon/\varphi^2)}\cdot \frac{1}{M}\right)$ time,
                \item has $O\left(h_3(k,\varphi,\varepsilon)\right)|C_i|$ misclassification error for each $i\in[k]$,
            \end{itemize}
            where we use $O_\varphi$ to suppress dependence on $\varphi$ and $\widetilde{O}$ to hide all $\textup{poly}(\log n)$ factors and:
            \begin{enumerate}[label=\arabic*]
                \item if $h_1(k,\varphi)=\frac{\varphi^3}{\log k}$, then $h_2(k,\varepsilon)=\left(\frac{k}{\varepsilon}\right)^{O(1)}$ and $h_3(k,\varphi,\varepsilon)=\frac{\varepsilon}{\varphi^3}\cdot \log k$.
            \end{enumerate}
        \end{theorem}

        \begin{proof}
            \textbf{Space and runtime.} In the preprocessing phase, as line $1$ of \textsc{FindCenters} (Alg. \ref{alg:gluch-find-centers}), it invokes \textsc{InitOracle}$(G,k,\xi,M)$ one time to get a matrix $\Psi$, which takes %
            $S_\textup{init}$ bits of space according to~\Cref{thrm:bounded-space-dot-product-oracle}. Then it samples $s=\frac{\varphi^2}{\varepsilon}k^4\log k$ vertices and tests all the possible $k$-partitions of the sample set. For each partition, it invokes Alg. \ref{alg:gluch-compute-ordered-partition} one time. Each run of Alg. \ref{alg:gluch-compute-ordered-partition} invokes Alg. \ref{alg:gluch-outer-conductance} $k\log k$ times. Each run of Alg. \ref{alg:gluch-outer-conductance} invokes Alg. \ref{alg:gluch-is-inside} $(s_1+s_2)$ times. Each run of Alg. \ref{alg:gluch-is-inside} computes $C_{\Pi\widehat{\mu},0.93}^{\textup{apx}}$ about $k^{O(1)}$ times, where $C_{\Pi\widehat{\mu},0.93}^{\textup{apx}}=\{x\in V,\frac{\langle\vf_x,\Pi\widehat{\mu} \rangle_\textup{apx}}{\lVert \Pi\widehat{\mu}\rVert^2_\textup{apx}}\ge 0.93\}$. According to ~\Cref{coro:subspace-1} and ~\Cref{coro:subspace-2}, computing $\frac{\langle\vf_x,\Pi\widehat{\mu} \rangle_\textup{apx}}{\lVert \Pi\widehat{\mu}\rVert^2_\textup{apx}}$ takes %
            $s^4\cdot (\frac{k\varphi}{\varepsilon})^{O(1)}\cdot n^{O(\varepsilon/\varphi^2)}\cdot M_\textup{query}\cdot \log^3 n$ bits of space, where we set $\xi =10^{-6}\cdot \frac{\sqrt{\varepsilon}}{\varphi}$. Therefore, Alg. \ref{alg:gluch-find-centers} uses $S_\textup{init}+k\log k\cdot (s_1+s_2)\cdot s^4\cdot (\frac{k\varphi}{\varepsilon})^{O(1)}\cdot n^{O(\varepsilon/\varphi^2)}\cdot M_\textup{query}\cdot \log^3 n$ bits of space. By setting $s_1\coloneqq\Theta\left(\frac{\varphi^2}{\varepsilon}k^5\log^2k\log (1/\eta)\right)$, $s_2\coloneqq\Theta\left(\frac{\varphi^4}{\varepsilon^2}k^5\log^2k\log (1/\eta)\right)$, $\eta= O(\log n)$ and $M_\textup{query}=M$, we get that Alg. \ref{alg:gluch-find-centers} %
            uses $(\frac{k\varphi}{\varepsilon})^{O(1)}\cdot n^{O(\varepsilon/\varphi^2)}\cdot M\cdot \poly(\log n)$ bits of space to get a matrix $\Psi$ and a collection of vertex sets $C$ that represents the cluster centers.
            
            In the query phase, \textsc{HyperplanePartitioning} (Alg. \ref{alg:gluch-hyperplane}) computes $C_{\Pi\widehat{\mu},0.93}^{\textup{apx}}$ about $k^{O(1)}$ times, where $C_{\Pi\widehat{\mu},0.93}^{\textup{apx}}=\{x\in V,\frac{\langle\vf_x,\Pi\widehat{\mu} \rangle_\textup{apx}}{\lVert \Pi\widehat{\mu}\rVert^2_\textup{apx}}\ge 0.93\}$. According to ~\Cref{coro:subspace-1} and ~\Cref{coro:subspace-2}, computing $\frac{\langle\vf_x,\Pi\widehat{\mu} \rangle_\textup{apx}}{\lVert \Pi\widehat{\mu}\rVert^2_\textup{apx}}$ takes 
            $s^4\cdot (\frac{k\varphi}{\varepsilon})^{O(1)}\cdot n^{O(\varepsilon/\varphi^2)}\cdot M\cdot \log^3 n$ bits of space and $s^4\cdot (\frac{k}{\varepsilon})^{O(1)}\cdot n^{1+O(\varepsilon/\varphi^2)}\cdot \frac{1}{M}\cdot \log^3 n\cdot \frac{1}{\varphi^2}$ time, where we set $\xi =10^{-6}\cdot \frac{\sqrt{\varepsilon}}{\varphi}$. By setting $s=\frac{\varphi^2}{\varepsilon}k^4\log k$, we get that Alg. \ref{alg:gluch-hyperplane} takes $(\frac{k\varphi}{\varepsilon})^{O(1)}\cdot n^{O(\varepsilon/\varphi^2)}\cdot M\cdot \poly(\log n)$ bits of space and $(\frac{k\varphi}{\varepsilon})^{O(1)}\cdot n^{1+O(\varepsilon/\varphi^2)}\cdot \frac{1}{M}\cdot \poly(\log n)$ time.

            Thus, the clustering oracle constructs a data structure $\mathcal{D}$ (including matrix $\Psi$, cluster centers $C$ and other information used by the query phase) using $(\frac{k\varphi}{\varepsilon})^{O(1)}\cdot n^{O(\varepsilon/\varphi^2)}\cdot M\cdot \poly(\log n)$ bits of space. Using $\mathcal{D}$, any \textsc{WhichCluster} query can be answered by Alg. \ref{alg:gluch-hyperplane} in $(\frac{k\varphi}{\varepsilon})^{O(1)}\cdot n^{1+O(\varepsilon/\varphi^2)}\cdot \frac{1}{M}\cdot \poly(\log n)$ time.

            \textbf{Correctness.} We highlight that the sublinear spectral clustering oracle is not our contribution. Note that the only modification we make to the clustering oracle is to replace the dot product oracle used in the original work \citep{gluch2021spectral} with our improved oracle. Since the correctness guarantees (i.e., conductance gap and misclassification error) of the clustering oracle rely on the properties of the dot product oracle, and our dot product oracle satisfies the same correctness guarantees with the previous one, the correctness of the overall clustering oracle follows directly from the correctness of the clustering oracle in \citet{gluch2021spectral}.
            
        \end{proof}

\section{Sublinear clustering oracle related to~\Cref{itm:main-case2} in~\Cref{thrm:main}}
    \label{appendix:section-poly(k)-oracle}
    In this section, we present the sublinear spectral clustering oracle with a $\poly(k)$ gap between inner and outer conductance, originally proposed in \citet{shen2024sublinear}, and adapt it by incorporating our dot product oracle, which operates with very little memory.

    Algorithm \ref{alg:shen-construct} first initializes our dot product oracle to get a matrix $\Psi$ (see line $5$). It then leverages our dot product oracle to estimate $\langle \vf_x,\vf_y\rangle$ for all pairs of vertices $x,y$ in the sample set $S$, which are subsequently used to construct a similarity graph $H$ (see lines $6\sim 9$).
    
    \begin{algorithm}[H]%
        \DontPrintSemicolon
        \caption{\textsc{ConstructOracle}($G, k, \varphi, \varepsilon, \gamma,M$)}
        \label{alg:shen-construct}
    
        Let $\xi=\frac{\sqrt{\gamma} }{1000}$ and let $s=\frac{10\cdot k\log k}{\gamma}$\;
        Let $\theta = 0.96(1-\frac{4\sqrt{\varepsilon}}{\varphi})\frac{\gamma k}{n}-\frac{\sqrt{k}}{n}(\frac{\varepsilon}{\varphi^2})^{1/6}-\frac{\xi}{n}$\;
        Sample a set S of $s$ vertices independently and uniformly at random from $V$\;
        Generate a similarity graph $H=(S, \emptyset)$\;
        Let $\Psi=$ \textsc{InitOracle}($G, k,\xi, M$)\;
        
        \For{any $u, v\in S$}{
            Let $\langle \vf_u,\vf_v\rangle_{\rm apx}=$ \textsc{QueryDot}($G, u, v, \xi, \Psi, M$)\;
            \If{$\langle \vf_u,\vf_v\rangle_{\rm apx}\ge \theta$}{
                Add an edge $(u, v)$ to the similarity graph $H$\;
            }
        }
        \eIf{$H$ has exactly $k$ connected components}{
            Label the connected components with $1,2,\dots,k$ (we write them as $S_1,\dots,S_k$)\;
            Label $x\in S$ with $i$ if $x\in S_i$\;
            Return $H$ and the vertex labeling $\ell$
        }{
            return $\textbf{fail}$
        }
    \end{algorithm}

    \begin{algorithm}[h]
        \DontPrintSemicolon
        \caption{\textsc{SearchIndex}($H,\ell,x,M$)}
        \label{alg:shen-search}
        
        \For{any vertex $u\in S$}{
            Let $\langle \vf_u,\vf_x\rangle_{\rm apx}=$ \textsc{QueryDot}($G, u, x, \xi, \Psi, M$)\;
        }
        \eIf{there exists a unique index $1\le i\le k$ such that $\langle \vf_u,\vf_x\rangle_{\rm apx}\ge \theta$ for all $u\in S_i$}{
            return index $i$
        }
        {
            return $\textbf{outlier}$
        }
    \end{algorithm}

    Algorithm \ref{alg:shen-which} corresponds to the query phase of the sublinear spectral clustering oracle, where it answers any \textsc{WhichCluster} query using matrix $\Psi$ and similarity graph $H$.

    \begin{algorithm}[H]
        \DontPrintSemicolon
        \caption{\textsc{WhichCluster}($G,x,M$)}
        \label{alg:shen-which}
     
        \If{preprocessing phase $\textbf{fails}$}{
            return $\textbf{fail}$
        }
        \eIf{\textsc{SearchIndex}($H,\ell,x,M$) return $\textbf{outlier}$}{
            return a random index$\in[k]$
        }
        {
            return \textsc{SearchIndex}($H,\ell,x,M$)
        }
    \end{algorithm}

\section{Deferred proofs of ~\Cref{thrm:distinguish}}\label{appendix:sec-distinguish}

    \begin{lemma}[Cheeger's inequality]
        \label{lemma:cheeger}
        In holds for any graph $G$ that
        \[
            \frac{\lambda_2}{2}\le \phi(G)\le \sqrt{2\lambda_2}
            .
        \]
    \end{lemma}

    Lemma~\ref{lemma:expander-M^t1_x} bounds the $\ell_2$-norm of the $t$-step random walk distribution starting from any vertex $x$ in a $d$-regular graph, distinguishing between the case where the graph is a single $\varphi$-expander and the case where it consists of two disjoint $\varphi$-expanders.
    
    \begin{lemma}[Expander related version of~\Cref{lemma:gluch_0}]
        \label{lemma:expander-M^t1_x}
        Let $\varphi\in (0,1)$. Let $G$ be a $d$-regular graph. Let $\mM$ be the random walk transition matrix of $G$. For any $t\ge \frac{20\log n}{\varphi^2}$ and any $x\in V$, 

        \begin{enumerate}[label=\arabic*]
            \item \label{itm:mt-case1} if $G$ is a $\varphi$-expander of size $n$, then $\lVert \mM^t\mathds{1}_x \rVert_2\le \sqrt{\frac{2}{n}}$,
            \item \label{itm:mt-case2} if $G$ is the disjoint union of two identical $\varphi$-expanders of size $n/2$, then $\lVert \mM^t\mathds{1}_x \rVert_2\le \sqrt{\frac{3}{n}}$.
        \end{enumerate}

        \begin{proof}
        \textbf{\Cref{itm:mt-case1}.} 
            Let $\mL$ be the normalized Laplacian matrix of $G$. Recall that we use $0= \lambda_1\le\dots\le\lambda_n\le 2$ to denote the eigenvalues of $\mL$ and we use $\vu_1,\dots,\vu_n$ to denote the corresponding eigenvectors, where $\vu_1,\dots,\vu_n$ form an orthonormal basis of $\mathbb{R}^n$ and $\vu_1(x)=\frac{1}{\sqrt{n}}$ for any $x\in V$. Note that $\mM=\mI-\frac{\mL}{2}$. Hence, the eigenvalues of $\mM$ are given by $1= 1-\frac{\lambda_1}{2}\ge \dots\ge 1-\frac{\lambda_n}{2}\ge 0$, and the corresponding eigenvectors are still $\vu_1,\dots,\vu_n$. For convenience, we relabel the eigenvalues of $\mM$ as $1=v_1(\mM)=(1-\frac{\lambda_1}{2})\ge v_2(\mM)=(1-\frac{\lambda_2}{2})\ge \dots\ge v_n(\mM)=(1-\frac{\lambda_n}{2})\ge 0$. Moreover, we can write that $\mathds{1}_x=\sum_{i=1}^n{\alpha_i\vu_i}$. Note that $\vu_j^T\mathds{1}_x=\sum_{i=1}^n{\alpha_i\vu_j^T\vu_i}=\alpha_j$. Therefore, $\alpha_j$ corresponds to $\vu_j^T\mathds{1}_x=\vu_j(x)$. Now, we have

            \[
                \mM^t\mathds{1}_x=\mM^t\sum_{i=1}^n{\alpha_i\vu_i}=\sum_{i=1}^n{\alpha_i\mM^t\vu_i}=\sum_{i=1}^n{\alpha_i\left(v_i(\mM)\right)^t\vu_i}.
            \]

            Thus, we have
            \begin{align*}
                \lVert \mM^t\mathds{1}_x\rVert_2^2&=(\mM^t\mathds{1}_x)^T(\mM^t\mathds{1}_x)=\sum_{i=1}^n{\alpha_i^2\left(v_i(\mM)\right)^{2t}}\\
                &=\alpha_1^2\left(v_1(\mM)\right)^{2t}+\sum_{i=2}^n{\alpha_i^2\left(v_i(\mM)\right)^{2t}}\\
                &\le \frac{1}{n}+\left(v_2(\mM)\right)^{2t}\cdot \sum_{i=2}^n{\alpha_i^2}\\
                &\le \frac{1}{n}+\left(v_2(\mM)\right)^{2t}\cdot (n-1).
            \end{align*}

            Since $G$ is a $\varphi$-expander, according to Cheeger's inequality (\Cref{lemma:cheeger}), we get that $\lambda_2\ge \frac{\varphi^2}{2}$. Therefore, for any $t\ge \frac{20\log n}{\varphi^2}$, we have
            \[
                v_2(\mM)^{2t}=\left(1-\frac{\lambda_2}{2}\right)^{2t}\le \left(1-\frac{\varphi^2}{4}\right)^{\frac{4}{\varphi^2}\cdot 10\log n}\le \frac{1}{n^{10}}.
            \]

            Combine above results together, we get that 
            \[
                \lVert \mM^t\mathds{1}_x\rVert_2^2\le \frac{1}{n}+\frac{1}{n^{10}}\cdot (n-1)=\frac{1}{n}+\frac{1}{n^9}\le \frac{2}{n}.
            \]

            \textbf{\Cref{itm:mt-case2}.} We use $C_1,C_2$ to denote the two $\varphi$-expanders in $G$. Since $C_1$ and $C_2$ are disconnected, the normalized Laplacian matrix $\mL$ of $G$ can be written in block-diagonal form as 
            \[
                \mL = \begin{pmatrix}
                \mL_{C_1} & 0 \\
                0   & \mL_{C_2}
                \end{pmatrix},
            \]
            where $\mL_{C_1}\in \mathbb{R}^{\frac{n}{2}\times \frac{n}{2}}$ and $\mL_{C_2}\in \mathbb{R}^{\frac{n}{2}\times \frac{n}{2}}$ are the normalized Laplacian matrix of $C_1$ and $C_2$, respectively. For $\mL_{C_i}$, we use $0=\lambda_1^{C_i}\le \dots\le \lambda_{n/2}^{C_i}\le 2$ to denote the eigenvalues of $\mL_{C_i}$ and we use $\vu_1^{C_i},\dots,\vu_{n/2}^{C_i}\in\mathbb{R}^{\frac{n}{2}\times \frac{n}{2}}$ to denote the corresponding eigenvectors, where $\vu_1^{C_i},\dots,\vu_{n/2}^{C_i}$ from an orthonormal basis of $\mathbb{R}^{\frac{n}{2}\times \frac{n}{2}}$ and $\vu_1^{C_i}(x)=\sqrt{\frac{2}{n}}$ for any $x\in V$. Therefore, the eigenvalues of $\mL$ are given by $0=\lambda_1\le \dots\le \lambda_{n/2}\le 2$, each of which has multiplicity two, where $\lambda_i=\lambda_i^{C_1}=\lambda_i^{C_2}$. For $\lambda_i$, we use $\vu_{2i-1},\vu_{2i}\in\mathbb{R}^n$ to denote the corresponding eigenvectors, where $\vu_{2i-1}=((\vu_i^{C_1})^T,0,\dots,0)^T$ and $\vu_{2i}=(0,\dots,0,(\vu_i^{C_2})^T)^T$. Note that $\mM=\mI-\frac{\mL}{2}$. Hence, the eigenvalues of $\mM$ are given by $1= 1-\frac{\lambda_1}{2}\ge \dots\ge 1-\frac{\lambda_{n/2}}{2}\ge 0$, each of which has multiplicity two, and the corresponding eigenvectors are still $\vu_1,\dots,\vu_n$. For convenience, we relabel the eigenvalues of $\mM$ as $1=v_1(\mM)=v_2(\mM)=(1-\frac{\lambda_1}{2})\ge v_3(\mM)=v_4(\mM)=(1-\frac{\lambda_2}{2})\ge \dots\ge v_{n-1}(\mM)=v_n(\mM)=(1-\frac{\lambda_{n/2}}{2})\ge 0$. %

            Similar to the proof of item $1$, we get
            \begin{align*}
                \lVert \mM^t\mathds{1}_x\rVert_2^2&=(\mM^t\mathds{1}_x)^T(\mM^t\mathds{1}_x)=\sum_{i=1}^n{\alpha_i^2\left(v_i(\mM)\right)^{2t}}\\
                &=\alpha_1^2+\alpha_2^2+\sum_{i=3}^n{\alpha_i^2\left(v_i(\mM)\right)^{2t}}\\
                &\le \frac{2}{n}+\left(v_3(\mM)\right)^{2t}\cdot \sum_{i=3}^n{\alpha_i^2}\\
                &\le \frac{2}{n}+\left(v_3(\mM)\right)^{2t}\cdot (n-2).
            \end{align*}

            Since $C_1$ and $C_2$ both are $\varphi$-expander, according to Cheeger's inequality (\Cref{lemma:cheeger}), we get that $\lambda_2^{C_1}=\lambda_2^{C_2}\ge \frac{\varphi^2}{2}$. Therefore, for any $t\ge \frac{20\log n}{\varphi^2}$, we have
            \[
                (v_3(\mM))^{2t}=\left(1-\frac{\lambda_2}{2}\right)^{2t}=\left(1-\frac{\lambda_2^{C_1}}{2}\right)^{2t}\le \left(1-\frac{\varphi^2}{4}\right)^{\frac{4}{\varphi^2}\cdot 10\log n}\le \frac{1}{n^{10}}.
            \]

            Combine above results together, we get that 
            \[
                \lVert \mM^t\mathds{1}_x\rVert_2^2\le \frac{2}{n}+\frac{1}{n^{10}}\cdot (n-2)=\frac{2}{n}+\frac{1}{n^9}\le \frac{3}{n}.
            \]
        \end{proof}
    \end{lemma}

    The following lemma shows that, under appropriate parameters, Alg. \ref{algo:bounded-space-esti-rw-dot} can estimate the dot product of the random walk distributions from any two vertices up to $\sigma_\textup{err}$, whether the graph is a single $\varphi$-expander or consists of two disjoint $\varphi$-expanders.
    
    \begin{lemma}[Expander related version of~\Cref{lemma:bounded-space-z}]
        \label{lemma:expander-bounded-space-z}

        Let $\varphi\in(0,1)$. Let $G=(V,E)$ be either a $d$-regular $\varphi$-expander with size $n$ or the disjoint union of two identical $d$-regular $\varphi$-expander of size $n/2$. Let $\mM$ be the random walk transition matrix of $G$. Let $Z$ be the output of \textsc{EstRWDot}$(G,R,t,M,x,y)$ (Alg. \ref{algo:bounded-space-esti-rw-dot}). Let $\sigma_\textup{err}>0$. Let $c>1$ be a large enough constant. For any $t\ge \frac{20\log n}{\varphi^2}$ and any $x,y\in V$, if $R\ge  \frac{c\cdot n^{-1}}{\sigma_\textup{err}^2M}$ and $1\le M\le O(n^{1/2})$, then with probability at least $0.99$, we have
        \[
            |Z-\langle \mM^t\mathds{1}_x,\mM^t\mathds{1}_y\rangle|\le \sigma_\textup{err}.
        \]

        Moreover, \textsc{EstRWDot}$(G,R,t,M,x,y)$ runs in $O(Rt)$ time and uses $O(M\cdot \log n)$ bits of space.
        
        \begin{proof}
            \textbf{Runtime and space}. See the proof of~\Cref{lemma:bounded-space-z}.
            
            \textbf{Correctness.} 
            
            By~\Cref{lemma:bounded-space-z-E-V} and~\Cref{lemma:expander-M^t1_x}, we can get that 
            
            \begin{align*}
                \Var[Z]&\le \frac{1}{R}\left[\frac{1}{M}\lVert \mM^t\mathds{1}_x\rVert_2\cdot \lVert \mM^t\mathds{1}_y\rVert_2+\left(\lVert \mM^t\mathds{1}_x\rVert_2 \cdot\lVert \mM^t\mathds{1}_y\rVert_2^2+\lVert \mM^t\mathds{1}_x\rVert_2^2 \cdot\lVert \mM^t\mathds{1}_y\rVert_2\right)\right]\\
                &=\frac{1}{R}\left(\frac{O(n^{-1})}{M}+O(n^{-3/2})\right).\\
            \end{align*}

            Using Chebyshev's inequality, we have
        
            \begin{align*}
                \Pr[|Z-\langle \mM^t\mathds{1}_x,\mM^t\mathds{1}_y\rangle|\ge \sigma_\textup{err}]&=\Pr[|Z-\E[Z]|\ge \sigma_\textup{err}]\\
                &\le \frac{\Var[Z]}{\sigma_\textup{err}^2}\\
                &\le \frac{1}{\sigma_\textup{err}^2}\cdot \frac{1}{R}\left(\frac{O(n^{-1})}{M}+O(n^{-3/2})\right)\\
                &\le \frac{1}{\sigma_\textup{err}^2}\cdot \frac{1}{R}\cdot O\left(\frac{n^{-1}}{M}\right) &&\text{$M\le O\left(n^{1/2}\right)$}\\
                &\le \frac{1}{100}.
            \end{align*}
            The last inequality holds by our choice of $R$ as follows, where $c$ is a large enough constant that cancels the constant hidden in $O\left(\frac{n^{-1}}{M}\right)$:
            \[
                R\ge \frac{c\cdot n^{-1}}{\sigma_\textup{err}^2M}.
            \]

        \end{proof}
    \end{lemma}

    \Cref{lemma:expander-bounded-space-G} asserts that, under suitable parameters, the output $\mathcal{G}$ of \textsc{EstColliProb} (Alg.~\ref{algo:bounded-space-esti-colli-prob}) approximates $(\mM^t \mS)^T (\mM^t \mS)$ in spectral norm, where the latter is the Gram matrix of the random walk distributions from sampled vertices, and this holds whether the graph is a single $\varphi$-expander or two disjoint $\varphi$-expanders.
    
    \begin{lemma}[Expander related version of~\Cref{lemma:bounded-space-G}, restatement of~\Cref{lemma:expander-bounded-space-G}]
        Let $\varphi\in(0,1)$. Let $G=(V,E)$ be either a $d$-regular $\varphi$-expander with size $n$ or the disjoint union of two identical $d$-regular $\varphi$-expander of size $n/2$. Let $\mM$ be the random walk transition matrix of $G$. Let $I_S=\{s_1,\dots,s_s\}$ be a multiset of $s$ indices chosen from $\{1,\dots,n\}$. Let $\mS\in\mathbb{R}^{n\times s}$ be the matrix whose $i$-th column equals $\mathds{1}_{s_i}$. Let $\mathcal{G}\in \mathbb{R}^{s\times s}$ be the output of \textsc{EstColliProb} $(G,R,t,M,I_S)$ (Alg. \ref{algo:bounded-space-esti-colli-prob}). Let $\sigma_\textup{err}>0$. Let $c>1$ be a large enough constant. For any $t\ge \frac{20\log n}{\varphi^2}$, if $R\ge \frac{c\cdot n^{-1}}{\sigma_\textup{err}^2M}$ and $1\le M\le O\left(n^{1/2}\right)$, then weith probability $1-n^{-100}$, we have
        \[
            \lVert \mathcal{G}-(\mM^t\mS)^T(\mM^t\mS)\rVert_2\le s\cdot\sigma_\textup{err}.
        \]

        Moreover, \textsc{EstColliProb} $(G,R,t,M,I_S)$ runs in $O(Rt\cdot \log n\cdot s^2)$ time and uses $O(M\cdot \log^2 n\cdot s^2)$ bits of space.

        \begin{proof}
            Note that we have established~\Cref{lemma:expander-bounded-space-z}, which is an analogue of~\Cref{lemma:bounded-space-z} for graph that is either a $\varphi$-expander of size $n$ or the disjoint union of two identical $\varphi$-expanders of size $n/2$. Since the proof of~\Cref{lemma:bounded-space-G} relies only on~\Cref{lemma:bounded-space-z}, the same augment immediately yields~\Cref{lemma:expander-bounded-space-G}, the corresponding analogue of~\Cref{lemma:bounded-space-G}. 
        \end{proof}
    \end{lemma}

    \Cref{lemma:expander-v_k(M^tS)} demonstrates that $(\mM^t \mS)(\mM^t \mS)^T$ has a clear spectral gap between the $1$-cluster and $2$-cluster cases.
    
    \begin{lemma}[Expander related version of~\Cref{lemma:gluch_2}]
        \label{lemma:expander-v_k(M^tS)}
        Let $\varphi\in (0,1)$. Let $G$ be a $d$-regular graph. Let $\mM$ be the random walk transition matrix of $G$. Let $I_S=\{s_1,\dots,s_s\}$ be a multiset of $s$ indices chosen independently and uniformly at random form $V=\{1,\dots,n\}$. Let $\mS\in\mathbb{R}^{n\times s}$ be the matrix whose $i$-th column equals $\mathds{1}_{s_i}$. For any $t\ge \frac{20\log n}{\varphi^2}$, with probability at least $1-n^{-100}$, we have
        
        \begin{enumerate}[label=\arabic*]
            \item \label{itm:v2-mts-case1} if $G$ is a $\varphi$-expander of size $n$ and $s\ge 1$, then $v_{2}\left(\frac{n}{s}\cdot (\mM^t\mS)(\mM^t\mS)^T\right)\le n^{-9}$,

            \item \label{itm:v2-mts-case2} if $G$ is the disjoint union of two identical $\varphi$-expanders of size $n/2$ and $s\ge c\cdot \log n$, where $c>1$ is a large enough constant, then $v_{2}\left(\frac{n}{s}\cdot (\mM^t\mS)(\mM^t\mS)^T\right)\ge 0.99$.
        \end{enumerate}
    \end{lemma}

    To prove~\Cref{lemma:expander-v_k(M^tS)}, we need the following lemma.

    \begin{lemma}[Lemma 21 in \citet{gluch2021spectral}]
        \label{lemma:gluch-21}
        Let $\mA\in \mathbb{R}^{n\times n}$ be a matrix. Let $b=\max_{\ell \in\{1,\dots,n\}}{\lVert (\mA\mathds{1}_\ell)(\mA\mathds{1}_\ell)^T\rVert_2}$. Let $0<\xi<1$. Let $s\ge \frac{40n^2 b^2\log n}{\xi^2}$. Let $I_S=\{s_1,\dots,s_s\}$ be a multiset of $s$ indices chosen independently and uniformly at random form $V=\{1,\dots,n\}$. Let $\mS\in\mathbb{R}^{n\times s}$ be the matrix whose $i$-th column equals $\mathds{1}_{s_i}$. Then we have
        \[
            \Pr\left[\lVert \mA\mA^T-\frac{n}{s}(\mA\mS)(\mA\mS)^T\rVert_2\ge \xi\right]\le n^{-100}.
        \]
    \end{lemma}
    
    \begin{proof}[Proof of~\Cref{lemma:expander-v_k(M^tS)}]
        \textbf{\Cref{itm:v2-mts-case1}.} 
        The proof follows directly from the proof of item $2$ of Lemma 28 in \citet{gluch2021spectral}.

        \textbf{\Cref{itm:v2-mts-case2}.} Let $A=(\mM^t)(\mM^t)^T=\mM^{2t}$, we get $v_2(A)=v_2(\mM)^{2t}$. Since $G$ is the disjoint union of two identical $\varphi$-expanders, $G$ has two connected components. Therefore, the normalized Laplacian matrix $\mL$ of $G$ has two smallest eigenvalues equal to $0$. Consequently, since $\mM=I-\frac{\mL}{2}$, the two largest eigenvalues of $\mM$ are $1-\frac{0}{2}=1$. Thus, $v_2(A)=1$. 
        
        Let $\widetilde{A}=\frac{n}{s}\cdot (\mM^t\mS)(\mM^t\mS)^T$. By~\Cref{itm:mt-case2} in~\Cref{lemma:expander-M^t1_x}, we have $b=\lVert (\mM^t\mathds{1}_x)(\mM^t\mathds{1}_x)^T\rVert_2\le \lVert \mM^t\mathds{1}_x\rVert_2^2\le \frac{3}{n}$. Let $\xi=\frac{1}{100}$. Therefore, for a large enough constant $c>1$, we have $s=c\cdot \log n\ge \frac{40n^2b^2\log n}{(\frac{1}{100})^2}$. Thus, according to~\Cref{lemma:gluch-21}, we get that with probability at least $1-n^{-100}$,
        \[
            \lVert A-\widetilde{A}\rVert_2\le \frac{1}{100}.
        \]

        By Weyl's inequality (\Cref{lemma:weyl}), we get that $v_2(\widetilde{A})\ge v_2(A)-\lVert\widetilde{A}\rVert_2\ge 1-\frac{1}{100}=0.99$.
        \end{proof}

    The proof of~\Cref{lemma:expander-bounded-space-W} follows directly from the proof of Lemma 24 in \citet{gluch2021spectral}. Nevertheless, for the sake of completeness, we provide a concise proof here. 
    
    \begin{lemma}[Expander related version of~\Cref{lemma:bounded-space-W}, restatement of~\Cref{lemma:expander-bounded-space-W}]
        Let $\varphi\in(0,1)$. Let $G=(V,E)$ be a $d$-regular graph. Let $I_S=\{s_1,\dots,s_s\}$ be a multiset of $s$ indices chosen independently and uniformly at random form $V=\{1,\dots,n\}$. Let $\mathcal{G}\in \mathbb{R}^{s\times s}$ be the output of \textsc{EstColliProb} $(G,R,t,M,I_S)$ (Alg. \ref{algo:bounded-space-esti-colli-prob}). Let $c_1>1$ be a large enough constant. For any $t\ge \frac{20\log n}{\varphi^2}$, if $R\ge \frac{c_1\cdot n}{M}$ and $1\le M\le O\left(n^{1/2}\right)$, then with probability at least $1-2\cdot n^{-100}$, 
        
        \begin{enumerate}[label=\arabic*]
            \item \label{itm:v2G-case1} if $G$ is a $\varphi$-expander of size $n$ and $s\ge 1$, then $v_2\left(\left(\frac{n}{s}\mathcal{G}\right)^2\right)=\left(v_2(\frac{n}{s}\mathcal{G})\right)^2<0.001$,

            \item \label{itm:v2G-case2} if $G$ is the disjoint union of two identical $\varphi$-expanders of size $n/2$ and $s\ge c_2\cdot \log n$, where $c_2>1$ is a large enough constant, then $v_2\left(\left(\frac{n}{s}\mathcal{G}\right)^2\right)=\left(v_2(\frac{n}{s}\mathcal{G})\right)^2>0.95$.
        \end{enumerate}
    \end{lemma}

    \begin{proof}
        Let $\mM$ be the random walk transition matrix of $G$. Let $\mS\in\mathbb{R}^{n\times s}$ be the matrix whose $i$-th column equals $\mathds{1}_{s_i}$. Let $\sqrt{\frac{n}{s}}\cdot \mM^t\mS=\widetilde{U}\widetilde{\Sigma}\widetilde{W}^T$ be an SVD of $\sqrt{\frac{n}{s}}\cdot \mM^t\mS$ where $\widetilde{U}\in\mathbb{R}^{n\times n},\widetilde{\Sigma}\in\mathbb{R}^{n\times n},\widetilde{W}\in\mathbb{R}^{s\times n}$. Let $\frac{n}{s}\cdot\mathcal{G}=\widehat{W}\widehat{\Sigma}\widehat{W}^T$ be an eigendecomposition of $\frac{n}{s}\cdot \mathcal{G}$. 
        
        \textbf{\Cref{itm:v2G-case1}.} Let $\sigma_\textup{err}=\frac{0.0001}{n}$. Let $c$ be the constant from~\Cref{lemma:expander-bounded-space-G}. By the assumption of the lemma, we have
        \[
            R=\frac{c_1\cdot n}{M}\ge \frac{c\cdot 10^8\cdot n}{M}= \frac{c\cdot n^{-1}}{\sigma_\textup{err}^2M}.
        \]
        Thus we can apply~\Cref{lemma:expander-bounded-space-G}. Hence, with probability at least $1-n^{-100}$, we have
        \[
            \lVert \mathcal{G}-(\mM^t\mS)^T(\mM^t\mS)\rVert_2\le s\cdot\sigma_\textup{err}.
        \]

        Let $\widetilde{A}=\frac{n}{s}\cdot (\mM^t\mS)^T(\mM^t\mS)=\widetilde{W}\widetilde{\Sigma}^2\widetilde{W}^T$ and $\widehat{A}=\frac{n}{s}\cdot \mathcal{G}$. Thus, we have $\widetilde{A}^2=\left(\frac{n}{s}\cdot (\mM^t\mS)^T(\mM^t\mS)\right)^2=\widetilde{W}\widetilde{\Sigma}^4\widetilde{W}^T$ and $\widehat{A}^2=\left(\frac{n}{s}\cdot \mathcal{G}\right)^2=\widehat{W}\widehat{\Sigma}^2\widehat{W}^T$. Moreover, we have $\lVert \widetilde{A}^2-\widehat{A}^2 \rVert_2=\left(\frac{n}{s}\right)^2\lVert \left((\mM^t\mS)^T(\mM^t\mS)\right)^2-\mathcal{G}^2 \rVert_2$. Using the triangle inequality and sub-multiplicativity of spectral norm and the above $\lVert \mathcal{G}-(\mM^tS)^T(\mM^tS)\rVert_2\le s\cdot\sigma_\textup{err}$ bound, we can get that 
        \[
            \lVert \left((\mM^t\mS)^T(\mM^t\mS)\right)^2-\mathcal{G}^2 \rVert_2\le (s\cdot \sigma_\textup{err})^2+2\cdot s\cdot \sigma_\textup{err}\lVert(\mM^t\mS)^T(\mM^t\mS) \rVert_2.
        \]

        Note that $\lVert(\mM^t\mS)^T(\mM^t\mS) \rVert_2\le \lVert(\mM^t\mS)^T(\mM^t\mS) \rVert_F=\sqrt{\sum_{i=1}^s{\sum_{j=1}^s{((\mM^t\mathds{1}_{s_i})^T(\mM^t\mathds{1}_{s_j}))^2}}}$, by Cauchy Schwarz inequality and~\Cref{itm:mt-case1} of~\Cref{lemma:expander-M^t1_x}, we can get that $\lVert(\mM^t\mS)^T(\mM^t\mS) \rVert_2\le s\cdot \frac{2}{n}$.
        Put them together and by the choice of $\sigma_\textup{err}=\frac{0.0001}{n}$, we have that 
        \[
            \lVert \widetilde{A}^2-\widehat{A}^2 \rVert_2\le \left(\frac{n}{s}\right)^2\cdot \left(s^2\sigma_\textup{err}^2+2\cdot s\cdot \sigma_\textup{err}\cdot s\cdot \frac{2}{n}\right)=n^2\sigma_\textup{err}^2+4n\sigma_\textup{err}\le 0.00005.
        \]

        Moreover, since $s\ge 1$, by~\Cref{itm:v2-mts-case1} of~\Cref{lemma:expander-v_k(M^tS)}, with probability at least $1-n^{-100}$, we have
        \[
            v_{2}\left(\widetilde{A}^2\right)=v_{2}\left(\left(\frac{n}{s}\cdot (\mM^t\mS)^T(\mM^t\mS)\right)^2\right)\le (n^{-9})^2=n^{-18}.
        \]

        By Weyl's inequality, we have that 
        \[
            v_2(\widehat{A}^2)\le v_2(\widetilde{A}^2)+\lVert \widetilde{A}^2-\widehat{A}^2 \rVert_2\le n^{-18}+0.0005\le 0.001.
        \]

        \textbf{\Cref{itm:v2G-case2}.} By the same augment of the proof of~\Cref{itm:v2G-case1} and~\Cref{itm:v2-mts-case2} of~\Cref{lemma:expander-M^t1_x}, we can get that $\lVert(\mM^t\mS)^T(\mM^t\mS) \rVert_2\le s\cdot \frac{3}{n}$. Thus, by the choice of $\sigma_\textup{err}=\frac{0.0001}{n}$, we have that
        \[
            \lVert \widetilde{A}^2-\widehat{A}^2 \rVert_2\le \left(\frac{n}{s}\right)^2\cdot \left(s^2\sigma_\textup{err}^2+2\cdot s\cdot \sigma_\textup{err}\cdot s\cdot \frac{3}{n}\right)=n^2\sigma_\textup{err}^2+6n\sigma_\textup{err}\le 0.0007.
        \]

        Moreover, since $s\ge c_2\cdot \log n$, by~\Cref{itm:v2-mts-case2} of~\Cref{lemma:expander-v_k(M^tS)}, with probability at least $1-n^{-100}$, we have
        \[
            v_{2}\left(\widetilde{A}^2\right)=v_{2}\left(\left(\frac{n}{s}\cdot (\mM^t\mS)^T(\mM^t\mS)\right)^2\right)\ge (0.99)^2>0.98.
        \]

        By Weyl's inequality, we have that 
        \[
            v_2(\widehat{A}^2)\ge v_2(\widetilde{A}^2)-\lVert \widetilde{A}^2-\widehat{A}^2 \rVert_2\ge 0.98-0.0007>0.95.
        \]
    \end{proof}

\section{Deferred proofs of~\Cref{thrm:distinguish_lower}}
    \label{appendix:section-lower-bound}

\subsection{Hard Instance I}
\label{subsec:hard-instance-1}

Before we start the proof, we would first introduce some basic definitions in
information theory.

\subsubsection{Basic definitions}

    \begin{defn} [Entropy]
        \upshape
        Given a random variable $X$ taking values in the set $\mathcal{X}$ and distributed according to $p:\mathcal{X}\rightarrow[0,1]$, the \emph{entropy} of $X$ is defined as
        \[
            H(X)\coloneqq -\sum_{x \in \mathcal{X}} p(x) \log p(x).
        \]

        In the special case where $X$ has only two possible outcoms, the entropy is given by
        \[
            H_2(X)\coloneqq -p\log p-(1-p)\log (1-p).
        \]

    \end{defn}
    
    The entropy of a random variable quantifies the average level of uncertainty or information associated with the random variable. Note that for the special case of $H_2$, we have the following property:

    \begin{lemma}
        $$1-H_2\left(\frac{1}{2}+a\right)=\frac{1}{2 \ln 2} \sum_{l=1}^{\infty} \frac{(2 a)^{2 l}}{l(2 l-1)}=O\left(a^2\right) .$$
    \end{lemma}

    Given the outcome of another random variable $Y$, we can also quantify this randomness using conditional entropy.

    \begin{defn} [Conditional entropy]
        \upshape
        Given random variables $X$ and $Y$ taking values in sets $\mathcal{X}$ and $\mathcal{Y}$, respectively, with joint distribution $p:\mathcal{X}\times \mathcal{Y}\rightarrow[0,1]$, the \emph{conditional entropy} of $X$ given $Y$ is defined as
        \[
            H(X \mid Y)=H(X,Y) - H(Y) =-\sum_{x \in \mathcal{X}, y \in \mathcal{Y}} p(x, y) \log \frac{p(x, y)}{p(y)}.
        \]

    \end{defn}
    
    Furthermore, the amount of information that is shared between two random variables is called mutual information. 

    \begin{defn}[Mutual Information]
        \upshape
        \label{defn:mutual-information}
        Given random variables $X$ and $Y$ taking values in $\mathcal{X}$ and $\mathcal{Y}$, respectively, the \emph{mutual information} between $X$ and $Y$ is defined as
        \[
        I(X; Y) = H(X) - H(X \mid Y) = H(Y) - H(Y \mid X).
        \]
        
        Similarly, given a random variable $Z$ taking values in $\mathcal{Z}$, the \emph{conditional mutual information} of $X$ and $Y$ given $Z$ is defined as
        \[
        I(X; Y \mid Z) = H(X \mid Z) - H(X \mid Y, Z).
        \]
    \end{defn}

    Our proof will also use the following key properties of mutual information. 
    
    \begin{lemma}[Data Processing Inequality]
        \label{lemma:data-processing-ineq}
        Given random variables $X,Y$ and $Z$ taking values in sets $\mathcal{X},\mathcal{Y}$ and $\mathcal{Z}$, respectively, such that $X \perp Z \mid Y$. Then
        \[
            I(X ; Z) \leq I(X ; Y).
        \]
    \end{lemma}
    
    \begin{lemma}[Chain Rule]
        \label{lemma:chain-rule}
        Given random variables $X,Y$ and $Z$ taking values in sets $\mathcal{X},\mathcal{Y}$ and $\mathcal{Z}$, respectively, we have 
        \[
        I(X ; Y, Z)=I(X ; Z)+I(X ; Y \mid Z). 
        \]
    \end{lemma}
    
\subsubsection{The proof}

    Now we prove~\Cref{thrm:hard-instance-1}.
    
    \begin{theorem}[Restatement of~\Cref{thrm:hard-instance-1}]
        Let $\mathcal{A}$ be an algorithm that detects the Hard Instance I with error at most $1 / 3$. The algorithm can access the samples in a single-pass streaming fashion using $M$ bits of space and $T$ samples. Furthermore, at each step, the algorithm may choose which set to sample by specifying $W_t$. We then have $T \cdot M=\Omega\left(n\right)$.
    \end{theorem}
    
    \begin{proof}[Proof of~\Cref{thrm:hard-instance-1}]
            
            In either case, we can think of the output of $p$ as being a pair $(C, V)$, where $C$ is an element of $[n]$ is chosen uniformly, and $V \in\{0,1\}$ is a fair coin if $X=0$ and has bias $ Y_CZ_t$ if $X=1$.
            
            Let $s_1, \ldots, s_T$ be the observed samples from $p$. Let $M_t$ denote the bits stored in the memory after the algorithm sees the $t$-th sample $s_t$.
            
            Since the algorithm $\mathcal{A}$ learns $X$ with probability at least $2 / 3$ after viewing $T$ samples, we know that $I\left(X ; M_T\right)>\Omega(1)$. On the other hand, $M_t$ is computed from $\left(M_{t-1}, s_t\right)$ without using any information about $X$. More formally, $X \perp M_t \mid\left(M_{t-1}, s_t\right)$ and therefore we can use the data processing inequality (\Cref{lemma:data-processing-ineq}) and chain rule (\Cref{lemma:chain-rule}) to get:
            
            $$
            I\left(X ; M_t\right) \leq I\left(X ; M_{t-1}, s_t\right)=I\left(X ; M_{t-1}\right)+I\left(X ; s_t \mid M_{t-1}\right).
            $$
            
            Since irrespective of $X, C$ is uniform over the pairs of bins, we note that $C$ is independent of $X$ even when conditioned on the memory $M$. Moreover, player's choice of $W_t$ is computed only from $M_{t-1}$. Thus,
            
            $$
            I\left(X ; s_t \mid M_{t-1}\right)=I\left(X ; C_t V_t \mid M_{t-1}\right)=I\left(X ; V_t \mid M_{t-1} C_t\right) = I\left(X ; V_t \mid M_{t-1} C_tW_t\right).
            $$

            Let $\alpha_{t-1}=\operatorname{Pr}\left[X=1 \mid M_{t-1} C_t W_t\right]$ and thus $\operatorname{Pr}\left[X=0 \mid M_{t-1} C_t W_t\right]=1-\alpha_{t-1}$.
            
            We have that 
            
            $$
            \begin{aligned}
            \operatorname{Pr}\left[V_t=0 \mid X=0, M_{t-1}, C_t, W_t\right] & =\frac{1}{2}, \\
            \operatorname{Pr}\left[V_t=0 \mid X=1, M_{t-1}, C_t, Z_t\right] & =\frac{1+ \E\left[Z_t Y_{C_t} \mid M_{t-1},W_t\right]}{2},\\
            \operatorname{Pr}\left[V_t=0 \mid M_{t-1}, C_t\right] & =\left(1-\alpha_{t-1}\right) \frac{1}{2}+\alpha_{t-1} \frac{1+ \E\left[Z_t Y_{C_t} \mid M_{t-1},W_t\right]}{2} \\&=\frac{1}{2}+\frac{\alpha_{t-1} \E\left[Z_t Y_{C_t} \mid M_{t-1},W_t\right]}{2} .
            \end{aligned}
            $$
            
            We can calculate
            
            $$
            \begin{aligned}
            I & \left(X ; V_t \mid M_{t-1} C_t W_t\right)=H\left(V_t \mid M_{t-1} C_t W_t\right)-H\left(V_t \mid M_{t-1} C_t W_t X\right) \\
            = & H_2\left(\Pr\left[V_t=0 \mid M_{t-1}, C_t, W_t\right]\right)\\
            &-\left\{\Pr\left[X=1 \mid M_{t-1} C_t W_t\right] H_2\left(\Pr\left[V_t=0 \mid X=1, M_{t-1}, C_t, W_t\right]\right)\right. \\
            & \left.+\Pr\left[X=0 \mid M_{t-1} C_t W_t\right] H_2\left(\Pr\left[V_t=0 \mid X=0, M_{t-1}, C_t, W_t\right]\right)\right\} \\
            = & H_2\left(\frac{1}{2}+\frac{\alpha_{t-1}  \E\left[Z_t Y_{C_t} \mid M_{t-1},W_t\right]}{2}\right)-\alpha_{t-1} H_2\left(\frac{1}{2}+\frac{ \E\left[Z_tY_{C_t} \mid M_{t-1},W_t\right]}{2}\right)-\left(1-\alpha_{t-1}\right) H_2\left(\frac{1}{2}\right) \\
            = & \alpha_{t-1}\left[1-H_2\left(\frac{1}{2}+\frac{ \E\left[Z_t Y_{C_t} \mid M_{t-1},W_t\right]}{2}\right)\right]-\left[1-H_2\left(\frac{1}{2}+\frac{\alpha_{t-1}  \E\left[Z_t Y_{C_t} \mid M_{t-1}, W_t\right]}{2}\right)\right] \\
            = & \Theta(1) \left[ \alpha_{t-1} \left(\frac{ \E\left[Z_t Y_{C_t} \mid M_{t-1},W_t\right]}{2}\right)^2 - \left(\frac{\alpha_{t-1} \E\left[Z_t Y_{C_t} \mid M_{t-1},W_t\right]}{2}\right)^2 \right] \\
            = & \Theta(1) \alpha_{t-1}\left(1-\alpha_{t-1}\right) \E\left[Z_t Y_{C_t} \mid M_{t-1},W_t\right]^2 \\
            \leq & O(1) \E\left[Z_t Y_{C_t} \mid M_{t-1},W_t\right]^2.
            \end{aligned}
            $$
            Since $C_t$ is uniformly random, we have that

            $$
            \begin{aligned}
            &I\left(X ; V_t \mid M_{t-1} C_t W_t\right)=\frac{1}{n} \cdot \sum_{j=1}^n O(1) \E\left[Z_t Y_j \mid M_{t-1}, W_t\right]^2.
            \end{aligned}
            $$

            Now to bound this part, note that we first have $H\left(M_{t-1}, W_t\right) \leq M$ that $I\left(Z_t Y_1 \ldots Z_t Y_n ; M_{t-1}, W_t\right) \leq M$. At the same time, notice that $Z_t$ is just flipping the value of $Y_1,\ldots,Y_n$ and thus $H\left(Z_t Y_1 \ldots Z_t Y_n\right) = H\left(Y_1 \ldots Y_n\right) = n$. Thus we have  
            
            $$H\left(Z_t Y_1 \ldots Z_t Y_n \mid M_{t-1}, W_t\right)=H\left(Z_t Y_1 \ldots Z_t Y_n\right)-I\left(Z_t Y_1 \ldots Z_t Y_n ; M_{t-1}, W_t\right) \geq n-M.$$
            
            On the other hand, we have that 
    
    $$
    \sum_{i=1}^n H\left(Z_t Y_i \mid M_{t-1}, W_t\right) \geq H\left(Z_t Y_1 \ldots Z_t Y_n \mid M_{t-1}, W_t\right) \geq n-M.
    $$

    Thus,
    
    $$
    M \geq \sum_{i=1}^n\left[1-H\left(Z_t Y_i \mid M_{t-1},W_t\right)\right]=\Theta\left(\sum_{i=1}^n \E\left[Z_t Y_i \mid M_{t-1}, W_t\right]^2\right),
    $$
    
    where the equality comes from the fact that if $\operatorname{Pr}\left[Z_t Y_i=1 \mid M_{t-1}, W_t\right]=\frac{1}{2}+\beta$, then 
    
    \begin{align*}
    \E\left[Z_t Y_i \mid M_{t-1}, W_t\right]&=\Pr\left[Z_t Y_i=1 \mid M_{t-1}, W_t\right](+1)+\Pr\left[Z_t Y_i=-1 \mid M_{t-1}, W_t\right](-1)\\&=\left(\frac{1}{2}+\beta\right)-\left(\frac{1}{2}-\beta\right)=2 \beta.
    \end{align*}
            
    We finally have that 
    
    $$
    \begin{aligned}
    \Omega(1) \leq I\left(M_T ; X\right) & =\sum_{t=0}^{T-1} I\left(M_{t+1} ; X\right)-I\left(M_t ; X\right) \\
    & =\sum_{t=0}^{T-1} I\left(M_t, S_{t+1} ; X\right)-I\left(M_t ; X\right) \\
    & =\sum_{t=0}^{T-1} I\left(S_{t+1} ; X \mid M_t\right) \\
    & =\sum_{t=0}^{T-1} I\left(V_{t+1} ; X \mid M_t, C_{t+1}, W_{t+1}\right) \\
    & = O(1) \frac{T\cdot M}{n} .
    \end{aligned}
    $$
    
    We conclude that $T\cdot M \geq \Omega(n)$.
    \end{proof}

\subsection{Hard Instance II}
\label{subsec:hard-instance-2}

To prove~\Cref{lemma:mixing time}, we first introduce the definition of mixing time.
    \begin{defn}[Mixing time]
        \label{defn:mixing time}
        \upshape
        Let $G=(V,E)$ be a $d$-regular graph on $n$ vertices. Let $\mM$ be the lazy random walk transition matrix of $G$. Let $\vm_t = \mM^t\vm_0$, where $\vm_0$ is a distribution over $[n]$. Let $\pi=(\frac{1}{n},\dots,\frac{1}{n})^T$ be the stationary distribution of $G$. Then the \emph{mixing time} $\tau_\varepsilon(\mM)$ is defined to be the smallest $t$ such that for any $\vm_0$, $d_\textup{TV}(\vm_x,\pi)\le \varepsilon$.
    \end{defn}
    
    \begin{proof}[Proof of~\Cref{lemma:mixing time}]
        Note that $\pi=(\frac{1}{n},\dots,\frac{1}{n})^T\in\mathbb{R}^n$ is the stationary distribution of $G$. According to spectral graph theory, we have $\tau_\varepsilon(\mM)=O(\frac{1}{\phi(G)^2})\log (\frac{n}{\varepsilon})$. Let $\varepsilon=\frac{0.01}{n^2}$. Note that $G$ is a $\varphi$-expander, we have that $\phi(G)=\varphi$ (see~\Cref{defn:conductance}). Therefore, according to the definition of mixing time, we get that for $t=\tau_\varepsilon(\mM)=O(\frac{1}{\varphi^2}\log (\frac{n}{\frac{0.01}{n^2}}))=O(\frac{\log n}{\varphi^2})$, we have that $d_\textup{TV}(\mM^t\mathds{1_x},\pi)\le \frac{0.01}{n^2}$.
    \end{proof}

\section{Experimental details}
    \label{appendix:section-experiment}
    
    \paragraph{Accuracy} Let $C_1,\dots,C_k$ be the ground-truth clustering and let $\widehat{C}_1,\dots,\widehat{C}_k$ be the clusters produced by the oracle, where $\widehat{C}_i=\{x\in V|$\textsc{WhichCluster}$(G,x)=i\}$. The accuracy is defined as $\frac{1}{n}\cdot \max_{\pi}\sum_{i=1}^k{|C_i\cap \widehat{C}_{\pi(i)}|}$, where $\pi:[k]\rightarrow [k]$ is a permutation. 

    \paragraph{Implementation details}

    In our experiments, we implemented three main components: (i) the new dot product oracle proposed in this paper (Alg. \ref{algo:bounded-space-initialize-oracle} and Alg. \ref{algo:bounded-space-query-dot}), (ii) the original dot product oracle in \citet{gluch2021spectral}, and (iii) the spectral clustering oracle relies on a $\poly(k)$ conductance gap itself. The clustering oracle relies on accurate dot product estimates to function correctly; hence, we first needed to identify parameters that ensure reliable dot product estimation performance. These parameters include (i) $s_\textup{dot}$, the number of sampled vertices in dot product oracle, (ii) $t$, the random walk length and (iii) $l$, the number of repetitions in the median trick, and a set of space--time-related parameters.

\begin{figure}[H]
    \centering
    \begin{subfigure}{0.49\textwidth}
        \centering
        \includegraphics[width=\linewidth]{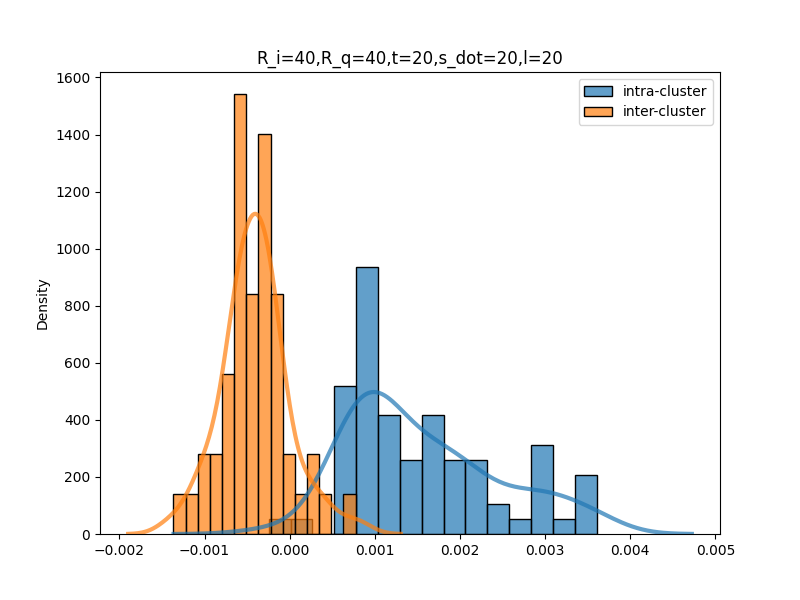}
        \caption{unsuitable parameter values: $R_\textup{init}=R_\textup{query}=40$}
        \label{fig:fig1}
    \end{subfigure}
    \hfill
    \begin{subfigure}{0.49\textwidth}
        \centering
        \includegraphics[width=\linewidth]{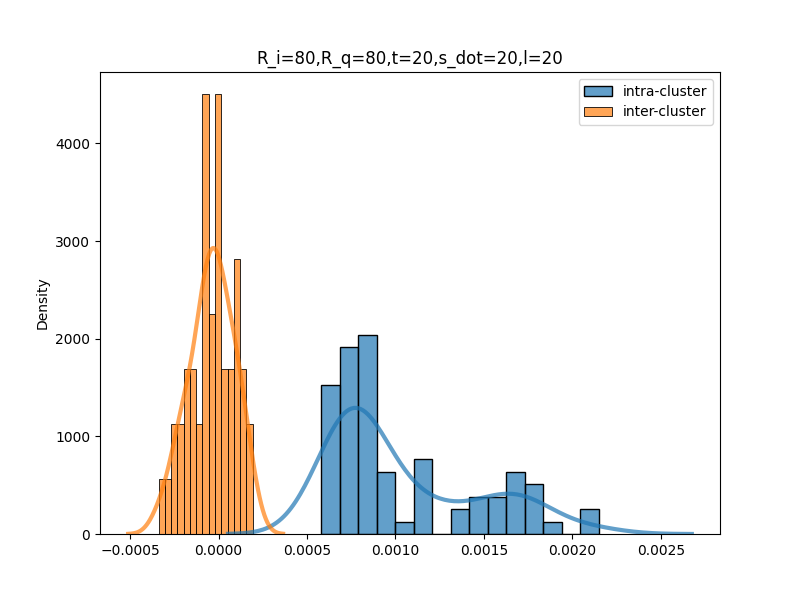}
        \caption{suitable parameter values: $R_\textup{init}=R_\textup{query}=80$}
        \label{fig:fig2}
    \end{subfigure}
    \caption{Effect of parameter settings on the original dot product oracle. (a): an unsuitable configuration where the estimated spectral dot products for intra-cluster and inter-cluster pairs overlap. (b): a suitable configuration where a clear gap emerges between the two distributions.}
    \label{fig:two_figs}
\end{figure}

    For the original dot product oracle in \citet{gluch2021spectral}, $R_\textup{init},R_\textup{query}$ are the space--time-related parameters. We set $R_\textup{init}$ and $R_\textup{query}$ according to the theoretical guarantee, which states that the oracle works when $R_\textup{init}=R_\textup{query}=O(\sqrt{n})$. Following the implementation details in \citet{shen2024sublinear}, we explored multiple parameter configurations for $s_\textup{dot},t,l,R_\textup{init}=R_\textup{query}$. For each configuration, we initialized the dot product oracle with the corresponding parameters, sampled a subset of vertex pairs, computed their estimated spectral dot products, and plotted the density graphs (see~\Cref{fig:two_figs}). The presence of a clear gap (see~\Cref{fig:fig2}) in the density graph was used as the criterion for selecting suitable parameter values. In fact, for a graph with parameters $n=3000$, $k=3$, $p=0.07$, and $q=0.002$, we found that $s_\textup{dot}=20$, $t=20$, $l=20$, and $R_\textup{init}=R_\textup{query}\ge 80$ provided reliable estimates. And we make $80\times 80$ a concrete instantiation of $O(\sqrt{n})\times O(\sqrt{n})=O(n)$.
    
    For the new dot product oracle, we set $s_\textup{dot}=20,t=20$ and $l=20$ like above. The space--time-related parameters $M_\textup{init}=M_\textup{query}$ serve as inputs, corresponding to $R_\textup{init}^\textup{our}=R_\textup{query}^\textup{our}=\frac{80\times 80}{M_\textup{init}}=\frac{6400}{M_\textup{init}}$ (see line $2$ of Alg. \ref{algo:bounded-space-initialize-oracle} and Alg. \ref{algo:bounded-space-query-dot}). In our experiments, we varied $M_\textup{init}=M_\textup{query}$ in the range $[30,80]$.
    
    Finally, for the clustering oracle itself, we determined the number of sampled vertices $s$ (see line $3$ of Alg. \ref{alg:shen-construct}) through extensive testing of multiple candidate values, and selected $s=21$ for all experiments. Additionally, we set a threshold $\theta$ (see line $8$ of Alg. \ref{alg:shen-construct}) to construct similarity graph; based on the density plots of estimated dot products (see~\Cref{fig:fig2}), we chose $\theta\approx 0.0005$.

\end{document}